\documentclass[12pt]{article}
\pdfoutput=1
\usepackage{jheppub}
\usepackage[utf8]{inputenc}
\usepackage{verbatim}
\usepackage{amsmath}
\usepackage{amssymb}
\usepackage{amsthm}
\usepackage{slashed}
\usepackage{amsfonts}
\usepackage{subfigure}
\usepackage{float}
\usepackage{physics}

\def\ea{\end{align}}
\def\be{\begin{equation}}
\def\ee{\end{equation}}
\def\bea{\begin{eqnarray}}
\def\eea{\end{eqnarray}}

\newtheorem{theorem}{Theorem}

\newtheorem{algorithm}[theorem]{Algorithm}

\def\fft#1#2{{\frac{#1}{#2}}}

\begin{document}
%\subheader{empty}

\hbox{}

\title{On Quantum Simulation Of Cosmic Inflation}
\author[a]{Junyu Liu}
\author[b]{Yue-Zhou Li}

\affiliation[a]{Walter Burke Institute for Theoretical Physics and Institute for Quantum Information and Matter, California Institute of Technology, Pasadena, CA 91125, USA}
\affiliation[b]{Physics Department, McGill University 3600 University Street, Montreal, QC, H3A 2T8 Canada}
%\emailAdd{jliu2@caltech.edu}

\abstract{In this paper, we generalize Jordan-Lee-Preskill, an algorithm for simulating flat-space quantum field theories, to 3+1 dimensional inflationary spacetime. The generalized algorithm contains the encoding treatment, the initial state preparation, the inflation process, and the quantum measurement of cosmological observables at late time. The algorithm is helpful for obtaining predictions of cosmic non-Gaussianities, serving as useful benchmark problems for quantum devices, and checking assumptions made about interacting vacuum in the inflationary perturbation theory.

Components of our work also include a detailed discussion about the lattice regularization of the cosmic perturbation theory, a detailed discussion about the in-in formalism, a discussion about encoding using the HKLL-type formula that might apply for both dS and AdS spacetimes, a discussion about bounding curvature perturbations, a description of the three-party Trotter simulation algorithm for time-dependent Hamiltonians, a ground state projection algorithm for simulating gapless theories, a discussion about the quantum-extended Church-Turing Thesis, and a discussion about simulating cosmic reheating in quantum devices.}
\maketitle

\newpage
\section{Introduction}
Cosmic inflation is probably the most famous paradigm describing the physics of the very early universe immediately after the big bang. Based on solid predictions from general relativity and quantum field theory in the curved spacetime, the theory of cosmic inflation successfully describes a homogenous and isotropic universe measured by the Cosmic Microwave Background (CMB) and Large Scale Structure (LSS). In the standard theory of inflation, the exponential expansion of the universe is driven classically by a massless scalar field, so-called inflaton, leading to a map of thermal CMB photons, while cosmic perturbations above classical trajectories of the inflaton are predicted by a usually weakly-coupled quantum field theory in the nearly de Sitter spacetime (see references about cosmic inflation \cite{Guth:1980zm,Linde:1981mu,Albrecht:1982wi,Starobinsky:1980te,Sato:1980yn,Fang:1980wi,Mukhanov:1981xt,Starobinsky:1982ee,Guth:1982ec,Bardeen:1983qw,Mukhanov:1990me,Maldacena:2002vr,Linde:2005ht,Chen:2006nt,Cheung:2007st,Weinberg:2008hq,Malik:2008im,Chen:2009zp,Baumann:2009ds,Chen:2010xka,Wang:2013zva,Baumann:2014nda}). Moreover, the complete description of the history of the universe after inflation, include reheating \cite{Abbott:1982hn,Dolgov:1982th,Albrecht:1982mp,Traschen:1990sw,Kofman:1994rk,Shtanov:1994ce}, electroweak phase transition and baryogenesis \cite{Sakharov:1967dj,Kuzmin:1985mm,Shaposhnikov:1987tw,Cohen:1993nk}, etc., might be related to the consequences of inflation. A similar exponential growth paradigm is also happening in the late histories of the universe: the dark energy \cite{Riess:1998cb,Carroll:2000fy,Peebles:2002gy}.

From its born, the theory of cosmic inflation with its corresponding cosmic perturbation has obtained great achievement in predicting cosmological observations. However, there are still a large number of open problems remaining for our generations, making it yet an active research direction both in theory and in observations. Until now, inflation is not fully confirmed by observations (for instance, lacking a relatively large tensor-to-scalar ratio makes it relatively hard to get distinguished from other theories). Even if the theory is correct, we still don't know what the precise dynamics of inflaton and other possible particles in the very early universe is. A more accurate, quantitative picture after inflation is not fully understood. Theoretical discussions about a full quantum gravitational description of de Sitter space are still working in progress, with the help of technologies from string theory, holography and AdS/CFT \cite{Witten:2001kn,Strominger:2001pn,Maloney:2002rr,Alishahiha:2004md,Arkani-Hamed:2015bza,Dong:2018cuv,Lewkowycz:2019xse,Chen:2020tes,Hartman:2020khs}. A full resolution of a possible cosmic singularity and a comprehensive analysis of de Sitter quantum gravity might rely on more precise, hardcore progress about string theory and Planckian physics, which are still in active development \cite{Kachru:2003aw}.

On the other hand, the rapid development of quantum information science provides a lot of opportunities for the next generation of possible computing technologies, and many of them also belong to fundamental physicists. At present and also in the near future, we have the capability to control 50-100 qubits and use them to perform operations, but the noise of quantum systems limits the scale of quantum circuits \cite{preskill2018quantum,arute2019quantum}. In the long run, we hope to be able to manufacture universal, fault-tolerant quantum computing devices that can perform precise calculations on specific problems and achieve quantum speedup. As high-energy physicists, we hope that future quantum computing devices can be used in the research of high-energy physics. Therefore, designing and optimizing quantum algorithms for a given high-energy physical process has gradually become a hot research field \cite{Jordan:2011ne,Jordan:2011ci,jordan2014quantum,jordan2018bqp,Preskill:2018fag}. At the same time, designing quantum algorithms for basic physical processes also has far-reaching theoretical significance. According to the quantum-extended Church-Turing Thesis, any physical process occurring in the real world can be simulated by a quantum Turing machine. Therefore, studying quantum algorithms for simulating physical processes can help support or disprove the quantum-extended Church-Turing Thesis. In practice, designing and running quantum algorithms from fundamental physics could also help us benchmark our near-term quantum devices.

In this paper, we initialize the study on extending the Jordan-Lee-Preskill algorithm for particle scattering \cite{Jordan:2011ne,Jordan:2011ci} towards inflationary spacetime. The Jordan-Lee-Preskill algorithm is a basic paradigm for simulating quantum field theory processes in a universal quantum computer. The algorithm constitutes encoding, initial state preparation, time evolution, and measurement. The algorithm is designed to run in the polynomial time, scaling with system size, consistent with the quantum-extended Church-Turing Thesis. Here in this work, we discuss the extension of the algorithm from particle scattering and evaluation of scattering matrix in the flat-space quantum field theory, to the computation of correlation functions of scalar fields in the inflationary background, connecting with the observables of CMB and LSS we discuss in cosmology.

In modern cosmology, the cosmic perturbation theory is constructed by a massless, weakly-coupled quantum field theory living in a nearly de-Sitter spacetime, written in the Friedmann-Robertson-Walker (FRW) metric. The free theory describes a Gaussian perturbation above the thermal spectrum of the CMB map, while the small non-Gaussian corrections are encoding the information, for instance, about gravitational nonlinearities and other fundamental fields appearing during inflation, which could in principle be detected by future experiments. Cosmic perturbation theory treated those perturbations in the standard quantum field theory language. Although in this paper, we only discuss scalar and curvature perturbations, other modes (vector and tensor) modes have also been studied significantly, and some of them are closely related to primordial gravitational waves.

So, why do we simulate a weakly-coupled field theory of the inflationary cosmology? In the standard description, the weakly-coupled nature is motivated by observations: we didn't observe large non-Gaussianity of the CMB map. Since correlation functions could be computed by Feynman diagrams in the weakly-coupled quantum field theory, at least at tree level, we should already know the answer. However, we point out the following reasons, showing that it is, in fact, important to simulate such theories.
\begin{itemize}
\item This work could serve as a starting point for understanding the Jordan-Lee-Preskill algorithm in the curved spacetime. Simulating quantum field theories in the curved space is an interesting topic, closely related to many open problems about quantum gravity. However, instead of considering general spacetime, we could treat de Sitter space with fruitful symmetric structures as a warmup example.
\item When computing correlation functions, we are using the in-in formalism in the cosmic perturbation theory. The formalism is useful to determine correlation functions reliably under certain physical assumptions, while further implications of the formalism and systematic understandings of cosmic non-Gaussianities, closely related to the nature of the interacting vacuum of quantum field theories in de Sitter space, is still under active research. We believe that our quantum simulation algorithm will have theoretical significance: it will help us validate the assumptions we made about the in-in formalism, and explore, estimate, and bound rigorous errors away from assumptions. (More details about this will be given in Section \ref{inin}.)
\item Strongly-coupled quantum field theory in the curved spacetime is an interesting topic to study and simulate by itself, despite the lack of phenomenological consequences. For instance, we might be interested if we have different phases and phase transitions in those theories.
\item It also might be interesting to study loop diagrams beyond the tree level diagrams in the weakly-coupled theories. Thus, quantum simulation of those theories could justify and enrich our understandings about cosmic perturbation theories at higher loops \cite{Weinberg:2005vy,Senatore:2009cf,Senatore:2012nq,Pimentel:2012tw}.
\item This work might be regarded as a starting point for studying quantum simulation of quantum field theories in the FRW spacetime, which is interesting for other cosmic phases. For instance, cosmic (p)reheating and cosmic phase transitions in the early universe might be described by strongly-coupled quantum field theories, which require the computational capability of quantum machines. Thus, this work could provide experiences, insights and intuitions for future studies of quantum simulation of other cosmic phases and phase transitions in those theories (in Appendix \ref{reheating}, we discuss a potential simulation problem about cosmic reheating.).
\item Explicit construction of the quantum algorithm simulating inflationary dynamics could support the statement of the quantum-extended Church-Turing Thesis. It is of great interest to study how the quantum-extended Church-Turing Thesis is compatible with general relativity since the definition of complexity requires a definition of time coordinates. Thus, our work about inflationary spacetime could provide such an example.
\item We could also take the advantage that we already know some aspects of inflationary perturbation theory. In fact, if we could really run the algorithm and simulate the dynamics in the quantum device, we could check the answer from some of our theoretical expectations. Thus, quantum simulation of such theories will be helpful for us to benchmark our quantum devices.
\item Explicit construction of such algorithms could also benefit our future study about the nature of quantum gravity. For instance, measuring some similar correlation functions in AdS could be interpreted as holographic scattering experiments, which might be reduced to flat-space scattering amplitudes in the limit of the large AdS radius. Such experiments are important for understanding the nature of holographic theories, conformal bootstrap, and the consistency requirements of holographic CFTs \cite{Gary:2009ae,Heemskerk:2009pn,Giddings:2009gj,Fitzpatrick:2010zm,ElShowk:2011ag,Fitzpatrick:2011hu,Fitzpatrick:2011dm,Maldacena:2015iua}.
\item It is fun.
\end{itemize}
In this paper, we will construct a full quantum simulation algorithm for measuring correlation functions of the following form
\begin{align}
\left\langle\Omega_{\text {in }}\left(t_{0}\right)\left|\mathcal{O}_{H}(t)\right| \Omega_{\text {in }}\left(t_{0}\right)\right\rangle~,
\end{align}
in an interacting scalar quantum field theory living in the four-dimensional inflationary spacetime. The state $\ket{\Omega_{\text{in}} (t_0)}$ is the vacuum state of the full interacting theory, and $\mathcal{O}_{H}(t)$ is the Heisenberg operator made by multiple scalar fields. In this theory, the Hamiltonian is time-dependent, so we denote $t_0$ as the time when inflation starts, while $t$ as the time when inflation ends. The meaning of this correlation function will be precisely given in the later discussion in Section \ref{inin}, and in fact, most experimental observables are given by the above field theory constructions.

The algorithm contains (adiabatic) initial state preparations, time evolution/particle scattering process, and measurement, and is argued to be polynomial, analogous to the Jordan-Lee-Preskill experiment. Moreover, since the theory and the simulation targets are very different from flat-space quantum field theories, it is necessary to include new ingredients appearing in the theory, modifying the Jordan-Lee-Preskill experiment. Here, as an outline of this paper, we will summarize the following most crucial points.
\begin{itemize}
\item Inflation in a lattice. We systematically construct the lattice quantum field theory of curvature perturbation, both for the free theory and the interacting theory in the Heisenberg picture. Furthermore, we compute the Bogoliubov transformation to diagonalize the Hamiltonian and determine the gap of the free system both for the continuum and the lattice field theories. The estimate of the mass gap is necessary for us to perform the adiabatic state preparation.

\item Adiabatic state preparation to find the interacting vacuum. In our algorithm, similar to Jordan-Lee-Preskill, the time evolution is used twice. We firstly use the adiabatic state preparation to find the interacting vacuum at the initial time $t_0$, and then we use the time-dependent Trotter simulation to simulate the time evolution in the Heisenberg picture. The treatment will help us determine the nature of the interacting vacuum in the curved spacetime.

\item Encoding from the causal structure. When we are simulating the Heisenberg evolution operator, we fix the computational basis to be the field basis at the initial time $t_0$. Then during the time evolution, we wish to encode our Hamiltonian for a general time $t$. We define the basis we use in the Heisenberg Hamiltonian for a general time $t$ by the free theory at $t_0$. Thus, we need to express the field operator at $t$ as a linear combination of field operators at $t_0$. The integration kernel is related to the Green's identity in the operator context. When we are using the conformal time, the past light cone is the same as the flat space, which is supporting the integration kernel, since the spacetime is conformally flat. This could be regarded as a toy version of the Hamilton-Kabat-Lifschytz-Lowe (HKLL) formula in AdS/CFT \cite{Hamilton:2005ju,Harlow:2018fse}.

\item Encoding bounds from the effective field theory (EFT) scale. In this paper, we choose the field basis to do our quantum simulation. In the original algorithm \cite{Jordan:2011ne,Jordan:2011ci}, the number of qubits we need to encode in this basis, is bounded by the scattering energy of the scattering experiments. However, in this paper, we do not naively have a quantity like scattering energy in the flat space. Instead, we use bounds of field (field momentum) fluctuations from an EFT scale, $\Lambda$, which sets the energy scale of cosmic inflation, combining with some other physics in cosmology. The setup of the value $\Lambda$ is from the inflationary physics we study.

\item Exponential expansion and the Trotter errors. When applying the product formula (Lie-Trotter-Suzuki formula) to the inflationary Hamiltonian discretized in a lattice, we notice that there is a dependence of Trotter errors from the scale factor, which increases exponentially measured in the physical time. In the conformal time, the error grows polynomially. The total time of inflation or the e-folding number could set limits on the accuracy of our computation and number of quantum gates.
\end{itemize}
This paper is organized as the following. In Section \ref{pre}, we discuss some inflationary quantum field theories and their lattice versions. In Section \ref{inin}, we discuss the in-in formalism, the basic method for computing correlation functions in the inflationary perturbation theory. In Section \ref{algorithm}, we discuss the details of our algorithm, including encoding, state preparation, time evolution, measurement, and error estimates. In Section \ref{remark}, we make some final remarks about quantum simulation of cosmic inflation, and more generic perspectives about quantum information theory, quantum gravity, and cosmology related to observations. In the appendices, we make some further comments about computation and quantum gravity \ref{MH}, and some comments about quantum simulation of cosmic reheating \ref{RH}.
\newpage

\section{The theory of inflationary perturbations}\label{pre}
In this section, we will review some basics about the theory we will simulate: a simple real massless scalar field moving in the (approximate) de Sitter background written in the FRW coordinate. We will describe the free theory, the lattice version, and interactions. In this section, some of the materials are standard, while some of them are new. In order to make this paper, we feel we have to introduce the full setup and add careful comments about our treatment, and we feel that this material might be helpful for readers. However, for a more detailed study, see some standard review articles \cite{Mukhanov:1990me,Linde:2005ht,Malik:2008im,Baumann:2009ds,Chen:2010xka,Wang:2013zva}.
\subsection{The free theory}
The inflationary spacetime is given by the following metric
\begin{align}
{d}{s^2} = {a^2}(\tau )\left[ { - d{\tau ^2} + d{{\bf{x}}^2}} \right]~.
\end{align}
Here, $\bf{x}$ is the coordinate as the three-vector. $\tau$ is called the conformal time, where the metric is written in the conformally flat form. The scale factor $a(\tau)$ is given by
\begin{align}
a(\tau ) =  - \frac{1}{{H\tau }}~,~~~~~~a(t)= a_0 e^{H(t-t_0)}~,
\end{align}
where $H$ is the Hubble parameter\footnote{We abuse the notation here, where $H$ is the Hubble constant, not the Hamiltonian.}, $a_0$ is the initial scale factor with the initial time $t_0$. The conformal time $\tau$ is related to the physical time $t$ by
\begin{align}
dt = a d\tau~.
\end{align}
At the beginning of inflation, we have $\tau_0 \sim -\infty$, while at the end, we have $\tau_\text{end}\sim 0^-$. Thus, at the beginning of the inflation, which corresponds to the big bang, we have $a(\tau_0) \sim 0$. In the end, we have a scale factor, which is very large. We often use the $e$-folding number
\begin{align}
N = \log \frac{{{a_{{\rm{end}}}}}}{{{a_{{0}}}}}~,
\end{align}
to measure the time of inflation.

Now, we consider a massless scalar field $\sigma$ moving in this background, with the action
\begin{align}
&S = \int {{d^3}} xdt\left( {\frac{{{a^3}}}{2}{{\dot \sigma }^2} - \frac{a}{2}{{\left( {{\partial _i}\sigma } \right)}^2}} \right) = \int {{d^3}} xd\tau \frac{{{a^2}}}{2}\left( {\sigma {'^2} - {{\left( {{\partial _i}\sigma } \right)}^2}} \right)~,
\end{align}
where $\partial_i$ denotes the spatial derivative. Furthermore, we define $A'=\partial_\tau A$ and $\dot A= \partial_t A$ for variable $A$. Since it is in the free theory, it could be exactly solved. The solution of the field operator in the Heisenberg picture reads
\begin{align}
\sigma (\tau ,{\bf{x}}) = \int {\frac{{{d^3}k}}{{{{(2\pi )}^3}}}} \left( {{u_k}{a_{\bf{k}}} + u_k^*a_{ - {\bf{k}}}^\dag } \right){e^{i{\bf{k}} \cdot {\bf{x}}}}~.
\end{align}
Here, $a$ and $a^\dagger$ are creation and annihilation operators
\begin{align}
\left[ {{a_{\bf{k}}},{a_{{\bf{k}}'}}} \right] = 0,\quad \left[ {{a_{\bf{k}}},a_{{\bf{k}}'}^\dag } \right] = {(2\pi )^3}{\delta ^3}\left( {{\bf{k}} - {\bf{k}}'} \right)~,
\end{align}
$u$ satisfies the classical equation of motion and $k \equiv \abs{\bf{k}}$\footnote{We apologize that we abuse the notation such that $\int d^3 k$ means the integration over the three-vector $\bf{k}$. Similar definition works for $x$ and $\bf{x}$. We will also use $k_i$ to denote the $i$-th component $k_i=\mathbf{k} \cdot  \mathbf{r}_i$, where $\mathbf{r}_i$ is the $i$-th unit vector.}. The definition of the mode function $u$, and the definition of the vacuum, have some ambiguities due to the nature of the curved spacetime. A canonical choice is so-called the Bunch-Davies vacuum \cite{Bunch:1978yq},
\begin{align}
{u_k}(\tau ) = \frac{H}{{\sqrt {2{k^3}} }}(1 + ik\tau ){e^{ - ik\tau }}~.
\end{align}
In the curved spacetime, the quantization defined by the creation and annihilation operators are analogous to the flat-space answer: the mode $\bf{k}$ is created from the creation operator $a^\dagger_{\bf{k}}$.

How is this related to the cosmic perturbation theory? In fact, the free field $\sigma$ here we are considering could be understood as the fluctuation of the inflaton field $\delta \phi (t,\bf{x})$, where
\begin{align}
\phi (t,{\bf{x}}) = \bar \phi (t) + \delta \phi (t,{\bf{x}})~.
\end{align}
Here, $\bar{\phi}(t)$ is the homogenous and isotropic classical background that is not related to the coordinate, and $ \delta \phi $ is understood as a quantum fluctuation (an operator) following the rules of quantum field theories. A similar treatment is done for computing thermal radiations of the black hole, where we are also expanding fluctuations around classical curved geometries.

In principle, cosmic perturbation theory requires a full consideration by expanding the metric, together with the inflationary field as well. Then, from the full expansion, one could decompose perturbations into different components based on the spin. In geometry, the curvature perturbation will come together with the inflaton, and there are ambiguities made by redefinitions of coordinates and fields. The redundancy could be fixed by the gauge choice. A standard and convenient choice is to replace the role of $\delta \phi $ by another variable $\zeta$, which represents the curvature perturbation, and this gauge is called the $\zeta$-gauge. Under the $\zeta$-gauge, we could write down the second-order action\footnote{We use the unit such that we set the Planck mass $M_{\text{pl}}=1$.}
\begin{align}
S = \epsilon \int d t{d^3}x\left[ {{a^3}{{\dot \zeta }^2} - a{{({\partial _i}\zeta )}^2}} \right]~,
\end{align}
where $\epsilon$ is so-called the slow-roll parameter
\begin{align}
\epsilon  \equiv  - \frac{{\dot H}}{{{H^2}}} \approx {\epsilon _V} = \frac{1}{2}{\left( {\frac{{\partial_{\bar{\phi}} V}}{V}} \right)^2} \approx \sqrt {\frac{{{{\dot {\bar {\phi}} }^2}}}{{2{H^2}}}}~,\label{eq: slow roll}
\end{align}
and $V=V(\bar{\phi})$ is the inflationary potential\footnote{Roughly speaking, one could interpret $\zeta$ as $\zeta  =  - H\delta \phi /\dot{\bar{\phi}} $.}.

More general single field inflationary models have the following second-order action in the $\zeta$-gauge,
\begin{align}
S = \epsilon \int d t{d^3}x\left[ {{a^3}\frac{1}{{c_s^2}}{{\dot \zeta }^2} - a{{(\partial_i \zeta )}^2}} \right]~,\label{action of inflation}
\end{align}
where $c_s$ is called the sound speed \cite{Chen:2006nt}\footnote{Defining the Lagrangian density
\begin{align}
L \supset \sqrt { - g} P(\phi,X)~,
\end{align}
for the whole inflaton field before perturbation, where $X \equiv -(\partial_\mu \phi \partial^\mu\phi) / 2$, we could define the sound speed by
\begin{align}
c_{s}^{2}=\frac{P_{, X}}{P_{, X}+2 X P_{, X X}}~.
\end{align}
}. In this case, the mode expansion reads
\begin{align}
\zeta (\tau ,{\bf{x}}) = \int {\frac{{{d^3}k}}{{{{(2\pi )}^3}}}} \left( {{v_k}{a_{\bf{k}}} + v_k^*a_{ - {\bf{k}}}^\dag } \right){e^{i{\bf{k}} \cdot {\bf{x}}}}~,\label{mode expansion}
\end{align}
where
\begin{align}
{v_k} = \frac{H}{{\sqrt {4\epsilon {c_s}{k^3}} }}\left( {1 + ik{c_s}\tau } \right){e^{ - ik{c_s}\tau }}~.
\end{align}
This action is exactly the free theory piece of the whole action we wish to simulate as the simplest example.
\subsection{The lattice free theory}\label{latticesec}
Now, motivated by the goal of quantum simulation, we define our theory in the lattice. The lattice version in this paper means that we are discretizing the spatial direction, replacing
\begin{align}
{\Omega _3} = b\mathbb{Z}_{\hat L}^3~,
\end{align}
where $b$ is the lattice spacing, $L$ is the total length of a single direction, and
\begin{align}
\hat{L}=\frac{L}{b}~.
\end{align}
The total number of sites is then given by
\begin{align}
\mathcal{V}=\hat{L}^3~.
\end{align}
The notation $\mathbb{Z}_{\hat{L}}^3$ denotes the three-dimensional integer periodic lattice with the periodicity $\hat{L}$ in each direction. One could define a lattice version of the scalar quantum field theory based on the action of the curvature perturbation. From the Lagrangian density in the continuum,
\begin{align}
{\mathcal{L}_\tau } = \epsilon {a^2}\left( {\frac{1}{{c_s^2}}\zeta {'^2} - {{\left( {{\partial _i}\zeta } \right)}^2}} \right)~,
\end{align}
we could write down a discrete version of the Lagrangian
\begin{align}
{L_\tau } = \epsilon {a^2}{b^3}\sum\limits_{{\bf{x}} \in {\Omega _3}} {\left[ {\frac{1}{{c_s^2}}\zeta '{{({\bf{x}})}^2} - {{\left( {{\nabla _i}\zeta ({\bf{x}})} \right)}^2}} \right]}~,
\end{align}
where
\begin{align}
{\nabla _i}\zeta ({\bf{x}}) \equiv \frac{{\zeta \left( {{\bf{x}} + b{{\bf{r}}_i}} \right) - \zeta ({\bf{x}})}}{b}{{\bf{r}}_i}~,
\end{align}
is defined as the discrete version of the spatial derivative, and $\mathbf{r}_i$ is the unit vector along the spatial direction $i$. Note that the definition of the Lagrangian depends on the time coordinate, and here we are using the conformal time $\tau$.

Now we could perform the Legendre transform to define the Hamiltonian. The field momentum is
\begin{align}
{\pi _\zeta } = \frac{{\partial {\cal L}}}{{\partial  \zeta '}} = \frac{{2{a^2}\epsilon }}{{c_s^2}} \zeta'~.
\end{align}
Thus, the Hamiltonian density in the continuum is
\begin{align}
{\cal H}_{\tau} = {\pi _\zeta } \zeta'  - {\cal{L}_{\tau}} = \frac{{c_s^2}}{{4{a^2}\epsilon }}\pi _\zeta ^2 + \epsilon a^2{\left( {{\partial _i}\zeta } \right)^2}~.
\end{align}
Thus, the lattice version is
\begin{align}
{H_\tau } = {b^3}\sum\limits_{{\bf{x}} \in {\Omega _3}} {\left[ {\frac{{c_s^2}}{{4{a^2}\epsilon }}\pi _\zeta ^2({\bf{x}}) + \epsilon a^2{{\left( {{\nabla _i}\zeta ({\bf{x}})} \right)}^2}} \right]}~.
\end{align}
Now we could write down the rule of canonical quantization
\begin{align}
&\left[ {\zeta (\tau ,{\bf{x}}),{\pi _\zeta }(\tau ,{\bf{y}})} \right] = i{b^{ - 3}}{\delta _{{\bf{x}},{\bf{y}}}}~,\nonumber\\
&[\zeta (\tau ,{\bf{x}}),\zeta (\tau ,{\bf{y}})] = \left[ {{\pi _\zeta }(\tau ,{\bf{x}}),{\pi _\zeta }(\tau ,{\bf{y}})} \right] = 0~.
\end{align}
Now we make some brief comments on the applicability of the above theory. In our theory, we choose the ultraviolet cutoff $\Lambda=1/b$, corresponds to the spatial lattice spacing $b$. In cosmic inflation, we usually demand,
\begin{align}
H<\sqrt {H{M_{{\rm{pl}}}}\epsilon } = \sqrt {H \epsilon } \sim \sqrt{\dot{\bar{\phi}}} \lesssim \frac{1}{b}=\Lambda \ll {M_{\text{pl}}} = 1~.
\end{align}
The right-hand side is concerning that we are avoiding Planckian physics, and we set the Planck constant to be 1 in our units.

Now, we extend our mode expansions to the lattice version,
\begin{align}
\zeta (\tau ,{\bf{x}}) = \sum\limits_{\mathbf{k} \in \Gamma_3} {\frac{1}{{{L^3}}}} \left( {v({\bf{k}}){a_{\bf{k}}} + {v^*}({\bf{k}})a_{ - {\bf{k}}}^\dag } \right){e^{i{\bf{k}} \cdot {\bf{x}}}}~,
\end{align}
where $\Gamma_3$ is the dual lattice
\begin{align}
\Gamma_3=\frac{2 \pi}{L} \mathbb{Z}_{\hat{L}}^{3}~,
\end{align}
and
\begin{align}
\left[ {{a_{\bf{p}}},{a_{\bf{q}}}} \right] = 0~,~~~~~~\left[ {{a_{\bf{p}}},a_{\bf{q}}^\dag } \right] = {L^3}{\delta _{{\bf{p}},{\bf{q}}}}~.
\end{align}
Note that in the formula, we assume that $v(-\mathbf{k})=v(\mathbf{k})$ to make sure we have a real field. $v(\mathbf{k})$ should satisfy the equation of motion for the corresponding classical discrete field theory. We find
\begin{align}
&v({\bf{k}}) = \frac{H}{{2\sqrt {\epsilon {c_s}\omega {{({\bf{k}})}^3}} }}\left( {1 + i\omega ({\bf{k}}){c_s}\tau } \right){e^{ - i\omega ({\bf{k}}){c_s}\tau }}~,\nonumber\\
&\omega ({\bf{k}}) = \frac{2}{b}{\left( {\sum\limits_{i = 1}^3 {{{\sin }^2}} \left( {\frac{{b{k_i}}}{2}} \right)} \right)^{\frac{1}{2}}}~.
\end{align}
When we compare this result with the continuous mode function, it is easy to realize the following properties of the solution.
\begin{itemize}
\item We indeed have $v(-\mathbf{k})=v(\mathbf{k})$.
\item The form of the solution is intuitively understandable. The form of the dispersion relation is already presented in the case of flat space in the lattice theory \cite{Jordan:2011ne,Jordan:2011ci}. Since the metric is conformally flat, the dispersion relation should stay the same in our coordinate. Using the dispersion relation $\omega (\mathbf{k})$ in the equation of motion for the Fourier space, we will obtain $v(\mathbf{k})$.
\item In the limit where $\tau \to -\infty$, the mode function has reduced to the flat space lattice solution with the sound speed
\begin{align}
v({\bf{k}}) = \frac{{iH\tau \sqrt {{c_s}} }}{{2\sqrt {\epsilon \omega ({\bf{k}})} }}{e^{ - i\omega ({\bf{k}}){c_s}\tau }} \Rightarrow \left| {\sqrt {2\epsilon } av({\bf{k}})} \right| = \frac{{\sqrt {{c_s}} }}{{\sqrt {2\omega ({\bf{k}})} }}{e^{ - i\omega ({\bf{k}}){c_s}\tau }}~.
\end{align}
Note that $\abs{\sqrt{2\epsilon} a v(\mathbf{k})}$ the honest scalar field defined consistently for the flat space \footnote{The coefficient $\sqrt{2\epsilon}$ comes from the relation between the curvature perturbation and the scalar field perturbation.}. This is indeed the mode function of the ``phonons" in the 3+1 dimensions. Thus, we indeed define a lattice version of the Bunch-Davies vacuum.
\item We could also look at the continuum limit $b\to 0$, where the solution is reduced to the continuum answer
\begin{align}
v({\bf{k}}) \Rightarrow {v_k} = \frac{H}{{\sqrt {4\epsilon {c_{\rm{s}}}{k^3}} }}\left( {1 + ik{c_s}\tau } \right){e^{ - ik{c_s}\tau }}~.
\end{align}
Note that the existence of the cubic lattice breaks the rotational symmetry, then in the lattice, $v({\bf{k}})$ is not only a function of $k$. We could also compute the next leading short distance corrections of the dispersion relation,
\begin{align}
\omega ({\bf{k}}) = k - {b^2}\frac{{\sum\limits_i {k_i^4} }}{{24k}} + O({b^4})~,
\end{align}
which will sequentially modify, for instance, the two-point function and the power spectrum. This is, in fact, the lattice effect of the Lorentz symmetry breaking in cosmic inflation. The lattice effect of quantum gravity is studied in a completely different context. They are assuming the lattice theory should be some ultraviolet completions of general relativity and ask for phenomenological consequences on bounding the lattice spacing $b$, for instance, the observational bound on the power spectrum of CMB \cite{Jacobson:2005bg}. But now we are completely under a different situation: we are using the lattice theory to simulate the quantum field theory in the continuum.
\end{itemize}

\subsection{The gap}
When we are talking about the flat space, we know that the massless scalar field theory is gapless if we don't consider the infrared cutoff when defining the continuum field theory. In the lattice version, the gap will scale as $1/L$, which is also the infrared cutoff of the theory. Similarly, a massless field in our inflationary FRW metric is also gapless. Although the Hamiltonian is time-dependent, the Bogoliubov transformation will make the diagonal mode of the Hamiltonian to be still gapless in the whole time during inflation.

A direct way to verify the statement is through a direct computation of the Bogoliubov transformation. To start, we take the free continuum Hamiltonian expressed in $\zeta$ fields and field momenta measured in the $\tau$ coordinates. Replacing the fields and field momenta by their solutions, we obtain the Hamiltonian evaluated in the Heisenberg picture,
\begin{align}
{H_\tau } = \int {\frac{{{d^3}k}}{{{{(2\pi )}^3}}}} \left( {{{\cal A}_k}\left( {2a_{\bf{k}}^\dag {a_{\bf{k}}} + {{(2\pi )}^3}{\delta ^3}(0)} \right) + \left( {{{\cal B}_k}a_{ - {\bf{k}}}^\dag a_{\bf{k}}^\dag  + {\rm{H}}.{\rm{c.}}} \right)} \right)~,
\end{align}
where H.c. means the Hermitian conjugate, and
\begin{align}
&{{\cal A}_k} = \frac{{1 + 2c_s^2{\tau ^2}{k^2}}}{{4{c_s}{\tau ^2}k}}~,\nonumber\\
&{{\cal B}_k} = \frac{{1 - 2i{c_s}\tau k}}{{4{c_s}{\tau ^2}k}}{e^{2i{c_s}\tau k}}~.
\end{align}
The existence of the off-diagonal term appearing in the Hamiltonian reveals the time-dependent nature of the Hamiltonian. To diagonalize the Hamiltonian, we make the following Bogoliubov transformation
\begin{align}
&a_{\bf{k}}^\dag  = {\alpha _{\bf{k}}}b_{\bf{k}}^\dag  + {\beta _{\bf{k}}}{b_{ - {\bf{k}}}}~,\nonumber\\
&{a_{\bf{k}}} = \alpha _{\bf{k}}^*{b_{\bf{k}}} + \beta _{\bf{k}}^*b_{ - {\bf{k}}}^\dag~.
\end{align}
Note that now we keep $a^{\dagger}_{\mathbf{k}}$ and $a_{\mathbf{k}}$ to be static. The coefficient $\alpha$, $\beta$, and the new operators $b^{\dagger}_{\mathbf{k}}$ and $b_{\mathbf{k}}$, are allowed to be time-dependent. We wish to ensure that $b^{\dagger}_{\mathbf{k}}$ and $b_{\mathbf{k}}$ are still canonical
\begin{align}
\left[b_{\mathbf{k}}(\tau), b_{\mathbf{k}^{\prime}}(\tau)\right]=0, \quad\left[b_{\mathbf{k}}(\tau), b_{\mathbf{k}^{\prime}}^{\dagger}(\tau)\right]=(2 \pi)^{3} \delta^{3}\left(\mathbf{k}-\mathbf{k}^{\prime}\right)~.
\end{align}
Furthermore, we wish to cancel the off-diagonal terms. We find that we need to demand
\begin{align}
&{\left| {{\alpha _{\bf{k}}}} \right|^2} = \frac{1}{2}\left( {1 + \frac{{{{\cal A}_k}}}{{\sqrt {{{\left| {{{\cal A}_k}} \right|}^2} - {{\left| {{{\cal B}_k}} \right|}^2}} }}} \right)~,\nonumber\\
&{\beta _{\bf{k}}} = \frac{{\sqrt {{{\left| {{{\cal A}_k}} \right|}^2} - {{\left| {{{\cal B}_k}} \right|}^2}}  - {{\cal A}_k}}}{{{{\cal B}_k}}}\alpha _{\bf{k}}^*~.
\end{align}
Now the Hamiltonian looks like
\begin{align}
{H_\tau } = \int {\frac{{{d^3}k}}{{{{(2\pi )}^3}}}} ({c_s}k)\left( {b_{\bf{k}}^\dag {b_{\bf{k}}} + {{(2\pi )}^3}\frac{1}{2}{\delta ^3}(0)} \right)~.
\end{align}
Ignoring the time-dependence of $b_{\mathbf{k}}$, it seems that the Hamiltonian is formally the same as the flat space. This is because the nature of the above Bogoliubov transformation is to keep re-defining the vacuum state during time evolution evaluated under the definition of $b_{\mathbf{k}}$, to make sure the form of the Hamiltonian is the same as the original time $\tau_0$.

The lattice calculation is very similar to the continuum version, since we only need to replace the continuous momenta to the discrete ones. The Hamiltonian is written after the Bogoliubov transformation as
\begin{align}
{H_\tau } = \sum\limits_{{\bf{k}} \in {\Gamma _3}} {\frac{1}{{{L^3}}}} ({c_s}\omega ({\bf{k}}))\left( {b_{\bf{k}}^\dag {b_{\bf{k}}} + \frac{1}{2}{L^3}} \right)~.
\end{align}
From now, we call the particle created by $b^\dagger$ as the ``diagonal mode" (normal mode), to distinguish from the particle defined by $a^\dagger$, the $\zeta$-particle (the curvature perturbation). The above calculation in the continuum explicitly shows that the Hamiltonian does not have the mass gap at each time slice. Furthermore, from the energy, we cannot distinguish the vacuum and the diagonal modes carrying zero momenta. In the lattice calculation, the gap scales as $1/L$ since our momenta are from the discrete set $\Gamma_3$. Those two facts are consistent when we wish to assign an infrared cutoff $1/L$ to the field theory.

We wish to make the following further comments.
\begin{itemize}
\item Similar calculations about the above Bogoliubov transformation are done in the circumstance of particle production in inflation, and ``cosmic decoherence" by showing similar Bogoliubov transformations, similar to the Unruh effect and black hole information problem, could lead to fruitful particle production. Particles produced by transformations of the vacua are closely related to decoherence where the reduced density matrix becomes nearly diagonal, and quantum perturbations in inflation could become approximately ``classical" when exiting the horizon. See some references \cite{Gibbons:1977mu,Abbott:1982hn,Guth:1985ya,Albrecht:1992kf,Lesgourgues:1996jc,Kiefer:1998qe,Liu:2016aaf}.

\item We are performing the calculation in the conformal time $\tau$. The story might be different when we mix time and space and define some other time coordinates. In fact, the definition of the mass gap, which is related to energy, is naturally associated with the definition of time. In general, one might consider defining some versions of the ADM mass, although it might face some technical troubles when we are considering the de Sitter boundary condition. This question might be interesting for quantum simulation since the value of the gap is naturally associated with the efficiency of the adiabatic state preparation process for finding the vacuum. This shows further problems about the compatibility between quantum algorithms using quantum mechanics (requiring a definition of time) and relativity (admitting ambiguities of defining time and space), and the nature of quantum-extended Church-Turing Thesis in the curved spacetime. It might also be interesting to consider different definitions of the gap in AdS. We leave those discussions for future research.
\end{itemize}

\subsection{Interactions}
Now we wish to introduce interactions beyond the second-order action. Currently, in the standard theory of inflation, we only assume a scalar field $\phi$, but the nature of this field is not very clear. Furthermore, it might be possible that there are some other particles interacting with inflaton, for instance, the standard model particles. Furthermore, there might be contributions from higher-order cosmic perturbation theory, gravity non-linearities, inflaton self-interaction, and so on. Those contributions will, in principle, give contributions beyond the free theory, and one could those contributions order by order in the perturbation theory framework.

In this paper, we could imagine that we are simulating the following specific model, although most discussions in this paper could be easily generalized to much more generic scenarios. We will consider the general single-field inflation framework, and we expand the action towards the third order \cite{Chen:2006nt}. We obtain
\begin{align}
{S_\text{int}} = \int d t{d^3}x{a^3}\left\{ { - 2\frac{\lambda }{{{H^3}}}{{\dot \zeta }^3} + \frac{\Sigma }{{{a^2}{H^3}}}\left( {1 - c_s^2} \right)\dot \zeta {{({\partial _i}\zeta )}^2}} \right\}~,
\end{align}
where $\lambda$ and $\Sigma$ are model parameters\footnote{In fact, we have
\begin{align}
&\Sigma  \equiv X{P_{,X}} + 2{X^2}{P_{,XX}} = \frac{{{H^2}\epsilon }}{{c_s^2}}~,\nonumber\\
&\lambda  \equiv {X^2}{P_{,XX}} + \frac{2}{3}{X^3}{P_{,XXX}}~,
\end{align}
where $P$ is defined in the last footnote. Furthermore, the coupling constants in the cubic interaction are comparable to ${O}(\epsilon^2)$.
}. This interaction gives a general form of cubic interaction, which will generate non-trivial three-point functions, namely, non-Gaussianities.

As we discussed in the introduction, the standard inflation theory is expected to be weakly-coupled. Namely, all interactions we include here should have small coupling constants, admitting perturbative expansions. In terms of quantum simulation, during the adiabatic procedure we will discuss later, we could extend the calculation towards the strongly-coupled regime. Thus, we do not really need to expand the action if we really wish to simulate the theory of general single-field inflation: we could even study the original action without perturbative expansion in principle. However, here we are describing an example where we indeed do an expansion before encoding towards quantum devices, with the motivations described in the introduction section of this paper.

Now we write down the Hamiltonian of the above interaction. We have
\begin{align}
\mathcal{L}_{\tau}=\epsilon a^{2}\left(\frac{{\zeta}'^{2}}{c_{s}^{2}}-\left(\partial_{i} \zeta\right)^{2}\right)+a\left(-\frac{2 \lambda}{H^{3}} {\zeta'}^{3}+\frac{\Sigma}{H^{3}}\left(1-c_{s}^{2}\right) {\zeta'}\left(\partial_{i} \zeta\right)^{2}\right)~.
\end{align}
Up to the leading order perturbation theory, the Legendre transform gives the following total Hamiltonian, split as the addition of the free and the interacting pieces (from now, we ignore the subscript $\tau$ if not necessary)
\begin{align}
&H=H_{0}+H_{I}~,\nonumber\\
&H_{0}=\int d^{3} x\left(\frac{c_{s}^{2} \pi_{\zeta}^{2}}{4 a^{2} \epsilon}+a^{2} \epsilon\left(\partial_{i} \zeta\right)^{2}\right)~,\nonumber\\
&H_{{I}}=\int d^{3} x\left(\frac{c_{s}^{6} \lambda}{4 a^{5} H^{3} \epsilon^{3}} \pi_{\zeta}^{3}+\frac{c_{s}^{2}\left(c_{s}^{2}-1\right) \Sigma}{2 a H^{3} \epsilon} \pi_{\zeta}\left(\partial_{i} \zeta\right)^{2}\right)~.
\end{align}
We wish to split the whole Hamiltonian to three pieces
\begin{align}
&H=H_1+H_2+H_3~,\nonumber\\
&H_{1}=\int d^{3} x\left(\frac{c_{s}^{2} \pi_{\zeta}^{2}}{4 a^{2} \epsilon}+\frac{c_{s}^{6} \lambda}{4 a^{5} H^{3} \epsilon^{3}} \pi_{\zeta}^{3}\right)~,\nonumber\\
&H_{2}=\int d^{3} x a^{2} \epsilon\left(\partial_{i} \zeta\right)^{2}~,\nonumber\\
&H_{3}=\int d^{3} x \frac{c_{s}^{2}\left(c_{s}^{2}-1\right) \Sigma}{2 a H^{3} \epsilon} \pi_{\zeta}\left(\partial_{i} \zeta\right)^{2}~.
\end{align}
Here, $H_1$ and $H_2$ stand for the Hamiltonians represented by purely the field quanta and the field momenta, respectively, and $H_3$ is the mixing term. The lattice version is given by
\begin{align}
&H=H_1+H_2+H_3~,\nonumber\\
&H_{1} =\sum_{\textbf{x} \in \Omega_3} b^{3}\left(\frac{c_{s}^{2} \pi_{\zeta}^{2}(\textbf{x})}{4 a^{2} \epsilon}+\frac{c_{s}^{6} \lambda}{4 a^{5} H^{3} \epsilon^{3}} \pi_{\zeta}^{3}(\textbf{x})\right)~,\nonumber\\
&H_{2}=\sum_{\textbf{x} \in \Omega_3} b^{3} a^{2} \epsilon(\nabla_i \zeta(\textbf{x}))^{2} ~,\nonumber\\
&H_{3} =\sum_{\textbf{x} \in \Omega_3} b^{3} \frac{c_{s}^{2}\left(c_{s}^{2}-1\right) \Sigma}{2 a H^{3} \epsilon} \pi_{\zeta}(\textbf{x})(\nabla_i \zeta(\textbf{x}))^{2}~.
\end{align}
The Hamiltonian $H$ here will be the Hamiltonian we wish to implement in the quantum device. We have the following comments.
\begin{itemize}
\item The above splitting $H=H_1+H_2+H_3$ is important when we are doing the Trotter simulation. In the product formula, the Trotter error is estimated by chains of commutators, and the commutation relation about the field and the field momentum will be important to determine those commutators. We will discuss details about this calculation later.
\item We owe some higher-order corrections when we are doing the Legendre transformation. That is, the Legendre transformation we have done from the Lagrangian to the Hamiltonian is correct up to the leading order coupling constant. This is because the non-trivial couplings appear in the Legendre transformation. Although we have promised that we are able to simulate the theory non-perturbatively, it should be the Hamiltonian, not the theory defined by the Lagrangian.
\end{itemize}

\subsection{Further comments}
Finally, we make some further comments on some aspects of inflationary perturbation theory.
\begin{itemize}
\item The field theory content and its analog in the lattice, described in this section, could be generically applicable for other field theories which are relativistic. One might apply techniques mentioned in this paper to determine problems they are interested in for non-relativistic quantum field theories.
\item Choice of momentum/coordinates. When doing quantum simulation, in this paper, we will use the action as the integral over local fields in the spacetime. We write local fields as variables of spacetime coordinates. However, cosmological observables are usually written directly in terms of momenta instead of coordinates. Thus, an extra Fourier transform will be performed classically after we perform quantum simulations of correlation functions measured in different space coordinates. The advantage of such treatment is that we still have a local quantum field theory when we are doing quantum simulation, admit quantum advantage when we are using the product formula to do digital simulations. Of course, one could consider directly simulating the quantum field theory with coupling in the momentum space. However, in this case, we lose our local quantum field theory description. For instance, in this case, we will have a summation of modes of the form $\sigma_{\mathbf{k}} \sigma_{-\mathbf{k}}$, which is not local in the space of momenta. Furthermore, the correlation functions measured in the momentum space will carry a delta function $\delta^3 (\sum \mathbf{k})$ due to momentum conservation, which is not easy to deal with numerically. In order to read the prefactor in front of the delta function, we probably need to make an extra integral over some momenta variables to integrate the delta function out. This will bring us some extra costs numerically. In any case, to get the full correlation functions, we always need the information specified by different momenta. We leave the direct treatment in the momentum space for future works. For simulating quantum field theories written in the non-local form, see a review in quantum computational chemistry \cite{mcardle2020quantum}, and recent applications \cite{Liu:2020eoa}.
\item The Bunch-Davies vacuum. Fruitful symmetries in de Sitter space lead to a series of definitions of vacuum states for quantum field theories living in it, which are usually called the $\alpha$-states or the $\alpha$-vacua. Among those states, the Bunch-Davies vacuum plays a special role: it is the vacuum state that looks instantly the Minkovski vacuum, namely, the zero-particle state observed by a free-fall observer. Furthermore, the Bunch-Davies vacuum minimizes the initial energy in our time-dependent Hamiltonian. Thus, the Bunch-Davies vacuum becomes a natural choice when we are doing the cosmic perturbation theory. For cosmological consequences of non-Bunch-Davies vacua, see \cite{Easther:2001fi,Jiang:2016nok} for discussions. For recent discussions about simulation of Bunch-Davies vacuum, see \cite{Lewis:2019oyx}. Exploring non-Bunch-Davies vacua and looking for possible simulation opportunities in quantum technology are important research directions, and in fact, they are closely related to some fundamental assumptions in cosmic perturbation theory, see \cite{Jiang:2016nok,Burgess:2009bs}.
\item Other coordinate systems. The coordinates we study in this paper are widely used in cosmology literature, but it is only one possible slicing of the global de Sitter. With our coordinate system, it is easy to generalize the flat-space tricks since, in our coordinate, it is manifestly conformally flat. However, one could indeed consider other coordinates. Other slicings and other patches in the de Sitter space might mix the notion of space and time in the current definition. There are similar studies in AdS where people are considering the discretization of the whole spacetime and study their corresponding lattice quantum field theories (see, for instance, \cite{Brower:2019kyh}, and some previous related research \cite{Harlow:2011az,Gubser:2016guj,Heydeman:2016ldy,Bao:2017iye,Osborne:2017woa}).
\item The picture. We are quantizing the field in the Heisenberg picture, and the Hamiltonian is also written in the Heisenberg picture. In the free theory, the mode function satisfies the equation of motion. Then the quantum field satisfies the Heisenberg equation. Note that now the Hamiltonian is completely time-dependent, so the Hamiltonian in the Heisenberg picture is not the same as the Hamiltonian in the Schr\"{o}dinger picture. In our time-dependent quantum field theory setup, it is more convenient to use the Heisenberg picture, but it is natural to understand quantum computing in the Schr\"{o}dinger picture. However, it does not affect the current purpose we have. At the first step that we propose in the introduction, we will use the adiabatic state preparation. Since the time is fixed now at $\tau_0$, and we slowly turn on the coupling constant dynamically, the Hamiltonian and Schr\"{o}dinger Hamiltonians are the same. In the second step, we will use the Trotter time evolution to compute the expectation value of correlation functions evaluated at a later time. Then, we could use quantum computation to evaluate the unitary transformation in the Heisenberg picture, acting on the interacting vacuum. Hence, we are mimicking the Heisenberg evolution operator using the Schr\"{o}dinger evolution of the quantum circuit. It might be interesting to consider using the Schr\"{o}dinger picture completely in future works.
\item Trans-Planckian problem. How early should we set $\tau_0$ when we are doing the simulation? We know that the comoving scale should be exponentially small in the past since the universe is exponentially expanding. Thus, early time should correspond to extremely short distances, which might be smaller than the Planck length. This is called the Trans-Planckian problem in cosmology, since we don't have a Planckian theory of quantum gravity \cite{Martin:2000xs}. We are not addressing the conceptual resolution to the problem, but it is practically related to where we could choose $\tau_0$ when we are doing the simulation. In fact, this could be understood as an advantage of our simulation. We expect $\tau_0$ would happen at the very far past, and we could test different values of initial time $\tau_0$ other than the ideal $-\infty$ when we are doing the perturbative quantum field theory. This might be helpful for theoretical physicists to address the Trans-Planckian issue. For instance, one could check if the vacuum is really adiabatic on sub-horizon scales. In fact, this issue is closely related to the in-in formalism we will discuss later since we care about interactions in the Trans-Planckian regime happening at an extremely early time\footnote{The adiabaticity of the vacuum here is not exactly the same as the adiabatic quantum computation we will discuss later. In the later treatment, we artificially turn on the interactions, so it is a time-dependent dynamical process at the fixed initial time. But the adiabatic vacuum is made by the time-dependent quantum field theory we discuss here introduced by the scale factor.}.
\end{itemize}
\newpage

\section{The in-in formalism}\label{inin}
The above section describes how to write down the action of the inflationary perturbations. However, it does not cover how we could make predictions from those actions. The way we are doing perturbation theory and making predictions in cosmology is quite different from what we are usually doing in the flat space. In this section, we will briefly introduce the most popular framework in cosmology, the in-in formalism, and how those observables are directly connected to observations and experiments. Understanding the framework is important for us to perform quantum simulation, which we will comment briefly in this section. Furthermore, exploring the bounded error beyond the approximation of the in-in formalism provides us important motivations for quantum simulation.
\subsection{The theory}
The in-in formalism is a powerful approach for computing correlation functions for the time-dependent Hamiltonian, and it is specifically useful for cosmic perturbation theory. Unlike the S-matrix in the flat space, which starts from infinite past to infinite future (the in-out formalism), the in-in formalism is concerning the correlation functions where bras and kets are both from the past, namely, both the ``in states".

As we know, our observational data about the universe should, in principle, come from observables in quantum mechanics. Generically, we wish to compute the following expectation value in the Heisenberg picture,
\begin{align}
\left\langle\Omega_{\mathrm{in}}\left(\tau_{0}\right)\left|\mathcal{O}_{H}(\tau)\right| \Omega_{\mathrm{in}}\left(\tau_{0}\right)\right\rangle~,
\end{align}
where the operator is now evolving with time, and the state is static, equivalent to the state at the initial time $\tau_0$.

In this section, we will describe how to evaluate the above variable in the cosmic perturbation theory, how it is connected to observations, and how it is related to quantum simulation. Since our theory could generically apply to time-dependent systems, we will use the notation $\tau$ for time, and when applying the theory to cosmic inflation, we will use the in-in formalism completely in the conformal time.
\subsubsection{Basics}
We start with some kindergarten quantum mechanics. Before we start our discussion, we need to quote the following celebrated Feynman's disentangling theorem.
\begin{theorem}[Feynman's disentangling theorem]\label{feyn}
For time-dependent operators $A$ and $B$, and $\tau_1>\tau_0$, we have
\begin{align}
\mathbb{T}\exp \left[ {\int_{{\tau _0}}^{\tau_1}  {\left( {A\left( {\tau '} \right) + B\left( {\tau '} \right)} \right)} d\tau '} \right] = \mathbb{T}\exp \left[ {\int_{{\tau _0}}^{\tau_1}  {A\left( {\tau '} \right)} d\tau '} \right]\mathbb{T}\exp \left[ {\int_{{\tau _0}}^{\tau_1}  {\tilde B\left( {\tau '} \right)d\tau '} } \right]~,
\end{align}
where
\begin{align}
\tilde B(\tau) \equiv \bar{\mathbb{T}}\exp \left[ { - \int_{{\tau _0}}^\tau  A \left( {\tau '} \right)d\tau '} \right]B(\tau )\mathbb{T}\exp \left[ {\int_{{\tau _0}}^\tau  A \left( {\tau '} \right)d\tau '} \right]~,
\end{align}
and $\mathbb{T}$ is the time-ordering operator, $\bar{\mathbb{T}}$ is the anti-time-ordering operator.
\end{theorem}
According to the theorem, we will discuss some time-dependent quantum mechanics. Consider that now we are studying the time-dependent Hamiltonian $H_S(\tau)$ in the Schr\"{o}dinger picture (or we could call it as the Schr\"{o}dinger Hamiltonian). The evolution of the state $\left|\psi_{S}(\tau)\right\rangle$ is given by
\begin{align}
\ket{\psi_S (\tau) }=U_S (\tau,\tau_0) \ket{\psi_S (\tau_0)}~,
\end{align}
where $U_S$ is given by a Dyson series
\begin{align}
{U_S}(\tau,{\tau_0}) = \mathbb{T}\exp \left[ { - i\int_{{\tau_0}}^\tau H_S \left( {\tau'} \right){d}\tau'} \right]~.
\end{align}
Now, let us consider measuring the observable $\mathcal{O}$. In the Schr\"{o}dinger picture, it is given by
\begin{align}
\left\langle {{\psi _S}(\tau )|{\cal O}|{\psi _S}(\tau )} \right\rangle  = \left\langle {{\psi _S}\left( {{\tau _0}} \right)\left| {U_S^\dag \left( {\tau ,{\tau _0}} \right){\cal O}{U_S}\left( {\tau ,{\tau _0}} \right)} \right|{\psi _S}\left( {{\tau _0}} \right)} \right\rangle~.
\end{align}
Alternatively, we could define the Heisenberg picture. We define the time-dependent Heisenberg operator
\begin{align}
U_S^\dag \left( {\tau ,{\tau _0}} \right){\cal O}{U_S}\left( {\tau ,{\tau _0}} \right) = {{\cal O}_H}(\tau )~,
\end{align}
and the time-dependent result of the measurement is given by
\begin{align}
\left\langle {{\psi _S}(\tau )|{\cal O}|{\psi _S}(\tau )} \right\rangle  = \left\langle {{\psi _S}\left( {{\tau _0}} \right)\left| {{{\cal O}_H}(\tau )} \right|{\psi _S}\left( {{\tau _0}} \right)} \right\rangle~.
\end{align}
Specifically, we define the Hamiltonian in the Heisenberg picture (or we call it the Heisenberg Hamiltonian),
\begin{align}
U_S^\dag \left( {\tau ,{\tau _0}} \right){H_S}(\tau ){U_S}\left( {\tau ,{\tau _0}} \right) = {H_H}(\tau )~.
\end{align}
We could compute the conjugate of the unitary operator $U_S$:
\begin{align}
&U_S^\dag \left( {\tau ,{\tau _0}} \right) = U_S^{ - 1}\left( {\tau ,{\tau _0}} \right)\nonumber\\
&= {\left( {1 - i\int_{{\tau _0}}^\tau  d \tau 'H\left( {\tau '} \right) - \int_{{\tau _0},{\tau _1} > {\tau _2}}^\tau  d {\tau _1}d{\tau _2}H\left( {{\tau _1}} \right)H\left( {{\tau _2}} \right) +  \ldots } \right)^\dag }\nonumber\\
&= {\left( {1 + i\int_{{\tau _0}}^\tau  d t'H\left( {\tau '} \right) - \int_{{\tau _0},{\tau _1} > {\tau _2}}^\tau  d {\tau _1}d{\tau _2}H\left( {{\tau _2}} \right)H\left( {{\tau _1}} \right) +  \ldots } \right)^\dag }\nonumber\\
&= \bar{\mathbb{T}}\exp \left( {i\int_{{\tau_0}}^\tau {d\tau'H(\tau')} } \right)~.
\end{align}
Note that the Heisenberg Hamiltonian and the Schr\"{o}dinger Hamiltonian are different in time-dependent quantum mechanics. Now we could use the Feynman's disentangling theorem to show that
\begin{theorem}\label{heisenbergtime}
For $\tau_2 > \tau_1$, we have
\begin{align}
{{\cal O}_H}({\tau _2}) = U_H^\dag \left( {{\tau _2},{\tau _1}} \right){{\cal O}_H}({\tau _1}){U_H}\left( {{\tau _2},{\tau _1}} \right)~,
\end{align}
where
\begin{align}
{U_H}\left( {{\tau _2},{\tau _1}} \right) = \mathbb{T}\exp \left[ { - i\int_{{\tau _1}}^{{\tau _2}} {{H_H}} \left( {\tau '} \right)d\tau '} \right]~.
\end{align}
\end{theorem}
\begin{proof}
Consider an operator $\mathcal{O}_S$ in the Schr\"{o}dinger picture. By definition, we have
\begin{align}
{{\cal O}_S} = {U_S}\left( {{\tau _1},{\tau _0}} \right){{\cal O}_H}({\tau _1})U_S^\dag \left( {{\tau _1},{\tau _0}} \right) = {U_S}\left( {{\tau _2},{\tau _0}} \right){{\cal O}_H}({\tau _2})U_S^\dag \left( {{\tau _2},{\tau _0}} \right)~.
\end{align}
Then
\begin{align}
{{\cal O}_H}({\tau _2}) = U_S^\dag \left( {{\tau _2},{\tau _0}} \right){U_S}\left( {{\tau _1},{\tau _0}} \right){{\cal O}_H}({\tau _1})U_S^\dag \left( {{\tau _1},{\tau _0}} \right){U_S}\left( {{\tau _2},{\tau _0}} \right)~.
\end{align}
Namely we want to show
\begin{align}
U_S^\dag \left( {{\tau _1},{\tau _0}} \right){U_S}\left( {{\tau _2},{\tau _0}} \right) = {U_H}\left( {{\tau _2},{\tau _1}} \right)~.
\end{align}
Treat $\tau_1$ as $\tau_2$ in the statement of Theorem \ref{feyn}. Then take
\begin{align}
&A(\tau ) = \texttt{Boole}({\tau _0} < \tau  < {\tau _1})\left( { - {H_S}(\tau )} \right)~,\nonumber\\
&B(\tau ) = \texttt{Boole}({\tau _1} < \tau  < {\tau _2})\left( { - {H_S}(\tau )} \right)~,
\end{align}
where
\begin{align}
\texttt{Boole}(x) = \left\{ {\begin{array}{*{20}{c}}
{x = \texttt{yes}:}&1\\
{x = \texttt{no}:}&0
\end{array}} \right.~.
\end{align}
Then we get the answer.
\end{proof}
The Heisenberg picture will be the starting point in the following discussion. So when we are talking about the Hamiltonian $H$, we mean the Heisenberg Hamiltonian in this paper.

We consider a time-dependent Hamiltonian in the perturbation theory. We define the following splitting in the Heisenberg picture
\begin{align}
H(\tau)= {H_0}(\tau) + {H_I}(\tau)~,
\end{align}
where $H_0$ here is the free Hamiltonian, and $H_I$ is the small interaction term. Now we define the interaction picture. The states in the interaction picture are defined as
\begin{align}
&\left| {{\psi _I}(\tau )} \right\rangle  = \bar{\mathbb{T}}\exp \left( {i\int_{{\tau _0}}^\tau  d \tau '{H_0}\left( {\tau '} \right)} \right)\left| {{\psi _S}(\tau )} \right\rangle \nonumber\\
&= \bar{\mathbb{T}}\exp \left( {i\int_{{\tau_0}}^\tau d \tau '{H_0}\left( {\tau '} \right)} \right)\mathbb{T}
\exp \left( { - i\int_{{\tau _0}}^\tau  H \left( {\tau '} \right)d\tau '} \right)\left| {{\psi _S}\left( {{\tau _0}} \right)} \right\rangle~.
\end{align}
From this formula, we obtain the evolution operator in the interaction picture,
\begin{align}
{U_I}\left( {\tau ,{\tau _0}} \right) = \mathbb{T}\exp \left( { - i\int_{{\tau _0}}^\tau  d \tau '{{\tilde H}_I}\left( {\tau '} \right)} \right)~,
\end{align}
where $\tilde{H}_I$ is the \emph{interacting Hamiltonian} $H_I$ \emph{in the interaction picture}. Defining
\begin{align}
{U_0}\left( {\tau ,{\tau _0}} \right) =\mathbb{T}\exp \left( { - i\int_{{\tau _0}}^\tau  d \tau '{H_0}\left( {\tau '} \right)} \right)~,
\end{align}
we define $\tilde{H}_I$ as
\begin{align}
{\tilde H_I}(\tau ) = U_0^\dag (\tau ,{\tau _0}){H_I}(\tau ){U_0}(\tau ,{\tau _0})~.
\end{align}
Now, we could use the Feynman's disentangling theorem to show that
\begin{theorem}
\begin{align}
{U_I}\left( {\tau,{\tau_0}} \right) = U_0^\dag \left( {\tau,{\tau_0}} \right)U_H \left( {\tau,{\tau_0}} \right)~.
\end{align}
\end{theorem}
\begin{proof}
Understand $\tau_2$ as $\tau$ in the statement of Theorem \ref{feyn}. Call
\begin{align}
A =  - {H_0}~,~~~~~~B =  - {H_I}~.
\end{align}
Then we get the answer.
\end{proof}
Thus, one could check the expectation value of operators by defining operators in the interaction picture $\tilde{\mathcal{O}}$,
\begin{align}
\tilde {\cal O}(\tau) = U_0^\dag \left( {\tau,{\tau_0}} \right){\cal O}{U_0}\left( {\tau,{\tau_0}} \right)~,
\end{align}
so we have
\begin{align}
\left\langle {{\psi _I}(\tau)} \right|\tilde{\mathcal{O}}(\tau)\left| {{\psi _I}(\tau)} \right\rangle  = \left\langle {{\psi _S}(\tau)} \right| \mathcal{O} \left| {{\psi _S}(\tau)} \right\rangle~.
\end{align}
Now, say that we are interested in computing the expectation value of the Heisenberg operator $\mathcal{O}_H(\tau)$ in the vacuum defined in the interaction theory in the far past $\tau_0$,
\begin{align}
\left\langle {{\Omega _{{\rm{in}}}}({\tau_0})} \right|{{\cal O}_H}(\tau)\left| {{\Omega _{{\rm{in}}}}({\tau_0})} \right\rangle~.
\end{align}
Here, notice that we define the vacuum by the ground state of the Heisenberg Hamiltonian $H(\tau_0)$, which is equal to the Schr\"{o}dinger Hamiltonian $H_S(\tau_0)$. We have
\begin{align}
\left\langle {{\Omega _{{\rm{in}}}}({\tau_0})} \right|{{\cal O}_H}(\tau)\left| {{\Omega _{{\rm{in}}}}({\tau_0})} \right\rangle  = \left\langle {{\Omega _{{\rm{in}}}}({\tau_0})} \right|U_I^\dag \left( {\tau,{\tau_0}} \right)\tilde {\cal O}{U_I}\left( {\tau,{\tau_0}} \right)\left| {{\Omega _{{\rm{in}}}}({\tau_0})} \right\rangle~.
\end{align}

\subsubsection{The $i\epsilon$-prescription and the in-in formalism}
The evolution operator in the interaction picture, $U_I$, could also help us prepare the interacting vacuum from the free theory vacuum in the flat space quantum field theory. This will happen in not only the time-independent Hamiltonian, but also the time-dependent Hamiltonian. We will discuss this point as follows.

Suppose the Hamiltonians $H$, $H_0$ and $H_I$ are all time-independent, we have
\begin{align}
&\exp ( - iH(\tau - {\tau_0}))\left| {{\Omega _{{\rm{free}}}}} \right\rangle = \sum\limits_n {\exp ( - i{E_n}(\tau - {\tau_0}))\left| n \right\rangle \left\langle n | {{\Omega _{{\rm{free}}}}} \right\rangle } \nonumber\\
&= \exp ( - i{E_0}(\tau - {\tau_0}))\left| {{\Omega _{{\rm{in}}}}} \right\rangle \left\langle {{\Omega _{{\rm{in}}}}} |{{\Omega _{{\rm{free}}}}} \right\rangle+ \sum\limits_{n \ne 0} {\exp ( - i{E_n}(\tau - {\tau_0}))\left| n \right\rangle \left\langle n | {{\Omega _{{\rm{free}}}}} \right\rangle }~.
\end{align}
Now, we write $\ket{\Omega_{\text{free}}}$ as the free theory vacuum, the ground state of $H_0$, and $\ket{\Omega_{\text{in}}}$ as the interaction vacuum, the ground state of $H$. Here $\ket{n}$ is used to denote all eigenstates of $H$, and $n=0$ is the ground state. We introduce the following prescription\footnote{Note that this $i\epsilon$-prescription now is the standard $i\epsilon$-prescription in quantum field theory, and this $\epsilon$ is not the slow-roll parameter in cosmic inflation.},
\begin{align}
\tau \to \tilde{\tau}= \tau (1-i\epsilon)~,
\end{align}
where $\epsilon$ is a small positive number. Now we have
\begin{align}
&\sum\limits_n {\exp ( - i{E_n}(\tau - {\tau_0}))\left| n \right\rangle \left\langle {n|{\Omega _{{\rm{free}}}}} \right\rangle } \nonumber\\
&\to \sum\limits_n {\exp ( - i{E_n}(\tilde{\tau} - {{\tilde {\tau}}_0}))\left| n \right\rangle \left\langle {n|{\Omega _{{\rm{free}}}}} \right\rangle } \nonumber\\
&= \sum\limits_n {\exp (\epsilon {E_n}{\tau_0})\exp (i{E_n}{\tau_0} - i{E_n}\tilde{\tau})\left| n \right\rangle \left\langle {n|{\Omega _{{\rm{free}}}}} \right\rangle }~.
\end{align}
Now we take the limit where $\tau_0 \to -\infty$. In this case, the summation over states is dominated by the first term, namely the ground state contribution,
\begin{align}
\exp ( - iH(\tilde{{\tau}} - {\tilde{{\tau}}_0}))\left| {{\Omega _{{\rm{free}}}}} \right\rangle  \approx \exp ( - i{E_0}(\tilde{\tau} - {\tilde{\tau}_0}))\left| {{\Omega _{{\rm{in}}}}} \right\rangle \left\langle {{\Omega _{{\rm{in}}}}} |{{\Omega _{{\rm{free}}}}} \right\rangle~.
\end{align}
From now on, we will not distinguish the difference between $\tau$ and $\tilde{\tau}$ in notation. So we get the formula for the interaction vacuum
\begin{align}
\exp ( - iH(\tau - {\tau_0}))\left| {{\Omega _{{\rm{in}}}}} \right\rangle  \approx \frac{{\exp ( - iH(\tau - {\tau_0}))\left| {{\Omega _{{\rm{free}}}}} \right\rangle }}{{\left\langle {{\Omega _{{\rm{in}}}}|{\Omega _{{\rm{free}}}}} \right\rangle }}~.
\end{align}
Here, $\ket{\Omega_\text{in}}$ is normalized but $\ket{\Omega_\text{free}}$ is not. In the time-dependent case, it is usually claimed that similar things should happen in the time-dependent Hamiltonian,
\begin{align}
\mathbb{T}\exp \left( { - i\int_{{\tau _0}}^\tau  {{{\tilde H}_I}} \left( {\tau '} \right)d\tau '} \right)\left| {{\Omega _{{\rm{in}}}}({\tau _0})} \right\rangle  \approx \frac{{\mathbb{T}\exp \left( { - i\int_{{\tau _0}}^\tau  {{{\tilde H}_I}} \left( {\tau '} \right)d\tau '} \right)\left| {{\Omega _{{\rm{free}}}}({\tau _0})} \right\rangle }}{{\left\langle {{\Omega _{{\rm{in}}}}({\tau _0})|{\Omega _{{\rm{free}}}}({\tau _0})} \right\rangle }}~,
\end{align}
and the same $i\epsilon$-prescription for time coordinates are used. One could argue that this will approximately happen if there is no pole for the Hamiltonian around time $\tau_0$, so the contribution of the whole time-ordered exponential is dominated roughly by $H_0$ for a sufficiently long period of time. Thus we have
\begin{align}
&\left\langle {{\Omega _{{\rm{in}}}}\left( {{\tau_0}} \right)\left| {{{\cal O}_H}(\tau)} \right|{\Omega _{{\rm{in}}}}\left( {{\tau_0}} \right)} \right\rangle \nonumber\\
&= \left\langle {{\Omega _{{\rm{in}}}}\left( {{\tau_0}} \right)\left| {U_I^\dag \left( {\tau,\tau_0} \right)\tilde {\cal O}{U_I}\left( {\tau,\tau_0} \right)} \right|{\Omega _{{\rm{in}}}}\left( {{\tau_0}} \right)} \right\rangle \nonumber\\
&= \left\langle {{\Omega _{{\rm{in}}}}\left( {{\tau_0}} \right)\left| {U_H^\dag \left( {\tau,\tau_0} \right){U_0}\left( {\tau,\tau_0} \right)\tilde {\cal O}U_0^\dag \left( {\tau,\tau_0} \right){U_H}\left( {\tau,\tau_0} \right)} \right|{\Omega _{{\rm{in}}}}\left( {{\tau_0}} \right)} \right\rangle \nonumber\\
&\approx \frac{{\left\langle {{\Omega _{{\rm{free}}}}\left( {{\tau_0}} \right)} \right|U_H^\dag \left( {\tau,\tau_0} \right){U_0}\left( {\tau,\tau_0} \right)\tilde {\cal O}U_0^\dag \left( {\tau,\tau_0} \right){U_H}\left( {\tau,\tau_0} \right)\left| {{\Omega _{{\rm{free}}}}\left( {{\tau_0}} \right)} \right\rangle }}{{{{\left| {\left\langle {{\Omega _{{\rm{in}}}}\left( {{\tau_0}} \right)|{\Omega _{{\rm{free}}}}\left( {{\tau_0}} \right)} \right\rangle } \right|}^2}}}\nonumber\\
&= \frac{{\left\langle {{\Omega _{{\rm{free}}}}\left( {{\tau_0}} \right)} \right|U_I^\dag \left( {\tau,\tau_0} \right)\tilde {\cal O}{U_I}\left( {\tau,\tau_0} \right)\left| {{\Omega _{{\rm{free}}}}\left( {{\tau_0}} \right)} \right\rangle }}{{{{\left| {\left\langle {{\Omega _{{\rm{in}}}}\left( {{\tau_0}} \right)|{\Omega _{{\rm{free}}}}\left( {{\tau_0}} \right)} \right\rangle } \right|}^2}}}~.
\end{align}
Note that since $\ket{\Omega_\text{in}(\tau_0)}$ is normalized, we could put an identity operator $\mathcal{O}=1$. If we believe that $U_H$, $U_I$ and $U_0$ are still unitaries (although now we use the $i\epsilon$-prescription), operator $\mathcal{O}$ in other pictures, $\mathcal{O}_H$ and $\tilde{\mathcal{O}}$ are also identities. Now we get
\begin{align}
{\left| {\left\langle {{\Omega _{{\rm{in}}}}\left( {{\tau_0}} \right)|{\Omega _{{\rm{free}}}}\left( {{\tau_0}} \right)} \right\rangle } \right|^2} = 1~.
\end{align}
Note that this does not mean that those two states are equal, since $\ket{\Omega_\text{free}(\tau_0)}$ is not normalized.

Finally, we arrive at the main formula used in various literature for computing inflationary perturbations,
\begin{align}
\left\langle {{\Omega _{{\rm{in}}}}\left( {{\tau_0}} \right)\left| {{{\cal O}_H}(\tau)} \right|{\Omega _{{\rm{in}}}}\left( {{\tau_0}} \right)} \right\rangle  = \left\langle {{\Omega _{{\rm{free}}}}\left( {{\tau_0}} \right)} \right|U_I^\dag \left( {\tau,\tau_0} \right)\tilde {\cal O}{U_I}\left( {\tau,\tau_0} \right)\left| {{\Omega _{{\rm{free}}}}\left( {{\tau_0}} \right)} \right\rangle~.
\end{align}
The full methodology is called the in-in formalism in the literature of cosmology. Now we comment on this formalism briefly,
\begin{itemize}
\item The approximation
\begin{align}
\mathbb{T}\exp \left( { - i\int_{{\tau _0}}^\tau  {{{\tilde H}_I}} \left( {\tau '} \right)d\tau '} \right)\left| {{\Omega _{{\rm{in}}}}({\tau _0})} \right\rangle  \approx \frac{{\mathbb{T}\exp \left( { - i\int_{{\tau _0}}^\tau  {{{\tilde H}_I}} \left( {\tau '} \right)d\tau '} \right)\left| {{\Omega _{{\rm{free}}}}({\tau _0})} \right\rangle }}{{\left\langle {{\Omega _{{\rm{in}}}}({\tau _0})|{\Omega _{{\rm{free}}}}({\tau _0})} \right\rangle }}~,
\end{align}
is promisingly reasonable since the integral is dominated by the early time $\tau_0$ because of the inflationary universe. However, this assumption is not as solid as the perturbative quantum field theory in the flat space, where we have a rigorous statement mathematically, which is called the Gell-Mann and Low theorem \cite{GellMann:1951rw}. Further issues include the statement where ${\left| {\left\langle {{\Omega _{{\rm{in}}}}\left( {{\tau_0}} \right)|{\Omega _{{\rm{free}}}}\left( {{\tau_0}} \right)} \right\rangle } \right|^2} = 1 $, which requires unitarity but in fact it is not, due to the Wick-rotated time. We believe that the current treatment in the $i\epsilon$-prescription of cosmology is physically correct, but can we rigorously prove it? Furthermore, can we bound the error from the above approximation? They are still open problems. Some related insightful discussions include \cite{Jiang:2016nok,Burgess:2009bs}.
\item However, our quantum simulation algorithm completely does not rely on the in-in formalism and the perturbative expansion: it is directly operated in the Heisenberg picture, and the calculation could be performed beyond the perturbative regime. As long as we know the interacting vacuum $\ket{\Omega_{\text{in}}(\tau_0)}$ in the quantum circuit, we could directly evaluate the following expression based on Theorem \ref{heisenbergtime}
\begin{align}
\left\langle {{\Omega _{{\rm{in}}}}\left( {{\tau _0}} \right)\left| {{{\cal O}_H}(\tau )} \right|{\Omega _{{\rm{in}}}}\left( {{\tau _0}} \right)} \right\rangle  = \left\langle {{\Omega _{{\rm{in}}}}\left( {{\tau _0}} \right)\left| {U_H^\dag \left( {\tau ,{\tau _0}} \right){{\cal O}_H}\left( {{\tau _0}} \right){U_H}\left( {\tau ,{\tau _0}} \right)} \right|{\Omega _{{\rm{in}}}}\left( {{\tau _0}} \right)} \right\rangle~,
\end{align}
by constructing the unitary evolution
\begin{align}
{U_H}\left( {\tau ,{\tau _0}} \right) = \mathbb{T}\exp \left[ { - i\int_{{\tau _0}}^\tau  {{H_H}} \left( {\tau '} \right)d\tau '} \right]~,
\end{align}
and measure the expectation value using, for instance, post-selection. Thus, we claim that our quantum simulation program could justify the correctness and bound the error of the in-in formalism, hence solve those open problems. Furthermore, since we could extend the algorithm in other cosmic phases (for instance, bouncing universe), it could help us compute correlation functions numerically for geometries beyond inflation where there is no initial manifest dominance in the integral.

\end{itemize}

\subsection{Experimental observables}
Now we give a brief discussion on how the above observables are connected to experiments. Usually, in the experiments, we study correlation functions in momentum space. For a given operator $\mathcal{O}(t,\bf{x})$, the Fourier transform is defined as
\begin{align}
&{\mathcal{O}_{\bf{k}}}(\tau ) = \int {{d^3}} x\mathcal{O}(\tau, {\bf{x}}){e^{ - i{\bf{k}} \cdot {\bf{x}}}}~,\nonumber\\
&\mathcal{O}(\tau, {\bf{x}}) = \int {\frac{{{d^3}k}}{{{{(2\pi )}^3}}}} {\mathcal{O}_{\bf{k}}}(\tau ){e^{i{\bf{k}} \cdot {\bf{x}}}}~.
\end{align}
For the two-point function, in the case of cosmology, we know that it is invariant under translation in the following sense,
\begin{align}
\left\langle {\mathcal{O}(\tau,{\bf{x}})\mathcal{O}(\tau,{\bf{y}})} \right\rangle  \equiv f^{\mathcal{O}}({\bf{x}} - {\bf{y}})~.
\end{align}
Then the two-point function looks like
\begin{align}
&\left\langle {{\mathcal{O}_{{{\bf{k}}_1}}}{\mathcal{O}_{{{\bf{k}}_2}}}} \right\rangle  = \int {{d^3}} x\int {{d^3}} yf^{\mathcal{O}}({\bf{x}} - {\bf{y}}){e^{ - i{{\bf{k}}_1} \cdot {\bf{x}} - i{{\bf{k}}_2} \cdot {\bf{y}}}}\nonumber\\
&= {(2\pi )^3}\delta^3 ({{\bf{k}}_1} + {{\bf{k}}_2})\int {{d^3}} uf^{\mathcal{O}}({\bf{u}}){e^{ - i{{\bf{k}}_1} \cdot {\bf{u}}}}~.
\end{align}
We could define the power spectrum $P_{\mathcal{O}}$
\begin{align}
&\left\langle {{\mathcal{O}_{{{\bf{k}}_1}}}{\mathcal{O}_{{{\bf{k}}_2}}}} \right\rangle  = {(2\pi )^3}{\delta ^3}\left( {{{\bf{k}}_1} + {{\bf{k}}_2}} \right)\frac{{2{\pi ^2}}}{{k_1^3}}{P_{\mathcal{O}}}({\mathbf{k}_1})~,\nonumber\\
&{P_{\mathcal{O}}}({\bf{k}}) = \frac{{{k^3}}}{{2{\pi ^2}}}\int {{d^3}} u{f^{\mathcal{O}}}({\bf{u}}){e^{ - i{\bf{k}} \cdot {\bf{u}}}}~.
\end{align}
For curvature perturbation $\mathcal{O}=\zeta$, we could firstly compute the correlation function in the coordinate space to obtain the function $f^{\zeta}$ using quantum simulation, then we could perform the Fourier transform above to compute the power spectrum. Observationally, the power spectrum $P_\zeta$ could be determined by the CMB map or the LSS, specifically at the time of ``horizon exit" where $k=aH$. At this time, quantum perturbations decohere and become statistical perturbations. Now the function $P_{\zeta}(\mathbf{k})=P_{\zeta}(k)$ becomes purely a function of the momentum norm (the wavenumber). To make further contact with experimental observables, we could define transfer functions that could keep track of cosmic evolution after the horizon exit. In the CMB, the angular power spectrum of CMB temperature fluctuations $C_\ell$ could be written as an integral of the primordial power spectrum $P_\zeta$. In the LSS, the late-time power spectrum of dark matter density fluctuations is proportional to the $P_\zeta$ with certain transfer functions. For further knowledge, see the lecture notes \cite{Baumann:2009ds} and references therein.

Similarly, one could consider three-point functions. The translational invariance of the three-point function in the coordinate space ensures that in the momentum space, there is a factor given by the delta function. Moreover, we could define the bispectrum $\mathcal{F}$ of the curvature perturbation $\zeta$
\begin{align}
\left\langle\zeta_{\mathbf{k}_{1}} \zeta_{\mathbf{k}_{2}} \zeta_{\mathbf{k}_{3}}\right\rangle=(2 \pi)^{7} \delta^{3}\left(\mathbf{k}_{1}+\mathbf{k}_{2}+\mathbf{k}_{3}\right) \frac{P_{\zeta}^{2}}{k_{1}^{2} k_{2}^{2} k_{3}^{2}} \mathcal{F}\left(k_{1} / k_{3}, k_{2} / k_{3}\right)~,
\end{align}
which is directly related to non-Gaussianities of the corresponding inflationary model. Non-Gaussianities provide important primordial information of particles and interactions in the early universe, and they are directly related to observations.

\subsection{Examples}
Inherited from our previous discussions, we present the results directly without specifying the process, referring to previous results in \cite{Chen:2006nt}. In general single-field inflationary models we discussed before, the leading-order power spectrum is computed directly from the free theory,
\begin{align}
{P_\zeta } = \frac{{{H^2}}}{{8{\pi ^2}{c_s}\epsilon }}~.
\end{align}
The leading-order bispectrum is given by the tree-level diagram using the in-in formalism.
\begin{align}
&{\cal F} = \left( {\frac{1}{{c_s^2}} - 1 - \frac{{2\lambda }}{\Sigma }} \right)\frac{{3{k_1}{k_2}{k_3}}}{{2{{\left( {{k_1} + {k_2} + {k_3}} \right)}^3}}}+ \left( {\frac{1}{{c_s^2}} - 1} \right)\times \nonumber\\
&\left[ { - \frac{{k_1^2k_2^2 + k_1^2k_3^2 + k_2^2k_3^2}}{{{k_1}{k_2}{k_3}\left( {{k_1} + {k_2} + {k_3}} \right)}} + \frac{{k_1^2k_2^3 + k_1^2k_3^3 + k_2^2k_3^3 + k_2^2k_1^3 + k_3^2k_1^3 + k_3^2k_2^3}}{{2{k_1}}} + \frac{{k_1^3 + k_2^3 + k_3^3}}{{8{k_1}{k_2}{k_3}}}} \right]~,
\end{align}
with the parameter $\lambda$ and $\Sigma$ defined above. Note that the result is perturbative in the slow-roll parameter.

\subsection{Further comments}
At the current stage, we wish to make some further comments.
\begin{itemize}
\item As we mentioned before, note that, in the above sections, we mostly describe the perturbative quantum field theory formalism theoretically. Progress and problems we have mentioned before about this theoretical prescription provide us motivations to perform quantum simulation in the future quantum device and benchmark the quantum simulation algorithms and devices using the known answer from quantum field theory calculations. However, when we are doing quantum simulation, we are completely not using the perturbative method in the interaction picture. In fact, we don't even need to introduce the interaction picture, and we could just focus on the Heisenberg picture when we are doing quantum simulation.
\item Furthermore, we wish to mention that when doing quantum simulation, in this paper, we are quoting the perturbative action we will use as examples. However, we could even imagine that we could just simulate the original action beyond the slow-roll expansion. Extra treatments are needed in this process to separate the classical and the quantum part of the action, and furthermore, implement them in the quantum device. It could potentially be an interesting generalization of our current work.
\item As we have described, the Hamiltonian we are considering is intrinsically time-dependent even for the Heisenberg picture since the spacetime is exponentially expanding with time. Thus, it might be interesting to apply methods and tricks from the study of quantum open-system and quantum thermodynamics in quantum information science, for instance, the Lindblad equation and the quantum resource theory. Moreover, the precise formulation of open systems in quantum field theory needs to be formalized \cite{JPwork}. Moreover, it might be interesting to make use of Floquet dynamics in quantum many-body physics to study conceptual problems in cosmology and analog quantum simulation, since the time-dependent dynamics in Floquet systems, as periodically driven open systems, might be similar to the cosmic evolution of cyclic or bouncing cosmologies (see a related work \cite{Easson:2018qgr}). Furthermore, for cyclic or bouncing theories, they might be more suitable for quantum algorithms, since they are strongly-coupled processes, which are relatively harder to understand compared to the perturbative dynamics.
\end{itemize}

\newpage
\section{Designing the algorithm}\label{algorithm}
\subsection{The original Jordan-Lee-Preskill algorithm}
The Jordan-Lee-Preskill algorithm \cite{Jordan:2011ne,Jordan:2011ci} could be regarded as a generic paradigm for simulating quantum field theories in the quantum computer. Here we briefly review the original version of the Jordan-Lee-Preskill algorithm, and our generalization will be discussed later on.

The original algorithm is designed for simulating the $\lambda \phi^4$ scalar quantum field theory in general spacetime dimensions, where the lattice version of Hamiltonian is
\begin{align}
{H_t } = {b^3}\sum\limits_{{\bf{x}} \in {\Omega _3}} {\left[ {\frac{1}{2}\pi _\phi ^2({\bf{x}}) + \fft{1}{2} {{\left( {{\nabla _i}\phi ({\bf{x}})} \right)}^2}}+\fft{1}{2}m_0^2 \phi ({\bf{x}})^2+\fft{\lambda_0}{4!}\phi ({\bf{x}})^4 \right]}~,
\end{align}
where the coupling $\lambda_0$ can either be weak or strong. Although the $\lambda\phi^4$ scalar quantum field theory now is defined in a lattice, the value of the scalar field $\phi({\bf x})$ at each site is continuous and generally unbounded, while the digital quantum computer is only capable of managing a finite number of qubits. The idea of \cite{Jordan:2011ne,Jordan:2011ci} is to bound the field value by $\phi_{\max}$ with discretize step size as $\delta\phi$ at each site, i.e.,
\begin{align}
\{  - {\phi _{\max }}, - {\phi _{\max }} + \delta \phi , \ldots ,{\phi _{\max }} - \delta \phi ,{\phi _{\max }}\}~,\label{bound of field}
\end{align}
such that the original Hilbert space is truncated to be finite-dimensional, allowing to encode the scalar quantum field theory with finite qubits into quantum computers. A key result of \cite{Jordan:2011ne,Jordan:2011ci} is to determine the truncation of the Hilbert space and the number of qubits per site by the scattering energy $E$. To simulate the quantum field scattering process, initial scattering states should be excited. As described in \cite{Jordan:2011ci}, this task can be fulfilled by preparing the initial vacuum state of free theories $|\Omega_{\rm free}\rangle$ followed by exciting the wave packets (which are represented by $|\psi\rangle=a_{\psi}^\dagger|\Omega_{\rm free}\rangle$) based on the prepared vacuum, where exciting the wave packets can be realized by involving an auxiliary Hamiltonian with additional ancilla qubits
\be
H_\psi=a_\psi^\dagger \otimes|1\rangle\langle0|+a_\psi \otimes|0\rangle\langle1|\,,\quad e^{-iH_\psi \fft{\pi}{2}}|\Omega_{\rm free}\rangle
\otimes |0\rangle=-i|\psi\rangle\otimes |1\rangle\,.
\ee
Adiabatically turning on the interaction for the wave packets, they will evolve in time to an ending time, which is precisely the scattering process. Simulating the time evolution is a well-known task, and \cite{Jordan:2011ci} have done this by simply splitting the Hamiltonian into field and field momentum piece, then apply the product formula of Trotter (which will be briefly introduced in subsection \ref{Trotter}). As the completion of the whole time evolution of the scattering process, it allows measuring any physical observable that satisfies our simulation goal.

As a summary of the above review, the original version of the algorithm consists of the following prescriptions:
\begin{algorithm}[Jordan-Lee-Preskill] The original Jordan-Lee-Preskill algorithm is an algorithm for simulating the $\lambda \phi^4$ scalar quantum field theory in general spacetime dimensions at both weak or strong couplings. It is given by the following steps.
\begin{itemize}
\item Encoding. We encode the lattice field theory Hilbert space into the quantum computer. The truncation of the original Hilbert space and the number of qubits encoded per site are determined by the scattering energy $E$.
\item Initial state preparation. We construct the initial state using an algorithm proposed by Kitaev and Webb \cite{kitaev2008wavefunction} for constructing multivariate Gaussian superpositions. The algorithm could be improved by some other classical methods \cite{bunch1974triangular,coppersmith1987matrix}.
\item Exciting the wave packets in the free theory. This part is done by introducing ancilla qubits.
\item Adiabatic state preparation. We adiabatically turn on the interaction to construct the wave packet in the interacting theory. The speed of adiabatic state preparation should be slow enough to make sure the resulting wave packet is still a reasonable single-particle wave packet.
\item Trotter simulation. We use the product formula to simulate the time evolution $e^{-iHt}$ by splitting the Hamiltonian to the field piece and the field momentum piece.
\item Measurement. After time evolution, we compute the correlation function we are interested in using the quantum circuit. We could measure field operators, number operators, stress tensor, and other quantities we are interested in.
\end{itemize}
\end{algorithm}
\subsection{Our generalization of Jordan-Lee-Preskill}
The original Jordan-Lee-Presklll algorithm we have described above needs further modifications in order to be applicable to cosmic inflation. Comparing to the original scattering process, the evolution process for computing cosmic correlation functions has the following features.
\begin{itemize}
\item In the initial state preparation, we are not setting the scattering energy to be $E$ since we are not doing exactly the scattering experiment. Instead, in cosmic inflation, we are able to set the energy scale of the inflationary perturbation theory $\Lambda=1/b$. We could use $\Lambda$ instead to bound the energy scale, and then the field momentum fluctuation, which will restrict the dimension of the truncated Hilbert space. A similar analysis in inflationary physics could help us bound the field configuration itself.
\item In the scattering experiment in the flat space, the Hamiltonian in the $\lambda \phi^4$ theory is static in the Heisenberg picture and also the Schr\"{o}dinger picture\footnote{However, in the interaction picture, the Hamiltonian is time-dependent.}. However, in the cosmic perturbation theory, the Hamiltonian should be time-dependent in general in both pictures. So the field basis we are using should be generically different for different times. In order to encode the time-dependent Hamiltonian and simulate the Heisenberg evolution
\begin{align}
U_{H}\left(\tau_\text{end}, \tau_{0}\right)=\mathbb{T} \exp \left[-i \int_{\tau_{0}}^{\tau_\text{end}} H_{H}\left(\tau^{\prime}\right) d \tau^{\prime}\right]~,
\end{align}
we need to figure out the field basis transformation depending on different times. This could be realized by computing the Green's identity in the free theory. We discuss this transformation as a reduced case of the HKLL formula \cite{Hamilton:2005ju,Harlow:2018fse}, which is usually appearing in AdS and involving both space and time.
\item We still need the Kitaev-Webb algorithm to prepare the free theory vacuum, but we don't need to excite the wavepacket, since our cosmological-motivated correlation functions are evaluated for the vacuum states.
\item In the adiabatic state preparation process, we start from the free theory at conformal time $\tau_0$, and then we use adiabatic state preparation to construct the interacting vacuum at $\tau_0$. Since the zero-momentum diagonal mode states and the ground state are degenerate in the free theory, as we have discussed before, extra treatment might be needed to split those states in the gapless regime. This treatment is called the ground state projection in the following discussions.
\item Then we use the Trotter simulation to evolve the time-dependent Heisenberg evolution. Note that due to possible mixings between quantum fields and field momenta we have in the Lagrangian, we generalize the original calculation in the Jordan-Lee-Preskill to the three-party product formula case, making use of results from the paper \cite{childs2019theory} and references therein.
\item Cosmic perturbation theories have certain measurement tasks that could be directly related to experimental observations. Here, we measure cosmic correlations
\begin{align}
\left\langle {{\Omega _{{\rm{in }}}}\left( {{\tau _0}} \right)\left| {{{\cal O}_H}(\tau_{\text{end}} )} \right|{\Omega _{{\rm{in }}}}\left( {{\tau _0}} \right)} \right\rangle  = \left\langle {{\Omega _{{\rm{in }}}}\left( {{\tau _0}} \right)\left| {U_H^\dag \left( {\tau_{\text{end}} ,{\tau _0}} \right){{\cal O}_H}\left( {{\tau _0}} \right){U_H}\left( {\tau_{\text{end}} ,{\tau _0}} \right)} \right|{\Omega _{{\rm{in }}}}\left( {{\tau _0}} \right)} \right\rangle~,
\end{align}
instead of other operators that are discussed in the original Jordan-Lee-Preskill algorithm.
\end{itemize}
Thus, we propose the following algorithm that is applicable to compute cosmic correlation functions.
\begin{algorithm}[Jordan-Lee-Preskill for cosmic inflation] We consider the following generalization beyond the original Jordan-Lee-Preskill algorithm, specifically applicable for cosmic correlation functions at general couplings.
\begin{itemize}
\item Encoding. We use the field basis at time $\tau_0$. The range and precision of the field basis are truncated based on the energy scale of EFT $\Lambda=1/b$, and some other inflationary physics. Furthermore, to encode the time-dependent Hamiltonian in the Heisenberg picture, we need the transformation from the time $\tau_0$ to a general time $\tau$. The transformation is given by the Green's identity as a special reduced case of the HKLL formula in the 3+1 dimensional de Sitter space.
\item Initial state preparation. We still use the Kitaev-Webb algorithm and its improvements to prepare the Gaussian vacuum state $\ket{\Omega_{\operatorname{free}}(\tau_0)}$. The variance matrix is given by the two-point function of the free theory.
\item Adiabatic state preparation. We use the adiabatic state preparation to prepare the interacting vacuum $\ket{\Omega_{\operatorname{in}}(\tau_0)}$. Extra treatment, namely the ground state projection, is needed to filter out the zero-momentum states of diagonal modes.
\item Trotter simulation. We use the Trotter algorithm to simulate the Heisenberg time-dependent time evolution. Note that including this evolution and also the adiabatic state preparation piece, we might face the mixing between field and field momenta in the Hamiltonian. Thus, we will use the three-party product formula to do the simulation.
\item Measurement. We measure $\left\langle {{\Omega _{{\rm{in }}}}\left( {{\tau _0}} \right)\left| {{{\cal O}_H}(\tau_{\operatorname{end}} )} \right|{\Omega _{{\rm{in }}}}\left( {{\tau _0}} \right)} \right\rangle$ after the evolution by, for instance, standard algorithms like the post-selection.
\end{itemize}
\end{algorithm}
More details in the above steps will be discussed in the following subsections.

\subsection{Encoding from the HKLL formula}
We start with the encoding problem in our quantum simulation program. At the time $\tau_0$, we define our field basis
\begin{align}
\hat \zeta ({\tau _0},{\bf{x}})\left| {\zeta ({\tau _0},{\bf{x}})} \right\rangle  = \zeta ({\tau _0},{\bf{x}})\left| {\zeta ({\tau _0},{\bf{x}})} \right\rangle~.
\end{align}
On the left hand side, $\hat{\zeta}(\tau_0,\mathbf{x})$ is understood as the curvature perturbation operator, and the state vector $\left| {\zeta ({\tau _0},{\bf{x}})} \right\rangle $ defines the eigenvector. The field value $\zeta (\tau,\mathbf{x})$ could be arbitrary for the fixed space and time  $\mathbf{x}$ and $\tau_0$. So for fixed $\mathbf{x}$, the local Hilbert space dimension is infinite.

Similarly, we define
\begin{align}
\hat \zeta ({\tau},{\bf{x}})\left| {\zeta ({\tau},{\bf{x}})} \right\rangle  = \zeta ({\tau},{\bf{x}})\left| {\zeta ({\tau},{\bf{x}})} \right\rangle~,
\end{align}
for an arbitrary time $\tau$\footnote{In most parts of the paper, we will not distinguish operators with their classical counterparts by the hat notation.}. Note that for different time slices, the state vector will be different in general, since the field operator $\hat \zeta ({\tau},{\bf{x}})$ and its eigenspace are different.

We will use the field basis $\left| {\zeta ({\tau _0},{\bf{x}})} \right\rangle $ to encode our Hamiltonian. Namely, we will represent all terms in our Hamiltonian on the above basis. For terms containing field momentum operators, the matrix elements could be determined by the canonical commutation relation.

However, the above treatment could only work for the initial time slice $\tau_0$. The reason is, our Heisenberg Hamiltonian is manifestly time-dependent. So the question is, can we determine the field operator and the field momentum operator at the time $\tau$ easily from $\tau_0$?

The answer is, yes! In fact, one could naively expect this to happen in the field equation. Since our encoding is based on the free system\footnote{One might worry about the existence of couplings might change the construction of the encoding basis. In fact, we don't really need to worry about it, since this is simply only a basis choice. For instance, we could use the harmonic oscillator basis to encode the $\lambda \phi^4$ strongly-coupled theory in the flat space, although at strong coupling, the harmonic oscillator in the free theory does not exist. In fact, the time evolution of the basis here illustrates how we \emph{define} our time-dependent quantum field theory.}, our field equation is linear. A naive choice is to use the Heisenberg evolution operator
\begin{align}
\zeta \left( \tau  \right) = U_H^\dag \left( {\tau ,{\tau _0}} \right)\zeta \left( {{\tau _0}} \right){U_H}\left( {\tau ,{\tau _0}} \right)~.
\end{align}
But the evolution operator $U_H$ itself contains operators later than $\tau_0$. Instead, we solve the dynamical equation, and the answer is expected to be linear
\begin{align}
\zeta (\tau ,{\bf{x}}) = \int_{{\partial _{{\tau _0}}}\text{LC}^-(\tau ,{\bf{x}})} {{d^3}y\left( {{K_\zeta }(\tau ,{\bf{x}};{\tau _0},{\bf{y}})\zeta ({\tau _0},{\bf{y}}) + {K_\pi }(\tau ,{\bf{x}};{\tau _0},{\bf{y}}){\pi _\zeta }({\tau _0},{\bf{y}})} \right)}~.
\end{align}
In the above equation, $\zeta$ and $\pi_{\zeta}$ are understood as operators, and the integration kernels $K_{\zeta}$ and $K_{\pi}$ are scalar functions. The kernel is supported on the set ${{\partial _{{\tau _0}}}\text{LC}^-(\tau ,{\bf{x}})}$, which is defined as part of the time slice $\tau_0$, intersecting with the past light cone starting from the point $(\tau,\textbf{x})$. The kernels should be supported on the set ${{\partial _{{\tau _0}}}\text{LC}^-(\tau ,{\bf{x}})}$ determined by the causal structure of the theory. Since our metric is manifestly conformally flat when we are using the conformal time, the light cone structures are the same as the flat space ones. Then we could directly write down
\begin{align}
{\partial _{{\tau _0}}}{\rm{L}}{{\rm{C}}^ - }(\tau ,{\bf{x}}) = \left\{ {({\tau _0},{\bf{\bar x}}) \in {\text{inflationary spacetime: }}\left| {{\bf{\bar x}} - {\bf{x}}} \right| < c_s\left| {\tau  - {\tau _0}} \right|} \right\}~.
\end{align}
Here, the notation $\abs{\ldots}$ is the same as the Euclidean distance. Note that here we also consider the non-trivial sound speed $c_s$. We use Figure \ref{HKLLdS} to illustrate the above encoding.

\begin{figure}[t]
  \centering
  \includegraphics[width=0.8\textwidth]{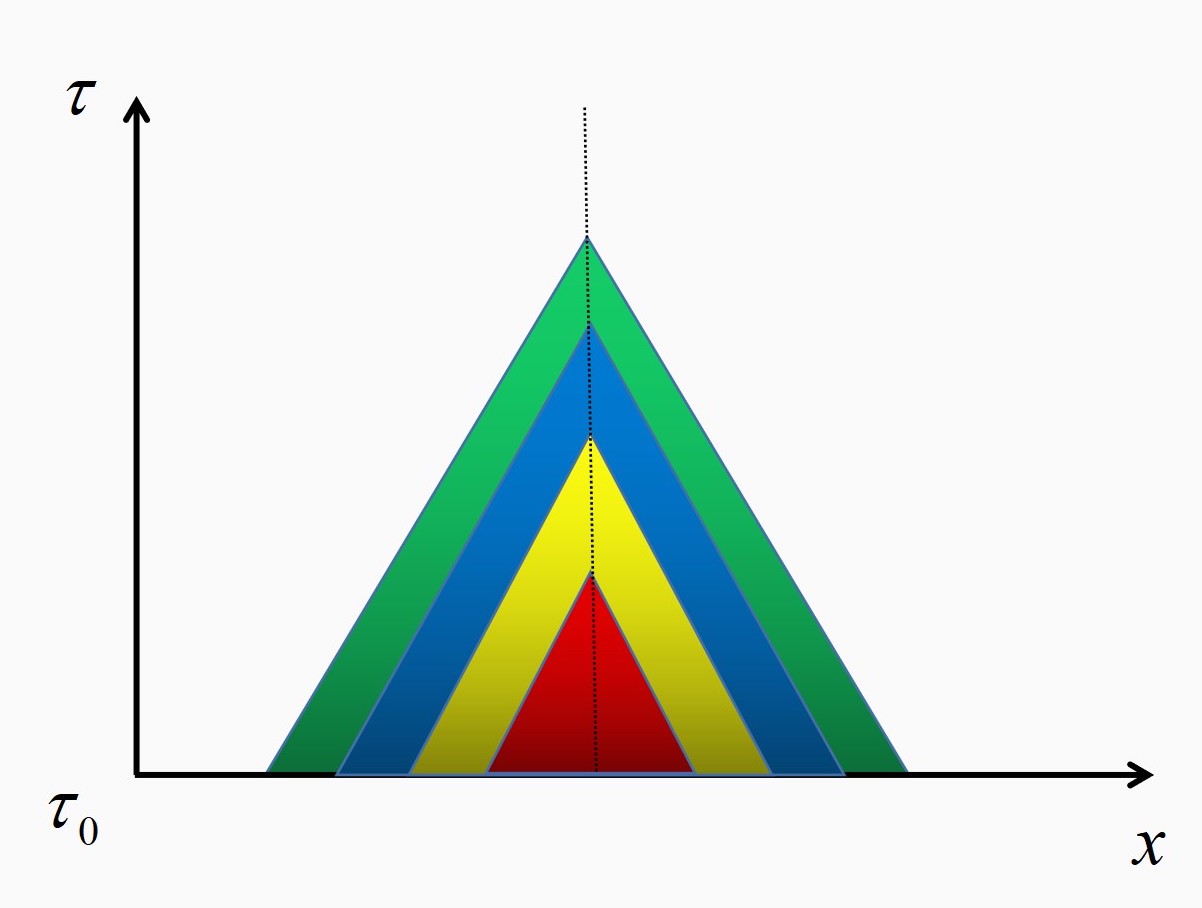}
  \caption{\label{HKLLdS} The encoding process from the past light cone. Here, we use a two-dimensional plot to illustrate past light cones for different $\tau$s at a fixed $\textbf{x}$. Those light cones are intersected with the time slice $\tau_0$. We use a one-dimensional real line to represent the direction of $\textbf{x}$, but in the real world, we have three spatial dimensions.}
\end{figure}

So how to determine the kernel? In our case, it is purely a PDE problem before we promote our variables to operators. In fact, it could be easily solved by Green's identity. Generally speaking, if we consider the relativistic theory for a Klein-Gorden scalar $\phi$ in 3+1 dimensions, the Green's function $G$ will satisfy the following linear equation
\begin{align}
{\nabla _\mu }{\nabla ^\mu }G(\tau ,{\bf{x}};{\tau _0},{\bf{y}}) = \frac{{{\delta ^3}({\bf{x}} - {\bf{y}})\delta (\tau  - {\tau _0})}}{{\sqrt { - g} }}~,
\end{align}
where the covariant derivative $\nabla^\mu$ is acting with respect to $(\tau_0,\mathbf{y})$. The corresponding Green's identity then reads,
\begin{align}
\phi (\tau ,{\bf{x}}) = \int_{{\partial _{{\tau _0}}}\text{LC}^-(\tau ,{\bf{x}})} {{d^3}} y\sqrt {|h|} \,{n^\mu }(\phi ({\tau _0},{\bf{y}}){\nabla _\mu }G(\tau ,{\bf{x}};{\tau _0},{\bf{y}}) - G(\tau ,{\bf{x}};{\tau _0},{\bf{y}}){\nabla _\mu }\phi ({\tau _0},{\bf{y}}))~,
\end{align}
where $h$ and $n$ are induced metric and unit normal vector for the surface ${\partial _{{\tau _0}}}\text{LC}^-(\tau ,{\bf{x}})$. In our case, the equation of motion for the Green's function is modified as
\begin{align}
\fft{2\epsilon}{c_s^2 H^2\tau_0^3}(c_s^2 \tau_0\partial_{y}^2+2\partial_{\tau_0}-\tau_0 \partial_{\tau_0}^2)G(\tau,\mathbf{x};\tau_0,\mathbf{y})=\delta^3(\mathbf{x}-\mathbf{y})
\delta(\tau-\tau_0)~.
\end{align}
The corresponding Green's identity for our purpose is
\begin{align}
\zeta(\tau,\mathbf{x})=\fft{2\epsilon}{c_s^2 H^2\tau_0^2}\int_{{\partial _{{\tau _0}}}\text{LC}^-(\tau ,{\bf{x}})} d^3y \big(G(\tau,\mathbf{x};\tau_0,\mathbf{y})\partial_{\tau_0}\zeta(\tau_0,\mathbf{y})-\zeta(\tau_0,\mathbf{y})\partial_{\tau_0}
G(\tau,\mathbf{x};\tau_0,\mathbf{y})\big)~.
\end{align}
Thus the kernels are given by
\begin{align}
K_\zeta(\tau,\mathbf{x};\tau_0,\mathbf{y})=-\fft{2\epsilon}{c_s^2 H^2\tau_0^2} \partial_{\tau_0}
G(\tau,\mathbf{x};\tau_0,\mathbf{y})~,\quad K_\pi(\tau,\mathbf{x};\tau_0,\mathbf{y})=G(\tau,\mathbf{x};\tau_0,\mathbf{y})~.
\end{align}
In our case, since we already extract the past light cone, we could take the Green's function directly to be the Wightman's two-point function
\begin{align}
G(\tau,\mathbf{x};\tau_0,\mathbf{y})=\int \fft{d^3k}{(2\pi)^3}v_k(\tau)v_k^\ast(\tau_0)e^{i\mathbf{k}\cdot(\mathbf{x}-\mathbf{y})}~.
\end{align}
Taking a derivative of the time coordinate, we will get the corresponding formula for the field momentum.

We will use the above formula to encode the Hamiltonian. What is the nature of the above encoding? In fact, it could be regarded as a reduced version of the HKLL formula in the study of the AdS/CFT correspondence (see the paper \cite{Hamilton:2005ju} and a review \cite{Harlow:2018fse}). In AdS/CFT, a typical problem is to determine the bulk data from the boundary dynamics. The HKLL formula describes how we write the bulk operator from the boundary in the semiclassical theory
\begin{align}
{\mathcal{O}_{{\text{bulk}}}} = \int {\text{HKLL kernel} \times {\mathcal{O}_{{\text{boundary}}}}}~,
\end{align}
from solving PDEs. In the AdS case, the situation is more complicated since space and time are mixed: the boundary contains the time direction. The allowed accessible range in the bulk determined by a given range of the boundary is called the \emph{causal wedge}. Figure \ref{adsrind} illustrates a standard example of the causal wedge reconstruction in the case of AdS-Rindler in three spacetime dimensions \cite{Almheiri:2014lwa,Dong:2016eik}. In our cosmology case, a little difference comparing to AdS/CFT is that right now we understand the time direction as the ``boundary", so we only need the past light cone instead of the full causal wedge. See another discussion about the HKLL formula and bulk reconstruction in \cite{Lewkowycz:2019xse}.

\begin{figure}[t]
  \centering
  \includegraphics[width=0.4\textwidth]{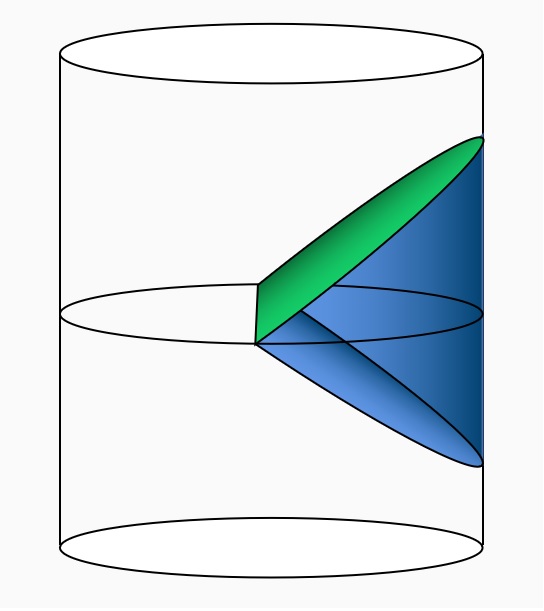}
  \caption{\label{adsrind} The AdS-Rindler bulk reconstruction. Here the bulk operators could be reconstructed from the boundary using the HKLL reconstruction formula. This is the standard example mentioned in, for instance, \cite{Harlow:2018fse}.}
\end{figure}

It is remarkable that we are able to use the HKLL formula for a different purpose beyond the usual study of AdS/CFT. This seems to indicate that studying the nature of spacetime might be closely related to quantum simulation in quantum gravity. Furthermore, it is worth noticing that in our discussion, we use the word ``encoding" in a completely different circumstance. The word ``encoding" we use here means that we are going to make our qubits from our brains to quantum computers, while the word ``encoding" in other literature about AdS/CFT usually means the encoding map in the quantum error correction code \cite{Almheiri:2014lwa,Pastawski:2015qua,Dong:2016eik}. Moreover, the encoding we discuss only requires the free theory. Hence, we are consistent with the semiclassical description and satisfied with the ``causal wedge reconstruction".

It will be interesting to study how the encoding of different time slices will change when turning on the coupling. From the perturbative point of view, we have to receive tree or loop corrections in the Witten diagrams. Non-perturbatively, the causal wedge reconstruction might be replaced by the entanglement wedge reconstruction, and the HKLL formula might be replaced by the modular Hamiltonian and the modular flow. It might be interesting to see how the story will go both in AdS and dS cases. There is an insightful discussion recently about holographic scattering and entanglement wedge \cite{May:2019odp}. Finally, when we are studying the scattering problem in AdS in the quantum computer, probably we might consider using the honest HKLL formula for encoding since, in AdS, space and time direction is mixed. Moreover, one might consider the discretized version of our de Sitter encoding formula in the lattice, and it should not be very hard to obtain since we currently only care about the free theory.

We end this subsection by commenting on the complexity we need to perform encoding. Obviously, the number of terms we need in the encoding map is proportional to the number of sites included in the past light cone regime on the time slice of $\tau_0$. So we have the complexity estimate:
\begin{align}
\text{Encoding Complexity} = O{\left( {\frac{{{c_s}\left| {{\tau _{{\rm{end}}}} - {\tau _0}} \right|}}{b}} \right)^3} = O{\left( {\frac{{{c_s}{\abs{\tau _0}}(1 - {e^{ - N}})}}{b}} \right)^3}~.
\end{align}
Here $N$ is the e-folding number during inflation.

\subsection{Encoding bounds from the EFT scale}
Here, we continue our discussions about the encoding. The ideal study in quantum field theory with infinite-dimensional Hilbert space is not promising for a digital quantum computer, so we have to discretize our field basis and make further truncations.

Now, let us consider the following prescription of truncations. We want to bound the range of the curvature perturbation $\zeta$ at the time $\tau_0$ by $\zeta_{\max}$. Furthermore, we wish the step size (precision) of the discretization of the field value to be $\delta \zeta$. Namely, our choices of field values on each site are exactly eq.~\eqref{bound of field} but replacing $\phi$ by $\zeta$. As a result, the number of qubits is estimated as
\begin{align}
{n_b} \sim \log \left( {\frac{{{\zeta _{\max }}}}{{\delta \zeta }}} \right)~.
\end{align}
How to choose the value of $\delta \zeta$ and $\zeta_{\max}$? Intuitively, we know that if the field fluctuation $\zeta$ is bounded probabilistically (for instance, in terms of expectation values), then we cannot make that much error if we choose $\zeta_{\max}$ to be comparable to the field fluctuation bounds. This intuition is explicitly proved by the original paper of Jordan-Lee-Presklll \cite{Jordan:2011ne,Jordan:2011ci}. We call it the ``Jordan-Lee-Preskill bound". In fact, assuming an $\epsilon_{\text{JLP}}$ error in truncation, the probability of the field values appearing outside the truncation window is controlled by the Chebyshev inequality for all possible probability distributions. For all sites, the total probability outside the truncation window $p_{\text{total}}$ is controlled by the union bound $\mathcal{V} p_{\text{single}}$, where $\mathcal{V}$ is the total number of sites and $p_{\text{single}}$ is the probability of making error for a single site.

Here, we just quote the result from the Jordan-Lee-Preskill bound without proving it. Say that we truncate the field to obtain the state $\ket{\psi_{\text{JLP}}}$ and we introduce the error $\epsilon_{\text{JLP}}$ such that $\left\langle {{\psi _{{\rm{true}}}}|{\psi _{{\rm{JLP}}}}} \right\rangle  = 1 - {\epsilon _{{\rm{JLP}}}}$, where $\psi_{\text{true}}$ is the actual state. We have
\begin{align}
{\zeta _{\max }}\sim \sqrt {\frac{\mathcal{V}}{{{\epsilon _{{\rm{JLP }}}}}}} \sqrt {{{\langle {\psi _{{\rm{true}}}}|\zeta^2 |{\psi _{{\rm{true}}}}\rangle }}}~.
\end{align}
The square root of the prefactor comes from the quadratic relation in the Chebyshev inequality.

Then, how could we bound the precision? It is easy to notice that from the definition of the canonical commutation relation we have
\begin{align}
{\pi _{\max,\zeta }}\sim \frac{1}{{{b^3}\delta \zeta }}~.
\end{align}
Applying the same Jordan-Lee-Preskill bound towards the field momentum, we get
\begin{align}
&\delta \zeta \sim \frac{1}{{{b^3}}}\sqrt {\frac{{{\epsilon _{{\rm{JLP }}}}}}{\mathcal{V}}} \frac{1}{{\sqrt {\langle {\psi _{{\rm{true}}}}|\pi _\zeta ^2|{\psi _{{\rm{true}}}}\rangle } }}~,\nonumber\\
&{n_b}\sim \log \left( {\frac{{\mathcal{V}{b^3}}}{{{\epsilon _{{\rm{JLP }}}}}}\sqrt {\langle {\psi _{{\rm{true}}}}|\pi _\zeta ^2|{\psi _{{\rm{true}}}}\rangle \langle {\psi _{{\rm{true}}}}|{\zeta ^2}|{\psi _{{\rm{true}}}}\rangle } } \right)~.
\end{align}

So how could we bound $\left\langle {{\zeta ^2}} \right\rangle $ and  $\left\langle {{\pi_{\zeta} ^2}} \right\rangle $? Now, we need some knowledge about cosmology. We start from $\pi_\zeta$. Since we already know that our ultraviolet cutoff is $1/b \sim \sqrt{H \epsilon} $, the cutoff must work for a single term in the Hamiltonian
\begin{align}
\frac{{{b^3}c_s^2}}{{{a_0^2}\epsilon }}\left\langle {\pi _\zeta ^2} \right\rangle  \lesssim \frac{1}{b}~.
\end{align}
So we get
\begin{align}
\left\langle {\pi _\zeta ^2} \right\rangle  \lesssim \frac{{{a_0^2}\epsilon }}{{{b^4}c_s^2}}~.
\end{align}
Note that here $a_0$ is the initial scale factor. The way of bounding quantities using the total energy is the same as what we did for the original Jordan-Lee-Preskill scattering experiment in the flat space. Furthermore, the bound could serve as a general bound working for all couplings.

However, how could we bound $\zeta$? At the late time $\tau_{\text{end}}$, we know experimentally that
\begin{align}
\left\langle {{\zeta ^2}} \right\rangle  \sim {10^{ - 10}}~,
\end{align}
which is purely from the experiment. In the early time, the situation may not be the same, and we could not naively use the bound from the late time. In the free theory, the inflaton is massless so we cannot bound $\zeta^2$ directly from the Hamiltonian. However, we could directly estimate the curvature perturbation from the free theory. We have
\begin{align}
&\left| {\left\langle {\zeta ({\tau _0},{\bf{x}})\zeta ({\tau _0},{\bf{y}})} \right\rangle } \right| = \left| {\int {\frac{{{d^3}k}}{{{{(2\pi )}^3}}}} {v_k}v_k^*{e^{i{\bf{k}} \cdot ({\bf{x}} - {\bf{y}})}}} \right| \le \left| {\int {\frac{{{d^3}k}}{{{{(2\pi )}^3}}}} {v_k}v_k^*} \right|\nonumber\\
&= \left| {\frac{{4\pi }}{{{{(2\pi )}^3}}}\int {{k^2}} {v_k}v_k^*dk} \right| \sim \left| {\frac{{4\pi }}{{{{(2\pi )}^3}}}\int_{{k_{{\rm{IR}}}}}^{{k_{{\rm{UV}}}}} {{v_k}v_k^*dk} } \right|\nonumber\\
&\sim \frac{{{H^2}\left( {c_s^2\tau _0^2\left( {k_{{\rm{UV}}}^2 - k_{{\rm{IR}}}^2} \right) + 2\log \left( {\frac{{{k_{{\rm{UV}}}}}}{{{k_{{\rm{IR}}}}}}} \right)} \right)}}{{16{\pi ^2}{c_s}\epsilon }} \sim \frac{{{H^2}\left( {\frac{{c_s^2\tau _0^2}}{{{b^2}}} + 2\log \left( {\hat L} \right)} \right)}}{{16{\pi ^2}{c_s}\epsilon }}~.
\end{align}
In the last step, we take the cutoff $k_\text{UV}=1/b$ and $k_\text{IR}=1/L$. One could see that the above result depends at most logarithmically on the system size. Thus, the dimension of local Hilbert space should be at most polynomial in size in general.

The above result is only a result of the free theory. But how about interacting theory at the time $\tau_0$? In general, the result should not be drastically changed if we are in the regime of the perturbation theory. The leading correction towards the above two-point function should be the one-loop diagram in the following plot, which is of order $\varepsilon^2$ if we call the coupling as $\varepsilon$\footnote{$\varepsilon$ is a combination of $\Sigma$, $\lambda$ and $1-c_s$, which we discuss before.}. (see Figure \ref{oneloop} for an illustration.) But the situation might change in the case of strong coupling. Now let us consider the system is approaching the critical point with a second-order phase transition. If such a critical point exists, the two-point function of curvature perturbation should scale as a power law with the distance and a scaling dimension, which is not a drastic dependence for our quantum computer. But what happens in general, in the middle of the renormalization group flow? Although it seems to be not very possible that the field fluctuation is exponential regarding the system size, since it is a non-perturbative problem, we could only make trials numerically if we do not have any theoretical control. In fact, assuming the field configuration is continuous, when constructing the state and measuring the field profile, we could actually get some indications if the size of the local Hilbert space is out of reach. Such trials will be helpful for determining an honest value of the field range up to some given error, with certain convergence conditions. We leave this topic for future research, especially for people with quantum devices and clean qubits.

\begin{figure}[t]
  \centering
  \includegraphics[width=0.5\textwidth]{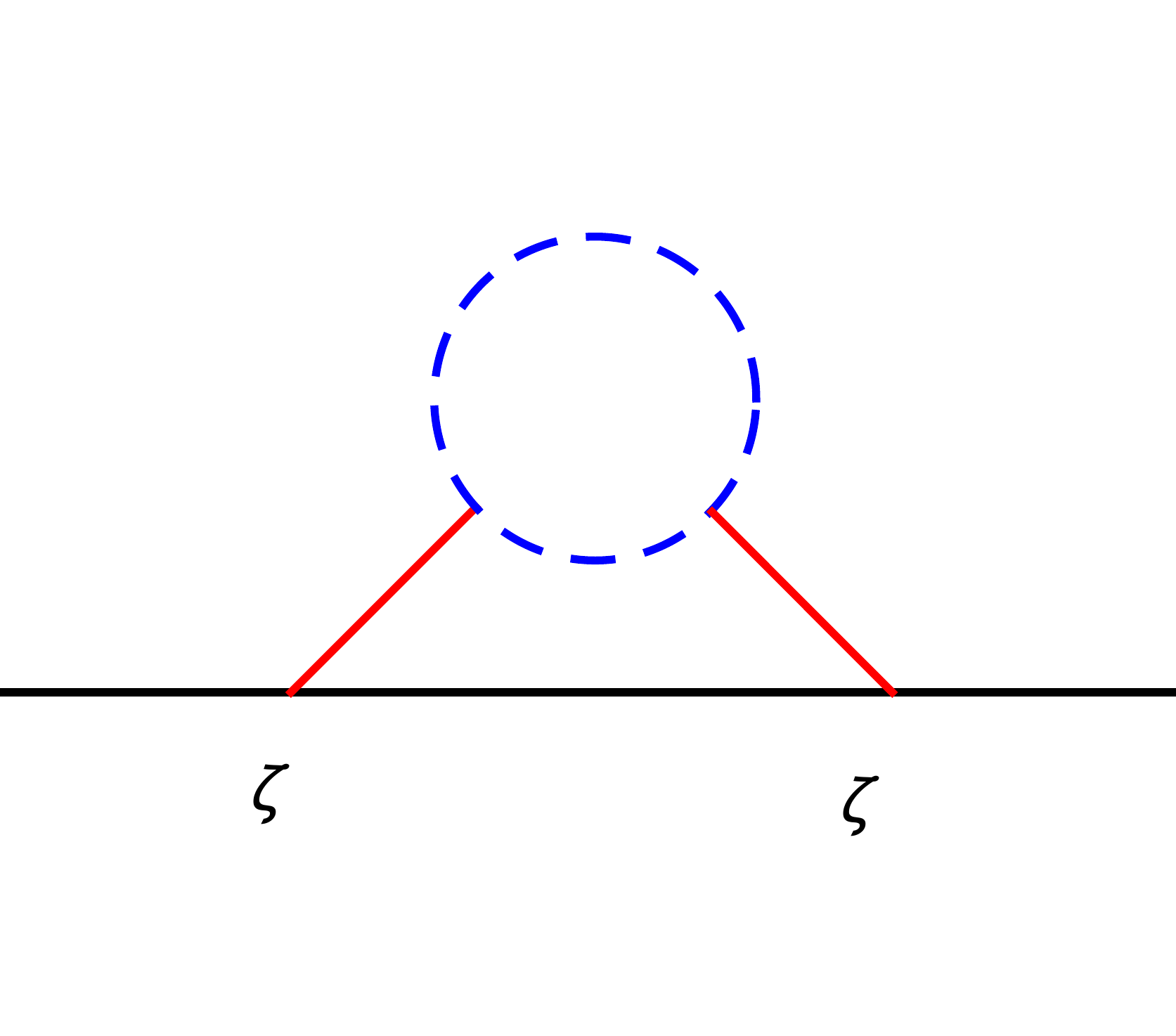}
  \caption{\label{oneloop} The one-loop diagram at the fixed time $\tau_0$. }
\end{figure}

\subsection{Initial state preparation by Kitaev and Webb}
After we address the encoding problem, here we discuss the initial state preparation. At the beginning, we wish to construct $\ket{\Omega_\text{free} (\tau_0)}$ in the quantum computer. In the free theory, the wave function is given by the Gaussian distribution in the field basis, with the probability distribution
\begin{align}
p(\vec \zeta ) = \frac{1}{{{{(2\pi )}^{{\cal V}/2}}|{\bf{M}}{|^{1/2}}}}\exp \left( { - \frac{1}{2}\vec \zeta  \cdot {{\bf{M}}^{ - 1}} \cdot \vec \zeta } \right)~.
\end{align}
Here, we define
\begin{align}
\vec{\zeta} = \left(\zeta\left(\mathbf{x}_{1}\right), \zeta\left(\mathbf{x}_{2}\right), \cdots, \zeta\left(\mathbf{x}_{\mathcal{V}}\right)\right)~,
\end{align}
and the matrix $\mathbf{M}$ is the two-point function
\begin{align}
{\mathbf{M}_{ij}} = \left\langle {\zeta ({\mathbf{x}_i})\zeta ({\mathbf{x}_j})} \right\rangle = G(\tau_0, \textbf{x}_i; \tau_0, \textbf{x}_j)~.
\end{align}
The square root of the probability distribution could define the components on the field basis. Thus, the problem of state preparation becomes a problem of preparing the Gaussian distribution with multiple variables.

This problem is discussed and solved in \cite{Jordan:2011ne,Jordan:2011ci}, and here we describe the solution. We could directly use the Kitaev-Webb algorithm \cite{kitaev2008wavefunction} to prepare the vacuum state. The idea is that one could firstly prepare a distribution of Gaussian state in a diagonal form, and then do a transformation to the desired basis. The main time cost in the algorithm is the singular value decomposition of the inverse covariance matrix we have in the Gaussian distribution, which could be improved by using classical algorithms in \cite{bunch1974triangular,coppersmith1987matrix}. With the known covariance matrix given by the two-point function, the complexity scales as $O(\mathcal{V}^{2.376})$, which is bounded by polynomials in system size.

There might exist some alternative methods for constructing Gaussian states. For instance, \cite{somma2015quantum} describes another algorithm for Gaussian state preparation, which is related to one-dimensional quantum systems. \cite{macridin2018digital} describes another variational algorithm for preparing the Gaussian state. There might be some future improvements about the Kitaev-Webb algorithm in 3+1 dimensions.

\subsection{Trotter simulation}
\label{Trotter}
Now, say that we already have a state $\ket{\Omega_{\text{free}}(\tau_0)}$. The next steps are to construct the interacting vacuum $\ket{\Omega_{\text{in}}(\tau_0)}$ and then evolve the Heisenberg unitary operator. Both steps are requiring Trotter simulation based on the product formula. Comparing to the flat space, the task we have here in the inflationary spacetime is pretty different in the following two aspects,
\begin{itemize}
\item Our time-dependent quantum field theory is different from the flat space by a scale factor. The scale factor will affect Trotter simulation errors by entering the commutators.
\item In the original Jordan-Lee-Preskill algorithm, the $\lambda \phi^4$ quantum field theory Hamiltonian could be split into two parts, $H_{\phi}$ and $H_{\pi}$ that contain fields and field momenta separately. However, in our case, there is a generic feature that we will have mixing terms between fields and field momenta.
\end{itemize}
The above differences motivate us to describe a generalized, time-dependent Trotter simulation theory, which will be presented in this subsection. We will mostly use the notation from \cite{Jordan:2011ne,Jordan:2011ci} and \cite{childs2019theory}.

We consider a general time-dependent time evolution operator, and we split the total time into $n_{\text{product}}$ intervals
\begin{align}
\mathbb{T}\exp \left( { - i\int_{{\tau _0}}^\tau  d \tau 'H(\tau ')} \right) \approx \prod\limits_{j = 1}^{{n_{{\rm{product}}}}} {\exp } \left( { - iH\left( {\frac{{j - 1}}{{{n_{{\rm{product}}}}}}(\tau  - {\tau _0})} \right)\frac{{\tau  - {\tau _0}}}{{{n_{{\rm{product}}}}}}} \right)~.
\end{align}
We could define a short-hand notation
\begin{align}
\frac{{j - 1}}{{{n_{{\rm{product }}}}}}\left( {\tau  - {\tau _0}} \right) = {\tau _j}~.
\end{align}
For each term in this exponential, we have the $k$-th order product formula
\begin{align}
\exp \left( { - iH\left( {{\tau _j}} \right)\frac{{\tau  - {\tau _0}}}{{{n_{{\rm{product}}}}}}} \right) = {\text{product term}} + O\left( {{\alpha _{{\mathop{\rm com}\nolimits} }}\frac{{{{(\tau  - {\tau _0})}^{2k + 1}}}}{{n_{{\rm{product}}}^{2k + 1}}}} \right)~.
\end{align}
The detailed expression in various forms is given in \cite{childs2019theory} in detail. The error constant $\alpha_{\text{com}}$ is determined by commutators
\begin{align}
{\alpha _{{\rm{com}}}} = \sum\limits_{{\beta _1}, \cdots ,{\beta _{2k + 1}} = 1}^{{n_{{\rm{split}}}}} {\left\| {\left[ {{H_{{\beta _{2k + 1}}}}, \cdots ,\left[ {{H_{{\beta _2}}},{H_{{\beta _1}}}} \right]} \right]} \right\|}~,
\end{align}
where we expand the total Hamiltonian by
\begin{align}
H = \sum\limits_{\beta  = 1}^{{n_{{\rm{split}}}}} {{H_\beta }}~.
\end{align}
In our case, we assume $n_{\text{split}}=3$. Physically, in our situation, we have the field variable term, the field momentum term, and the mixing term, where we call them $H_1$, $H_2$ and $H_3$.
\subsection{Adiabatic state preparation and the ground state projection}
We start from the adiabatic state preparation with applications of our Trotter simulation formula.

Generically, for a generic Hamiltonian $H(s)$ which is parametrized by $s$ ranging in $[0,1]$, we call the $\ell$-th eigenstate $\ket{v_\ell(s)}$ for a given $s$,
\begin{align}
H(s)\left| {{v_\ell }(s)} \right\rangle  = {e_\ell }(s)\left| {{v_\ell }(s)} \right\rangle~.
\end{align}
The state $\ket{v}$ could be approximately achieved by the $s$-dependent time evolution from 0 to 1, starting from $s=0$. We call such state $\ket{u_k}$ for the energy level $e_k$. The transition amplitude, where we call it the ``adiabatic error" $\epsilon_{\text{ad}}$, is given by
\begin{align}
{\epsilon _{\text{ad}}} \equiv \left| {\left\langle {{v_\ell }|{u_k}} \right\rangle } \right|\sim \left| {\frac{1}{{T {{\left( {{e_k} - {e_\ell }} \right)}^2}}}\frac{{dH(s)}}{{ds}}} \right|~.
\end{align}
Here, $T$ is the total time we use when we are turning on the interaction. In our situation, we wish to turn on the interaction slowly and linearly. So we have
\begin{align}
\frac{{dH(s)}}{{ds}} = {H_I}~,
\end{align}
and the adiabatic error is bounded by
\begin{align}
{\epsilon _{\text{ad}}} \lesssim  \frac{{\left\| {{H_I}} \right\|}}{{T \times {\rm{ga}}{{\rm{p}}^2}}}~.
\end{align}
Here, we are starting from the vacuum state of the free theory, so we are interpreting the energy difference as the mass gap in the Heisenberg Hamiltonian. As we have mentioned before, our Hamiltonian has a further problem: the theory is massless. We cannot only use energy to distinguish the vacuum state and zero-momentum states with multiple diagonal modes. Thus, some extra treatments should be used when we are doing the adiabatic state preparation.

Here, we propose the following ground state projection algorithm. We notice that the degeneracy problem among diagonal modes and the vacuum should only happen in free theory. Generically, we don't expect that the transition amplitude will be large for finite coupling constant. Namely, there should exist energy-level splitting when we are slowly turning on the coupling. Thus, we only need to resolve the tunneling around the free theory regime. Note that in the free theory, diagonal modes should carry positive particle numbers. Thus, starting from the free theory vacuum, we keep doing the measurement of the state in the quantum computer by measuring the following operator
\begin{align}
{N_{{\bf{k}} = 0}} = b_{{\bf{k}} = 0}^\dag {b_{{\bf{k}} = 0}}~.
\end{align}
We will only select the result when we get zero from the measurement data. The operator ${N_{{\bf{k}} = 0}} $ is exactly the number operator for zero-momentum diagonal modes. We need to perform the measurement for the first few steps in our simulation. After we obtain a significant amount of the coupling, we don't need to do this measurement, and around the finite coupling regime, the number operator has no physical meaning.

Using the above ground state projection protocol, we could bound our adiabatic error
\begin{align}
{\epsilon _{{\rm{ad}}}} \lesssim  \frac{{\left\| {{H_I}} \right\|{L^2}}}{T}~.
\end{align}
Note that the $1/L$ gap introduces extra polynomial factors in the system size. The explicit determination of the evolution time $T$ should come from the Trotter computation. In fact, we find
\begin{align}
{\epsilon _{{\rm{ad}}}} \lesssim {\varepsilon} \times {\left( {\frac{{\cal V}}{{{\epsilon _{{\rm{JLP}}}}}}} \right)^{2k + 3/2}}\frac{{\abs{\tau _0}^2{T^{2k + 1}}}}{{{n_{\text{ad}}^{2k}}}}\frac{1}{{{\epsilon ^{3/2}}}}~.
\end{align}
Here $n_{\text{ad}}$ is the number of splitting we use during the adiabatic process. Thus, the total gate estimate, ${n_{{\rm{ad,total}}}}$, scales as
\begin{align}
&{n_{{\rm{ad,total}}}}\sim {\cal V} \times {n_{{\rm{ad}}}}\nonumber\\
&\sim { O}\left( {{\varepsilon}{^{1/2k}} \times \frac{{{{\cal V}^{2 + 3/4k}}}}{{\epsilon _{{\rm{JLP}}}^{1 + 3/4k}}} \times {T^{1 + 1/2k}} \times {{\left( {\frac{{\abs{\tau _0}^2}}{{{\epsilon ^{3/2}}}}} \right)}^{1/2k}}} \right)~.
\end{align}
Note that here, we make use of the scaling of our basis norm for $\zeta$ and $\pi_{\zeta}$.

Finally, we comment briefly on the issue of the inflaton mass. Generically, the effective field theory of inflation might receive higher derivative corrections suppressed by the cutoff. If the inflaton is not protected, it will receive radiative corrections. A connection between the unstable inflaton mass and the parameter $\eta$ in slow roll is so-called the $\eta$-problem, suggesting that the radiative correction of the inflaton might prevent the inflationary expansion of the spacetime. In our case, we consider the lattice regularization of the theory by the ultraviolet cutoff $1/b$, and the theory in the short distance is still massless in the sense of lattice many-body systems. In a sense, we are regularizing the theory in the way that the inflaton mass is protected from being zero in the free theory case. Thus, we are away from the $\eta$-problem. However, when we turn on the interaction, the inflaton mass will receive corrections from the coupling, and the theory is generically gapped (although the gap might be very small).

\subsection{The efficiency of inflation}
Now say that we already obtain $\ket{\Omega_{\text{in}}(\tau)}$ based on the above adiabatic state preparation procedure. The next step is to simulate the following unitary operation acting on the state
\begin{align}
{U_H}\left( {\tau_{\text{end}} ,{\tau _0}} \right)\left| {{\Omega _{{\rm{in}}}}\left( {{\tau _0}} \right)} \right\rangle  = \mathbb{T}\exp \left[ { - i\int_{{\tau _0}}^{\tau_{\text{end}}}  H \left( {\tau '} \right)d\tau '} \right]\left| {{\Omega _{{\rm{in}}}}\left( {{\tau _0}} \right)} \right\rangle~,
\end{align}
using our Heisenberg Hamiltonian $H$. Again, we will use the Trotter formula to simulate the above computation. Note that now the Hamiltonian is honestly time-dependent. It is depending on the conformal time, which is different from the adiabatic state preparation case where we are slowly turning on the coupling.

Here, we could estimate the efficiency of the Trotter simulation during the expansion of the scale factor. Assuming the Trotter error $\epsilon_{\text{inflation}}$ and the number of time steps $n_{\text{inflation}}$, we have,
\begin{align}
{\epsilon _{{\rm{inflation}}}} \sim O\left( {\sum\limits_{j = 1}^{{n_{{\rm{inflation}}}}} {{\alpha _{{\rm{com }}}}} \left( {{\tau _j}} \right)\frac{{{{({\tau _{{\rm{end}}}} - {\tau _0})}^{2k + 1}}}}{{n_{{\rm{inflation}}}^{2k + 1}}}} \right) \le O\left( {{\alpha _{{\rm{com }}}}({\tau _{{\rm{end}}}})\frac{{{{({\tau _{{\rm{end}}}} - {\tau _0})}^{2k + 1}}}}{{n_{{\rm{inflation}}}^{2k}}}} \right)~.
\end{align}
Here, we use a bound about the time dependence, assuming the dominance of the late-time Hamiltonian because of the expansion of the scale factor.

Rigorously computing the error ${\epsilon _{{\rm{inflation}}}} $ is a hard problem. Here, we will make an intuitive analysis based on the time dependence. In fact, we expect that the time dependence of our three Hamiltonians $H_1$, $H_2$ and $H_3$ at the time $\tau$ will be bounded by
\begin{align}
\left\|  H_{1} \right\| \lesssim O(\frac{\mathcal{V}}{\epsilon_{\mathrm{JLP}}})~,~~~~~\left\|  H_{2} \right\| \lesssim  O(\frac{\mathcal{V}}{\epsilon_{\mathrm{JLP}}} \frac{1}{\tau^{2}})~,~~~~~ \quad \left\|  H_{3} \right\| \lesssim O\left(\frac{\mathcal{V}}{\epsilon_{\mathrm{JLP}}}\right)^{3/2}~.
\end{align}
Here we make a brief explanation on the above bound. The term $\mathcal{V}/{\epsilon_{\mathrm{JLP}}}$ comes from the bound of the field range of the Chebyshev inequality. Thus, the pure cubic term $H_3$ will scale with a power 3/2. The time dependence is purely coming from the counting of the geometric factor and the time dependence on the quantum fields. A remarkable feature is the $1/\tau$ dependence appearing in $H_2$. This is, in fact, coming from the cutoff logarithmic term that is independent of time in the solution of modes. We expect that this term will also be presented in the interacting theory in the short distance limit, especially since we are working under an exact lattice regularization of the quantum field theory. Thus, we expect that the Trotter constant will dominate at late time.

Now assuming the late-time dominance, an example of the dominant piece in the Trotter formula will scale in the following form
\begin{align}
\left[ {{H_1},\left[ {{H_2},\left[ {{H_1},\left[ {{H_2}, \cdots ,\left[ {{H_1},\left[ {{H_3},{H_2}} \right]} \right]} \right]} \right]} \right]} \right]\sim {\left( {\frac{{\cal V}}{{{\epsilon _{{\rm{JLP}}}}}}} \right)^{2k + \frac{3}{2}}}{\tau ^{ - 2k}}~.
\end{align}
Furthermore, the commutator will bring us an extra factor in $\mathcal{V}$,
\begin{align}
\left[ {{H_i},{H_j}} \right] = \sum\limits_{\bf{x}} {\sum\limits_{\bf{y}} {{\delta _{{\bf{x}},{\bf{y}}}}} }  \cdots  = \sum\limits_{\bf{x}}  \cdots~.
\end{align}
Thus, we have an estimate on the Trotter error,
\begin{align}
{\epsilon _{{\rm{inflation}}}}\sim O\left(\frac{{{{\cal V}^{2k + 5/2}}}}{{\epsilon _{{\rm{JLP}}}^{2k + 3/2}}}\frac{{{{\left( {{\tau _{{\rm{end}}}} - {\tau _0}} \right)}^{2k + 1}}}}{{{{(\left| {{\tau _{{\rm{end}}}}} \right|{n_{{\rm{inflation}}}})}^{2k}}}}\right)\sim O\left(\frac{{{{\cal V}^{2k + 5/2}}}}{{\epsilon _{{\rm{JLP}}}^{2k + 3/2}}}\frac{{\left| {{\tau _0}} \right|{{\left( {1 - {e^{ - N}}} \right)}^{2k + 1}}}}{{{{({e^{ - N}}{n_{{\rm{inflation}}}})}^{2k}}}}\right)~.
\end{align}
Thus, the total gate counting is given by
\begin{align}
{n_{{\rm{inflation,total}}}}\sim O\left( {\frac{{{{\left( {{\tau _{{\rm{end}}}} - {\tau _0}} \right)}^{1 + 1/2k}}{{\cal V}^{2 + 5/4k}}}}{{\left| {{\tau _{{\rm{end}}}}} \right|\epsilon _{{\rm{JLP}}}^{1 + 3/2k}}}} \right)\sim O\left( {\frac{{{{\left| {{\tau _0}} \right|}^{1/2k}}{{\left( {1 - {e^{ - N}}} \right)}^{1 + 1/2k}}{{\cal V}^{2 + 5/4k}}}}{{{e^{ - N}}\epsilon _{{\rm{JLP}}}^{1 + 3/2k}}}} \right)~.\label{gates of Trotter}
\end{align}
We will leave two comments here for the above formula
\begin{itemize}
\item The above formula is a clear example on how the e-folding number $N$ changes the efficiency of the Trotter simulation, and how the expansion history depending on $\tau$ will change the number of resources we demand. The calculation of Trotter constants might also be important for other situations as well, especially for other time-dependent quantum field theory or quantum many-body problems.
\item We have to admit that the above analysis is a rough estimate based on certain assumptions. Firstly, it assumes the late-time dominance of the Hamiltonian norms and drops the dependence on the coupling constant. Secondly, we make assumptions that there are certain terms dominating the Trotter series based on the late-time dominance assumption. A more careful analysis is useful to fully characterize the Trotter error in the inflation process, where we wish to leave those calculations for future research.
\end{itemize}

\subsection{Measurement}
Here we briefly discuss the issue about measurement. Say that we could already construct the state
\begin{align}
\ket{\text{result}}=U_{H}\left(\tau_{\mathrm{end}}, \tau_{0}\right)\left|\Omega_{\mathrm{in}}\left(\tau_{0}\right)\right\rangle~.
\end{align}
Then, measuring the expectation value
\begin{align}
\left\langle {\text{result}} \right|{\cal O}\left( {{\tau _0}} \right)\left| {\text{result}} \right\rangle~,
\end{align}
is a standard problem in quantum computation. Since we already know the operator $\mathcal{O}(\tau_0)$, we could do a probabilistic calculation using the post-selection experiment. The method is statistical, and the total cost we need to perform should scale at most polynomially in $1/\text{error}$.

We wish to make the following comments.
\begin{itemize}
\item It might be interesting to look at operators beyond simply curvature perturbations and non-Gaussianities. For instance, we might consider measuring tensor perturbations, other operators like energies or stress tensors in inflationary perturbation theory, or studying the nature of the interacting vacuum by checking the assumptions in the in-in formalism calculation in perturbative field theory.
\item The algorithm about the post-selection that we have mentioned here is statistical and probabilistic. One might consider some other algorithms which are deterministic. For instance, we could consider the quantum signal processing algorithm discussed in \cite{low2017optimal}. Those algorithms might demand extra effort about block encoding and qubitization, and demand some oracles to perform selections. We leave those developments applying to cosmology in future research.
\end{itemize}
\subsection{Errors, de Sitter trees, and loops}
Here we make a brief analysis of the error induced from the finite lattice to approximate the continuum.

From the standard quantum field theory argument, the lattice effect will result in equivalently a series of irrelevant operators in the Lagrangian, which are suppressed by the lattice spacing. The leading irrelevant operators preserving shift symmetry are at least dimension six, so the leading corrections should be $\dot{\zeta}^3$ or $\dot{\zeta} (\partial_i \zeta)^2$. Thus, the lattice effect will induce an error in non-Gaussianity scaling as $O(b^2)$. Those corrections are at tree level. Note that, those corrections are nothing but adding a small perturbative piece of couplings into the original interaction (which could be non-perturbative in principle), which completely makes sense because we are assuming the interaction of the general single-field inflation.

Corrections from a finite $b$ towards the two-point function should appear at the one-loop order, scaling as $b^4$. The analysis and the diagram are the same as the case in \ref{oneloop}, although the motivation is completely different from the previous discussion. Now, taking the interaction $\dot{\zeta}^3$ as an example, we could borrow the result from \cite{Senatore:2009cf} about this one-loop diagram using the language of the effective field theory of inflation. Using our notation, we start from the interaction
\begin{align}
S = {b^2}{\int {{d^3}xdt \times \kappa {a^3}\dot \zeta } ^3}~.
\end{align}
Here $\kappa$ is a dimensionful coupling. The result of the one-loop correction to the two-point function of the curvature perturbation in the momentum space should scale as
\begin{align}
{\text{one loop correction}} \sim \frac{{{b^4}{\kappa ^2}{H^{14}}}}{{{{\dot H}^4}}}\frac{1}{{{k^3}}} \times \log \left( {Hb} \right)~.
\end{align}
The paper \cite{Senatore:2009cf} claims that the result is consistent from both the cutoff and the dimensional regularization. In principle, one could use those one-loop formulas to control the systematic error produced by lattice regularization.

\newpage
\section{Final remarks}\label{remark}
In this paper, we present a complete analysis of the generalized Jordan-Lee-Preskill algorithm for inflationary spacetime. The algorithm contains the encoding analysis, initial state preparation, adiabatic state preparation, time evolution during inflation, and the measurement. We compute the time cost of the algorithm and argue that the complexity is polynomial in system size for quantum devices, sharpening the statement of the quantum-extended Church-Turing Thesis. The analysis includes various techniques from high energy physics to quantum information. We also make suggestions on the physical questions we could answer when running the algorithm in the quantum computer.

In this section, we will make further suggestions and comments on the future directions related to this paper.

\subsection{Improvement of the algorithm}
About the algorithm itself, there are still rooms for future improvement. We point out some of them as the following.
\begin{itemize}
\item About the encoding treatment, in this paper, we are working on the field basis at the initial time $\tau_0$ and using the HKLL-type formula to encode the field and the field momentum at a general time $\tau$. However, it might be useful to explore the possibility of other bases. For instance, one could consider encoding the field from $\tau_{\text{end}}$ and evolving it back. If one could realize the treatment, it might be more convenient to bound the field fluctuations since we have the late-time experimental data about the curvature perturbation. Furthermore, one could also consider encode the Hamiltonian in the momentum basis or in the harmonic oscillator basis, as we have discussed before.
\item About the Jordan-Lee-Preskill bound from EFT, in this paper, we make some analysis of the curvature perturbation and the field momentum perturbation using some knowledge from cosmology. It might be better to improve those bounds by taking more physics into considerations. For instance, in some specific models, those fluctuations might have better theoretical control. Moreover, it might be interesting to explore the phase diagram of couplings we have mentioned, in the general single-field inflation or some other models.
\item About the adiabatic state preparation, we discuss an algorithm here to filter out the degeneracy of zero-momentum diagonal modes, projecting towards the ground state. It might be interesting to make a more detailed analysis of how the error would scale during this process, both in the case of the free theory and the interacting theory with small couplings. It will also be helpful if we could explore the mass gap depending on couplings analytically or numerically. Those research might demand some help from the near-term quantum device. It might also be helpful to try some other algorithms in the adiabatic state preparation process, for instance, the Wan-Kim algorithm proposed recently in \cite{wan2020fast}.
\item About the Trotter evolution, in this paper, we mainly discuss some estimations, mostly focusing on the time dependence of the Hamiltonian. It might be helpful to make a more detailed analysis of all other parameters, for instance, the value of couplings, and some exact computations on the Trotter error numerically or analytically.
\item About the measurement and the representation of the Hamiltonian, it might be helpful to discuss some improvements of the algorithm using methods with oracle constructions, for instance, quantum signal processing. Those might be associated with other methods to represent the Hamiltonian, for instance, qubitization. See \cite{mcardle2020quantum} for a review about them.
\end{itemize}

\subsection{The quantum Church-Turing Thesis}
Here we make some comments about the quantum-extended Church-Turing Thesis. The quantum-extended Church-Turing Thesis is a claim about the capacity of quantum computation to simulate our real world. It claims that all physical process happening in the real universe, could be computed efficiently using the model of the quantum Turing machine or the quantum circuit. Namely, the thesis is claiming about the efficiency of simulation. As we all know, the cost scaling polynomially with the system size will satisfy computer scientists, while the cost scaling exponentially with the system size will be disappointing.

However, it might be challenging when we suggest an implementation of the quantum-extended Church-Turing Thesis in general relativity. As we all know, the theory of general relativity has the ambiguity of defining space and time, but quantum mechanics or quantum computation require the definition of time to be specific. More precisely, if the time cost represented in one coordinate suggests that the complexity is polynomial, it might be exponential in some other time coordinates. So the statement of the quantum-extended Church-Turing Thesis seems requiring a covariant definition if we believe that the low energy effective description of general relativity and the fundamental existence of time.

In general, the quantum-extended Church-Turing Thesis is widely believed to be correct, although we cannot really prove the statement. It will be very interesting if someone could prove it in some general assumptions, and it will also be interesting if someone could find some cases where the claim is violated. Furthermore, we could imagine that the quantum-extended Church-Turing Thesis could even serve as a swampland criterion in the landscape of theories: if we find in some models, the quantum-extended Church-Turing Thesis is violated, then the model might be unphysical or the operation corresponding to the time cost might be not allowed. Otherwise, we may accept that the model of the quantum Turing machine may not be universal and powerful enough, and we might consider using some other physical process to perform quantum computation.

Our analysis of Trotter simulation could potentially provide evidence for interpreting the quantum-extended Church-Turing Thesis as a swampland criterion, in terms of a potential connection with the Distance Conjecture (DC) \cite{Ooguri:2006in} and the Trans-Planckian Censorship Conjecture (TCC) \cite{Bedroya:2019snp}, which are known as swampland criteria in the literature. More precisely, according to eq.~\eqref{gates of Trotter}, the resource demand by the Trotter simulation of inflation approximately grows exponentially in e-folding number, indicating the failure of the quantum-extended Church-Turing Thesis for a large enough e-folding number. Indeed, both the DC and the TCC can be written in terms of the bounds on the e-folding number, see, e.g., \cite{Bedroya:2019snp,Agrawal:2018own}
\bea
&& \text{DC:}\qquad \Delta\phi\sim \sqrt{2\epsilon}N <\mathcal{O}(1)\,,
\cr &&
\cr && \text{TCC:}\qquad e^N<\fft{1}{H_{\rm end}}\,.
\eea

Current research suggests further issues about the quantum-extended Church-Turing Thesis in gravity. First, general relativity could seemingly bring us the ambiguity about the efficiency, but also the computability of the Turing machine in classical or quantum computation. A typical experiment of this type is named by Malament and Hogarth \cite{hogarth1996predictability}, suggesting that general relativity is able to solve the Halting Problem (see some comments in Appendix \ref{MH}, and there is also a related discussion about computability and the gravitational path integral \cite{geroch1986computability}). Furthermore, there are recent discussions about the quantum-extended Church-Turing and the holographic correspondence: we know that the entanglement entropy and the complexity in the boundary CFT should be generically computationally hard to evaluate, but in the dual side, the area and the volume in gravity seem to be computationally easy. The puzzle brings other further questions about the nature of holographic correspondence and the quantum-extended Church-Turing Thesis in quantum gravity (see \cite{Bouland:2019pvu,Gheorghiu:2020sko,Susskind:2020kti,Kim:2020cds,Yoshida:2020wpd}).

Here, we wish to point out that our calculations might provide a concrete framework to address the quantum-extended Church-Turing Thesis in the curved spacetime. In our work, we argue that the time cost should be polynomial when we are using the conformal coordinate. But how about the physical time coordinate? How about mixing space and time? One could try to proceed with a similar analysis following our paper. One could also try to generalize our discussions in other spacetime, like quantum field theories in pure AdS or black hole backgrounds.

\subsection{The role of constants in the Trotter algorithm}
Here we wish to point out another perspective in our Trotter simulation. Usually, the Trotter constant appearing in the commutator terms may not be that important, especially when the Hamiltonian is time-independent. However, the constant $\alpha_{\text{com}}$ appearing in our calculation is time-dependent during inflation. Here, the constant has clear physical meanings, carrying the geometric factor for the inflationary spacetime.

Here we wish to point out that the Trotter constant might be interesting physically in other circumstances, and it might be interesting to point out the physical meaning of the product formula. For instance, we know that the chaotic Hamiltonian might be harder to simulate comparing to the integrable Hamiltonian based on the intuition of the butterfly effect. Thus, the hardness of quantum simulation based on chaos might be shown in the Trotter commutators. If one could explicitly show some quantitative connections between the commutators in the Trotter formula and the out-of-time-ordered correlators, it will be interesting for a possible physical understanding of our quantum simulation algorithms. This might also be related to the recent discussions about improvements of the Trotter simulation algorithm using randomness techniques (see \cite{campbell2019random,sieberer2019digital,chen2020quantum,Ashwork}).

\subsection{Verifying statements in quantum field theories}
In our paper, we point out another issue regarding the value of quantum simulation. As we know, quantum field theories are hard to study, especially when the theories are strongly coupled. Sometimes, the arguments we made are physically intuitive but may not be completely rigorous that could satisfy mathematicians and computer scientists. Thus, we claim that quantum computation, if making our field theory problems exactly verifiable and numerically achievable, might provide us new venues to test and clarify our physical claims in quantum field theory. We believe that our example of the in-in formalism and the interacting vacua is not unique. The future computational power provided by quantum devices might force us to make everything clear and bound the error we make about our physical arguments. For instance, some claims in quantum field theories about duality, from AdS/CFT to duality webs applied to high energy physics and condensed-matter theory recently, might be checked and verified in some similar ways.

\subsection{Quantum gravity in the lab: analog and near-term simulation}
Our paper could be regarded as simulating a certain category of quantum gravitational theory in the digital quantum computer. There are some similar works recently about the quantum simulation of other perspectives of quantum gravitational theories, mostly related to AdS/CFT and black hole physics. For instance, the celebrated Sachdev-Ye-Kitaev (SYK) model and some related quantum gravity effects are considered to be simulated in the analog platforms \cite{Danshita:2016xbo,Landsman:2018jpm,Brown:2019hmk,Kobrin:2020xms} (see also digital simulations in \cite{Garcia-Alvarez:2016wem,Babbush:2018mlj,Xu:2020shn}, and also some early papers about analog simulation of cosmic inflation \cite{Fedichev:2003bv,Fischer:2004bf}). Thus, it might be interesting to consider those analog platforms in the cosmological setup (see \cite{Antonini:2019qkt,Lewkowycz:2019xse%,1810357
}). For the quantum simulation in the near-term devices, it might be interesting to explore perspectives about variational quantum simulation, and hybrid quantum-classical algorithms making use of classical knowledge, for instance, the matrix product state \cite{Yuan:2020xmq}.

\subsection{Multi-field inflation and dS conjecture}

In this work, we discuss the simulation of single field inflation models of the early universe. The above calculations could be naturally done in the multi-field case. One application of those calculations could exam dS conjectures discussed in the recent literature. 

Single-field inflation model is considered to be possibly violating the dS conjectures \cite{Agrawal:2018own}. dS conjecture argues that scalar field theory consistent with quantum gravity has to satisfy the following criterion,
\be
|\partial_{\bar{\phi}} V|\geq c V\,,
\ee
where $c$ represents an $\mathcal{O}(1)$ constant. For single field inflation, as we can observe in eq.~\eqref{eq: slow roll}, this swampland criterion then implies
\be
\epsilon\sim\epsilon_V\geq \fft{1}{2}c^2\,,
\ee
which is in tension with the slow roll condition $\epsilon \ll 1$. 

However, it is claimed that the dS conjecture does not exclude all inflationary models, e.g., it turns out that the multi-field inflation can survive \cite{Achucarro:2018vey}. We expect the algorithm in this work can be extended to multi-field inflation and then test dS conjecture by simulating the physical process. For example, it is shown that when the angular velocity $\Omega$ of inflation trajectory in multi-field space is nearly a constant, heavy modes $\sigma$ orthogonal to inflation trajectory can be integrated out, and the resulting model can be effectively described by eq.~\eqref{action of inflation} supplemented with further interaction terms \cite{Tolley:2009fg,Achucarro:2010jv,Achucarro:2010da}, where the sound speed $c_s$ encodes the details of additional heavy modes
\be
c_s=(1+\fft{4\Omega^2}{M_\sigma^2})^{-\fft{1}{2}}\,.
\ee
Thus most of our algorithm and our analysis of complexity should still be applicable, except that we have to modify the interaction terms that may alter the details of evolution. On the other hand, this model does not necessarily violate the dS conjecture, because now we have
\be
\epsilon=\epsilon_V (1+\fft{\Omega^2}{9H^2})^{-1}\geq \fft{c^2}{2}(1+\fft{\Omega^2}{9H^2})^{-1}\,.
\ee 
Large enough $\Omega$ can then rescue the slow roll condition, and the lower bound is given in \cite{Achucarro:2018vey}
\be
\fft{\Omega}{H}\geq 3\big((\fft{c N}{\Delta})^2-1\big)^{\fft{1}{2}}\,,
\ee
where $\Delta$ is $\mathcal{O}(1)$ constant associated with the DC. This bound will be important if one aims to simulate such multi-field inflation without violating the dS conjecture. A successful simulation requires a delicate balance between the appropriate input of $\Omega$ and Trotter gates eq.~\eqref{gates of Trotter} when using specific e-folding number $N$.

\subsection{Other phases of cosmology}
It might also be interesting to extend this work to other phases of cosmology. For instance, there are paradigms in the early universe as alternatives of inflation (see, for instance, \cite{Wands:1998yp,Finelli:2001sr,Khoury:2001wf}). Several proposals, for instance, the ekpyrotic universe proposal in \cite{Khoury:2001wf} contains brane configurations in string theory, are related to non-perturbative physics. Furthermore, as is mentioned before, it might be interesting to validate the in-in computations of correlation functions in those spacetime backgrounds using quantum devices.

It will also be interesting to look at other cosmic eras. For instance, one might consider cosmic reheating (see Appendix \ref{RH}) and electroweak phase transition in the early universe, where those eras are closely related to strongly-coupled physics in quantum field theory, and are also observationally relevant \cite{JPtoappear1,JPtoappear2}.

It is also important to extend the simulation algorithm in the current paper to quantum field theory living in FLRW universes arising from time-dependent compactification of string theory \cite{Russo:2018akp,Russo:2019fnk}. Simulating these FLRW spacetimes may shed light on resolving the inconsistency between inflation scenario and dS conjecture \cite{Obied:2018sgi}.

\subsection{Holographic scattering, sub-AdS locality, and the computational complexity of AdS/CFT}
Another possibility of generalizing the current work is to look at scattering experiments in AdS. Specifically, one could prepare similar states made by mode functions in the AdS spacetime backgrounds. A typical example we might consider is the holographic scattering experiments, where we start from wave packets prepared near the boundary and shoot them into the bulk. Correlation functions might be computed by tree-level Witten diagrams semiclassically, or the dual correlation functions in the boundary CFT. When turning on the coupling, we might receive loop corrections. There are several confusions and discussions regarding this process that are related to the concept of locality beyond the semiclassical theory at the sub-AdS scale \cite{Gary:2009ae,Heemskerk:2009pn,Giddings:2009gj,Fitzpatrick:2010zm,ElShowk:2011ag,Fitzpatrick:2011hu,Fitzpatrick:2011dm,Maldacena:2015iua}. Those processes are well-motivated to study when we have the computational power from quantum devices.

When simulating the above theories in the quantum computer, one should firstly discretize the theory in the lattice. Some triangulation procedures might be needed for us to write a lattice theory in some AdS coordinates. Furthermore, one has to resolve similar problems, for instance, the encoding process, in order to generalize the Jordan-Lee-Preskill algorithm in the AdS spacetime. In the dual side, one has to construct a complete algorithm in order to estimate some correlation functions in the boundary CFT, which might be made by some critical points of spin chains.

This problem is also helpful for the conceptual problem we mentioned before about the quantum Church-Turing Thesis. Naively, if we have the boundary large-$N$ CFT that is nearly a generalized free theory, the dual theory is nearly a free particle propagating in the pure AdS. In this setup, both sides should be computationally easy to simulate, although it might be interesting still to concretely come up with some algorithms. However, the situation might be changed if we include a black hole in the AdS background, and the dual theory has some higher energy states describing black hole microstates. We feel that this is a concrete setting to check the computational complexity of the AdS/CFT correspondence in the original sense of Witten \cite{Witten:1998qj}. If we could really find there are processes here that are computationally hard, it might be an indication that the quantum-extended Church Turing Thesis is violated. Furthermore, in this concrete setting, one might be able to verify the statement made by Susskind \cite{Susskind:2020kti}. Digging this problem deeper, we might find some relations between the simulation problem and some fundamental aspects of holography and quantum information theory, for instance, pseudorandomness \cite{Kim:2020cds} and quantum entanglement \cite{May:2019odp}.

\subsection{Other generalizations and further open problems}
Here we propose some further generalizations of our work.
\begin{itemize}
\item It might be interesting to study the complexity class of the problems in cosmic perturbation theory. For instance, one could try to identify the complexity class of the computational tasks presented in this paper in a more rigorous way, similar to the study in \cite{jordan2018bqp}. More precisely, we wish to prove that the cosmic perturbation theory is BQP-complete rigorously in the quantum information language.
\item It might be interesting to make a full investigation on the commutators appearing in the Trotter formula and precisely compute the Trotter constant, even for the quantum field theories in the flat space, for instance, the $\lambda \phi^4$ theory or even the free theory. It is also interesting to study the physical implications of the Trotter formula in the continuum limit.
\item It might be interesting to relate the discussions about quantum simulation in the bulk of the de Sitter space to the future time slice $\tau_{\text{end}}$, namely, the ``boundary". Those discussions are potentially related to some proposals of realizing de Sitter space in a dual theory, for instance, the dS/CFT correspondence and some other proposals \cite{Strominger:2001pn,Maloney:2002rr,Alishahiha:2004md,Dong:2018cuv,Lewkowycz:2019xse,Chen:2020tes,Hartman:2020khs}.
\item It might be interesting to investigate further about swampland conjectures that have been recently discussed (see, for instance,  \cite{Obied:2018sgi,Ooguri:2018wrx,Kachru:2018aqn,Bedroya:2019snp}). Although we are focusing on the effective field theory point of view in this paper, we expect that one day, one could construct quantum algorithms from the ultraviolet perspective and investigate further about those conjectures, besides the swampland criterion about the Church-Turing Thesis condition that has been discussed before.
\end{itemize}

\section*{Acknowledgement}
We thank Robert Brandenberger, Alex Buser, Cliff Cheung, Hrant Gharibyan, Masanori Hanada, Jim Hartle, Masazumi Honda, Isaac Kim, Hank Lamm, Ying-Ying Li, Don Marolf, David Meltzer, Ash Milsted, John Preskill, Burak Sahinoglu, Eva Silverstein, David Simmons-Duffin, Yuan Su, Jinzhao Sun, Guifre Vidal, Dong-Gang Wang, Yi Wang, Mark Wise, Zhong-Zhi Xianyu and Xiao Yuan for discussions related to this paper. JL is supported in part by the Institute for Quantum Information and Matter (IQIM), an NSF Physics Frontiers Center (NSF Grant PHY-1125565) with support from the Gordon and Betty Moore Foundation (GBMF-2644), and by the Walter Burke Institute for Theoretical Physics.

\newpage
\appendix

\section{Comments on computation and quantum gravity}\label{MH}
In this section, we make some brief discussions about computation and quantum gravity. There are several historical comments that are related to this topic, for instance, insightful discussions in \cite{Aaronson:2005qu}.

As we mentioned before, the concept of computation is naturally associated with the definition of time, which requires more consideration when we mix the definition of time and space in the theory of relativity. In fact, if we consider this problem in special relativity, we could boost the observer, and the computation might get some speedups in the observer's inertial frame. However, it requires extra energy to perform the boost. Thus, in terms of special relativity, it seems that we have to consider the computational resource by merging the time cost and the energy cost together, which is suggested, in a sense, by the Bekenstein bound \cite{Bekenstein:1980jp}, which relates energy and entropy \cite{Aaronson:2005qu}. How is this related to cosmic inflation, where we are simulating the spacetime where the energy is not conserved and even a free lunch during the cosmic expansion in this paper?

More weirdness might appear when we consider the theory of general relativity. Taking the boost example mentioned before, one could boost the particle drastically near the horizon of a black hole. Moreover, we could even construct spacetimes where one could see some entire worldlines by traveling a finite amount of time in the inertial frame. This is so-called the Malament-Hogarth spacetime \cite{hogarth1996predictability}. In the general relativity language, it means that there exist a past-extendible time-like curve and a point (event), such that the curve itself is infinitely long, but the curve is part of the time-like past of the point. One could construct such events in the Kerr or AdS spacetimes. Thus they could be Malament-Hogarth (see \cite{Welch:2006va}). In the Malament-Hogarth spacetime, one could even solve the Halting Problem.

When facing the above weirdness in general relativity, a natural argument is to admit that we have to use quantum gravity when defining those computational tasks. The construction of infinite boosts in general relativity requires the Trans-Planckian statement, where we have to discuss short-distance physics that are smaller than the Planck length. Thus, it is expected that the Malament-Hogarth construction is not physical, and the quantum gravitational backreactions will prevent it from happening. Some related questions are addressed by \cite{Hayden:2007cs,Yoshida:2020wpd,geroch1986computability} on the black hole physics, or the gravitational path integral.

However, before a full equipment of quantum gravity theories, we could make some ``temporary choices" when we are simulating some specific theories,
\begin{itemize}
\item We could make some specific choices on time coordinates before studying how covariant the system is. For instance, we choose conformal coordinates in this paper. Some similar coordinates include the Poincare coordinates in AdS, or the boundary coordinates in the asymptotically flat spacetime (for instance, the Schwarzschild black hole in the asymptotically flat background). Those coordinates are not that different from the flat space coordinates in the sense that we could take some ``limits" or conformal transformations. One could consider constructing algorithms to simulate dynamics in those spacetimes at the starting point.
\item We could focus on some well-defined tasks in quantum many-body systems and quantum field theories. For instance, we could temporarily forget how special relativity is emergent from some quantum many-body systems, and just pick some time slices equipped with some well-defined quantum physics machinery.
\end{itemize}
The freedom of temporary choices we could make does not mean that we will stop thinking about improvements to the framework. We feel that in the future, the following research directions will be valuable.
\begin{itemize}
\item Starting from continuum physics, we could think about boosting particles or changing coordinates in some specific systems. For instance, the Rindler transformation in the flat space or the black hole, or transforming from the conformal coordinate to the physical coordinate in cosmology. Moreover, in some models with weakly coupled gravity, we could compute the perturbative correction from quantum gravity to see if it could change the answer, to make the backreaction story we discuss before more precise.
\item Starting from discrete physics, we could think about how the light cone is emergent from many-body systems, and how it could be related to quantum complexity. To do this, we might consider spin models with a second-order phase transition as a starting point. We could somehow redefine the computational resource by concerning the energy costs. This might be related to the applications of quantum information theory in thermodynamics and resource theories (see a review in \cite{halpern2017toward}).
\end{itemize}

\newpage
\section{Comments on quantum simulation of cosmic reheating}\label{reheating}\label{RH}
The result presented in this paper is about cosmic inflation, but many techniques introduced in this paper could be directly applied to another problem in cosmology: cosmic reheating. Unlike cosmic inflation, the reheating process is naturally strongly-coupled. Thus, it is challenging to formulate some exact quantum field theory statements analytically in the reheating process, naturally motivating us to think about numerical investigations and possible computational speedup using quantum devices (see some references in cosmic reheating \cite{Abbott:1982hn,Dolgov:1982th,Albrecht:1982mp,Traschen:1990sw,Kofman:1994rk,Shtanov:1994ce}). The digital quantum simulation of cosmic reheating might deserve another lengthy paper to investigate. Here, we only point out our idea and leave this possibility for future research.

Inflation is a super ``cool" process where the universe is expanding with a very low temperature. At the end of inflation, inflaton condensates will decay to other particles. This process (preheating) is claimed to be violent and far from equilibrium, associated with non-perturbative phenomena, such as stochastic resonances. Finally, the universe will thermalize (reheating), setting up the scenario for big-bang nucleosynthesis. The process might also be associated with cosmological observations, such as CMB associated with observations of inflation, gravitational waves, magnetic fields, and baryon asymmetry.

Due to the non-perturbative nature of the (p)reheating process, it is challenging to make predictions only relying on perturbative analysis in quantum field theories. Several lattice simulation has been investigated to explore the process, although a full calculation in the setup of quantum field theory in the curved spacetime with a complete description in the Hilbert space is still in development \cite{Felder:2000hq}. Thus, we are facing a situation similar to the lattice gauge theory in the flat space: we might require quantum devices for our future computation to make reliable predictions. As a result, it is well-motivated to study cosmic reheating using quantum computers.

For instance, a typical model in the early study of reheating \cite{Abbott:1982hn} is to consider the following couplings between the inflaton and the matter fields
\begin{align}
S = \int {{d^3}} x dt\sqrt { - g} \left( { - \sigma \phi {\chi ^2} - h\phi \bar \psi \psi } \right)~.
\end{align}
Here $\chi$ and $\psi$ are scalar and fermionic decay products with the couplings $\sigma$ and $h$. Simple perturbative analysis shows that the decay rate is given by
\begin{align}
{\Gamma _{\phi  \to \chi \chi }} = O(\frac{{{\sigma ^2}}}{m})~,~~~~~{\Gamma _{\phi  \to \bar \psi \psi }} = O({h^2}m)~,
\end{align}
where $m$ is associated with the inflationary potential $V=m^{2} \phi^{2} / 2$. Using those formulas, one could study how matter contents are generated and estimate the temperature based on the above perturbative analysis. However, the analysis could only be trusted in a heuristic sense: the nature of the process is highly non-perturbative, and we cannot use perturbative quantum field theory techniques at tree level. Thus, a simple experiment we could start is to simulate the above process with large couplings using the digital quantum device.

The situation here is more similar to the original Jordan-Lee-Preskill setup, merging with our methods in the curved spacetime. The background scale factor should be different, so the mode functions are different in this process. But the quantizations of fields in the free theory are very similar. We could start from the free theory wave packet states of the inflaton and adiabatically turn on the interaction. Finally, we measure the decay rate of particles by evolving the state and do a similar post-selection.

The generalization from our algorithm to the above (p)reheating model should be straightforward, except now we have to encode the solution of fermions. In the flat space, it is challenging also to encode fermion fields in higher dimensions. In \cite{jordan2014quantum}, an algorithm about encoding fermions in 1+1 dimensions is established with the Bravyi-Kitaev and the Jordan-Wigner transformations. However, in the Trotter simulation, if the Hamiltonian is non-local after encoding (which is the situation of the naive Bravyi-Kitaev or Jordan-Wigner transformations), we do not have a good scaling in complexity (see \cite{childs2019theory} for a conclusion about $k$-local interactions). However, we are lucky to have other algorithms which could maintain locality (for instance, see the ``Superfast simulation of fermions" in Bravyi-Kitaev \cite{bravyi2002fermionic}, or \cite{Chen:2017fvr,Chen:2018nog,Chen:2019wlx}), suggesting novel ways of ordering, which preserves the locality in the bosonization process. This encoding should be not only useful for studying reheating, but also useful for fermionic quantum field theory simulations using quantum devices in the flat space. In the curved spacetime, several details, for instance, the spin connection in the lattice, need to be figured out. Furthermore, it might also be interesting to try analog simulations for cosmic reheating \cite{Chatrchyan:2020syc}.

\bibliographystyle{utphys}
\bibliography{refinflation}

\providecommand{\href}[2]{#2}\begingroup\raggedright\begin{thebibliography}{100}

\bibitem{Guth:1980zm}
A.~H. Guth, ``{The Inflationary Universe: A Possible Solution to the Horizon
  and Flatness Problems},''
  \href{http://dx.doi.org/10.1103/PhysRevD.23.347}{{\em Adv. Ser. Astrophys.
  Cosmol.} {\bfseries 3} (1987) 139--148}.

\bibitem{Linde:1981mu}
A.~D. Linde, ``{A New Inflationary Universe Scenario: A Possible Solution of
  the Horizon, Flatness, Homogeneity, Isotropy and Primordial Monopole
  Problems},'' \href{http://dx.doi.org/10.1016/0370-2693(82)91219-9}{{\em Adv.
  Ser. Astrophys. Cosmol.} {\bfseries 3} (1987) 149--153}.

\bibitem{Albrecht:1982wi}
A.~Albrecht and P.~J. Steinhardt, ``{Cosmology for Grand Unified Theories with
  Radiatively Induced Symmetry Breaking},''
  \href{http://dx.doi.org/10.1103/PhysRevLett.48.1220}{{\em Adv. Ser.
  Astrophys. Cosmol.} {\bfseries 3} (1987) 158--161}.

\bibitem{Starobinsky:1980te}
A.~A. Starobinsky, ``{A New Type of Isotropic Cosmological Models Without
  Singularity},'' \href{http://dx.doi.org/10.1016/0370-2693(80)90670-X}{{\em
  Adv. Ser. Astrophys. Cosmol.} {\bfseries 3} (1987) 130--133}.

\bibitem{Sato:1980yn}
K.~Sato, ``{First Order Phase Transition of a Vacuum and Expansion of the
  Universe},'' {\em Mon. Not. Roy. Astron. Soc.} {\bfseries 195} (1981)
  467--479.

\bibitem{Fang:1980wi}
L.~Fang, ``{Entropy Generation in the Early Universe by Dissipative Processes
  Near the Higgs' Phase Transitions},''
  \href{http://dx.doi.org/10.1016/0370-2693(80)90421-9}{{\em Phys. Lett. B}
  {\bfseries 95} (1980) 154--156}.

\bibitem{Mukhanov:1981xt}
V.~F. Mukhanov and G.~V. Chibisov, ``{Quantum Fluctuations and a Nonsingular
  Universe},'' {\em JETP Lett.} {\bfseries 33} (1981) 532--535.

\bibitem{Starobinsky:1982ee}
A.~A. Starobinsky, ``{Dynamics of Phase Transition in the New Inflationary
  Universe Scenario and Generation of Perturbations},''
  \href{http://dx.doi.org/10.1016/0370-2693(82)90541-X}{{\em Phys. Lett. B}
  {\bfseries 117} (1982) 175--178}.

\bibitem{Guth:1982ec}
A.~H. Guth and S.~Pi, ``{Fluctuations in the New Inflationary Universe},''
  \href{http://dx.doi.org/10.1103/PhysRevLett.49.1110}{{\em Phys. Rev. Lett.}
  {\bfseries 49} (1982) 1110--1113}.

\bibitem{Bardeen:1983qw}
J.~M. Bardeen, P.~J. Steinhardt, and M.~S. Turner, ``{Spontaneous Creation of
  Almost Scale - Free Density Perturbations in an Inflationary Universe},''
  \href{http://dx.doi.org/10.1103/PhysRevD.28.679}{{\em Phys. Rev. D}
  {\bfseries 28} (1983) 679}.

\bibitem{Mukhanov:1990me}
V.~F. Mukhanov, H.~Feldman, and R.~H. Brandenberger, ``{Theory of cosmological
  perturbations. Part 1. Classical perturbations. Part 2. Quantum theory of
  perturbations. Part 3. Extensions},''
  \href{http://dx.doi.org/10.1016/0370-1573(92)90044-Z}{{\em Phys. Rept.}
  {\bfseries 215} (1992) 203--333}.

\bibitem{Maldacena:2002vr}
J.~M. Maldacena, ``{Non-Gaussian features of primordial fluctuations in single
  field inflationary models},''
  \href{http://dx.doi.org/10.1088/1126-6708/2003/05/013}{{\em JHEP} {\bfseries
  05} (2003) 013}, \href{http://arxiv.org/abs/astro-ph/0210603}{{\ttfamily
  arXiv:astro-ph/0210603}}.

\bibitem{Linde:2005ht}
A.~D. Linde, {\em {Particle physics and inflationary cosmology}}, vol.~5.
\newblock 1990.
\newblock \href{http://arxiv.org/abs/hep-th/0503203}{{\ttfamily
  arXiv:hep-th/0503203}}.

\bibitem{Chen:2006nt}
X.~Chen, M.-x. Huang, S.~Kachru, and G.~Shiu, ``{Observational signatures and
  non-Gaussianities of general single field inflation},''
  \href{http://dx.doi.org/10.1088/1475-7516/2007/01/002}{{\em JCAP} {\bfseries
  01} (2007) 002}, \href{http://arxiv.org/abs/hep-th/0605045}{{\ttfamily
  arXiv:hep-th/0605045}}.

\bibitem{Cheung:2007st}
C.~Cheung, P.~Creminelli, A.~Fitzpatrick, J.~Kaplan, and L.~Senatore, ``{The
  Effective Field Theory of Inflation},''
  \href{http://dx.doi.org/10.1088/1126-6708/2008/03/014}{{\em JHEP} {\bfseries
  03} (2008) 014}, \href{http://arxiv.org/abs/0709.0293}{{\ttfamily
  arXiv:0709.0293 [hep-th]}}.

\bibitem{Weinberg:2008hq}
S.~Weinberg, ``{Effective Field Theory for Inflation},''
  \href{http://dx.doi.org/10.1103/PhysRevD.77.123541}{{\em Phys. Rev. D}
  {\bfseries 77} (2008) 123541},
  \href{http://arxiv.org/abs/0804.4291}{{\ttfamily arXiv:0804.4291 [hep-th]}}.

\bibitem{Malik:2008im}
K.~A. Malik and D.~Wands, ``{Cosmological perturbations},''
  \href{http://dx.doi.org/10.1016/j.physrep.2009.03.001}{{\em Phys. Rept.}
  {\bfseries 475} (2009) 1--51},
  \href{http://arxiv.org/abs/0809.4944}{{\ttfamily arXiv:0809.4944
  [astro-ph]}}.

\bibitem{Chen:2009zp}
X.~Chen and Y.~Wang, ``{Quasi-Single Field Inflation and Non-Gaussianities},''
  \href{http://dx.doi.org/10.1088/1475-7516/2010/04/027}{{\em JCAP} {\bfseries
  04} (2010) 027}, \href{http://arxiv.org/abs/0911.3380}{{\ttfamily
  arXiv:0911.3380 [hep-th]}}.

\bibitem{Baumann:2009ds}
D.~Baumann,
  \href{http://dx.doi.org/10.1142/9789814327183\_0010}{``{Inflation},''} in
  {\em {Theoretical Advanced Study Institute in Elementary Particle Physics}:
  {Physics of the Large and the Small}}, pp.~523--686.
\newblock 2011.
\newblock \href{http://arxiv.org/abs/0907.5424}{{\ttfamily arXiv:0907.5424
  [hep-th]}}.

\bibitem{Chen:2010xka}
X.~Chen, ``{Primordial Non-Gaussianities from Inflation Models},''
  \href{http://dx.doi.org/10.1155/2010/638979}{{\em Adv. Astron.} {\bfseries
  2010} (2010) 638979}, \href{http://arxiv.org/abs/1002.1416}{{\ttfamily
  arXiv:1002.1416 [astro-ph.CO]}}.

\bibitem{Wang:2013zva}
Y.~Wang, ``{Inflation, Cosmic Perturbations and Non-Gaussianities},''
  \href{http://dx.doi.org/10.1088/0253-6102/62/1/19}{{\em Commun. Theor. Phys.}
  {\bfseries 62} (2014) 109--166},
  \href{http://arxiv.org/abs/1303.1523}{{\ttfamily arXiv:1303.1523 [hep-th]}}.

\bibitem{Baumann:2014nda}
D.~Baumann and L.~McAllister,
  \href{http://dx.doi.org/10.1017/CBO9781316105733}{{\em {Inflation and String
  Theory}}}.
\newblock Cambridge Monographs on Mathematical Physics. Cambridge University
  Press, 5, 2015.
\newblock \href{http://arxiv.org/abs/1404.2601}{{\ttfamily arXiv:1404.2601
  [hep-th]}}.

\bibitem{Abbott:1982hn}
L.~Abbott, E.~Farhi, and M.~B. Wise, ``{Particle Production in the New
  Inflationary Cosmology},''
  \href{http://dx.doi.org/10.1016/0370-2693(82)90867-X}{{\em Phys. Lett. B}
  {\bfseries 117} (1982) 29}.

\bibitem{Dolgov:1982th}
A.~Dolgov and A.~D. Linde, ``{Baryon Asymmetry in Inflationary Universe},''
  \href{http://dx.doi.org/10.1016/0370-2693(82)90292-1}{{\em Phys. Lett. B}
  {\bfseries 116} (1982) 329}.

\bibitem{Albrecht:1982mp}
A.~Albrecht, P.~J. Steinhardt, M.~S. Turner, and F.~Wilczek, ``{Reheating an
  Inflationary Universe},''
  \href{http://dx.doi.org/10.1103/PhysRevLett.48.1437}{{\em Phys. Rev. Lett.}
  {\bfseries 48} (1982) 1437}.

\bibitem{Traschen:1990sw}
J.~H. Traschen and R.~H. Brandenberger, ``{Particle Production During
  Out-of-equilibrium Phase Transitions},''
  \href{http://dx.doi.org/10.1103/PhysRevD.42.2491}{{\em Phys. Rev. D}
  {\bfseries 42} (1990) 2491--2504}.

\bibitem{Kofman:1994rk}
L.~Kofman, A.~D. Linde, and A.~A. Starobinsky, ``{Reheating after inflation},''
  \href{http://dx.doi.org/10.1103/PhysRevLett.73.3195}{{\em Phys. Rev. Lett.}
  {\bfseries 73} (1994) 3195--3198},
  \href{http://arxiv.org/abs/hep-th/9405187}{{\ttfamily arXiv:hep-th/9405187}}.

\bibitem{Shtanov:1994ce}
Y.~Shtanov, J.~H. Traschen, and R.~H. Brandenberger, ``{Universe reheating
  after inflation},'' \href{http://dx.doi.org/10.1103/PhysRevD.51.5438}{{\em
  Phys. Rev. D} {\bfseries 51} (1995) 5438--5455},
  \href{http://arxiv.org/abs/hep-ph/9407247}{{\ttfamily arXiv:hep-ph/9407247}}.

\bibitem{Sakharov:1967dj}
A.~Sakharov, ``{Violation of CP Invariance, C asymmetry, and baryon asymmetry
  of the universe},''
  \href{http://dx.doi.org/10.1070/PU1991v034n05ABEH002497}{{\em Sov. Phys.
  Usp.} {\bfseries 34} no.~5, (1991) 392--393}.

\bibitem{Kuzmin:1985mm}
V.~Kuzmin, V.~Rubakov, and M.~Shaposhnikov, ``{On the Anomalous Electroweak
  Baryon Number Nonconservation in the Early Universe},''
  \href{http://dx.doi.org/10.1016/0370-2693(85)91028-7}{{\em Phys. Lett. B}
  {\bfseries 155} (1985) 36}.

\bibitem{Shaposhnikov:1987tw}
M.~Shaposhnikov, ``{Baryon Asymmetry of the Universe in Standard Electroweak
  Theory},'' \href{http://dx.doi.org/10.1016/0550-3213(87)90127-1}{{\em Nucl.
  Phys. B} {\bfseries 287} (1987) 757--775}.

\bibitem{Cohen:1993nk}
A.~G. Cohen, D.~Kaplan, and A.~Nelson, ``{Progress in electroweak
  baryogenesis},''
  \href{http://dx.doi.org/10.1146/annurev.ns.43.120193.000331}{{\em Ann. Rev.
  Nucl. Part. Sci.} {\bfseries 43} (1993) 27--70},
  \href{http://arxiv.org/abs/hep-ph/9302210}{{\ttfamily arXiv:hep-ph/9302210}}.

\bibitem{Riess:1998cb}
{\bfseries Supernova Search Team} Collaboration, A.~G. Riess {\em et~al.},
  ``{Observational evidence from supernovae for an accelerating universe and a
  cosmological constant},'' \href{http://dx.doi.org/10.1086/300499}{{\em
  Astron. J.} {\bfseries 116} (1998) 1009--1038},
  \href{http://arxiv.org/abs/astro-ph/9805201}{{\ttfamily
  arXiv:astro-ph/9805201}}.

\bibitem{Carroll:2000fy}
S.~M. Carroll, ``{The Cosmological constant},''
  \href{http://dx.doi.org/10.12942/lrr-2001-1}{{\em Living Rev. Rel.}
  {\bfseries 4} (2001) 1},
  \href{http://arxiv.org/abs/astro-ph/0004075}{{\ttfamily
  arXiv:astro-ph/0004075}}.

\bibitem{Peebles:2002gy}
P.~Peebles and B.~Ratra, ``{The Cosmological Constant and Dark Energy},''
  \href{http://dx.doi.org/10.1103/RevModPhys.75.559}{{\em Rev. Mod. Phys.}
  {\bfseries 75} (2003) 559--606},
  \href{http://arxiv.org/abs/astro-ph/0207347}{{\ttfamily
  arXiv:astro-ph/0207347}}.

\bibitem{Witten:2001kn}
E.~Witten, ``{Quantum gravity in de Sitter space},'' in {\em {Strings 2001:
  International Conference}}.
\newblock 6, 2001.
\newblock \href{http://arxiv.org/abs/hep-th/0106109}{{\ttfamily
  arXiv:hep-th/0106109}}.

\bibitem{Strominger:2001pn}
A.~Strominger, ``{The dS / CFT correspondence},''
  \href{http://dx.doi.org/10.1088/1126-6708/2001/10/034}{{\em JHEP} {\bfseries
  10} (2001) 034}, \href{http://arxiv.org/abs/hep-th/0106113}{{\ttfamily
  arXiv:hep-th/0106113}}.

\bibitem{Maloney:2002rr}
A.~Maloney, E.~Silverstein, and A.~Strominger, ``{De Sitter space in
  noncritical string theory},'' in {\em {Workshop on Conference on the Future
  of Theoretical Physics and Cosmology in Honor of Steven Hawking's 60th
  Birthday}}, pp.~570--591.
\newblock 5, 2002.
\newblock \href{http://arxiv.org/abs/hep-th/0205316}{{\ttfamily
  arXiv:hep-th/0205316}}.

\bibitem{Alishahiha:2004md}
M.~Alishahiha, A.~Karch, E.~Silverstein, and D.~Tong, ``{The dS/dS
  correspondence},'' \href{http://dx.doi.org/10.1063/1.1848341}{{\em AIP Conf.
  Proc.} {\bfseries 743} no.~1, (2004) 393--409},
  \href{http://arxiv.org/abs/hep-th/0407125}{{\ttfamily arXiv:hep-th/0407125}}.

\bibitem{Arkani-Hamed:2015bza}
N.~Arkani-Hamed and J.~Maldacena, ``{Cosmological Collider Physics},''
  \href{http://arxiv.org/abs/1503.08043}{{\ttfamily arXiv:1503.08043
  [hep-th]}}.

\bibitem{Dong:2018cuv}
X.~Dong, E.~Silverstein, and G.~Torroba, ``{De Sitter Holography and
  Entanglement Entropy},''
  \href{http://dx.doi.org/10.1007/JHEP07(2018)050}{{\em JHEP} {\bfseries 07}
  (2018) 050}, \href{http://arxiv.org/abs/1804.08623}{{\ttfamily
  arXiv:1804.08623 [hep-th]}}.

\bibitem{Lewkowycz:2019xse}
A.~Lewkowycz, J.~Liu, E.~Silverstein, and G.~Torroba, ``{$ T\overline{T} $ and
  EE, with implications for (A)dS subregion encodings},''
  \href{http://dx.doi.org/10.1007/JHEP04(2020)152}{{\em JHEP} {\bfseries 04}
  (2020) 152}, \href{http://arxiv.org/abs/1909.13808}{{\ttfamily
  arXiv:1909.13808 [hep-th]}}.

\bibitem{Chen:2020tes}
Y.~Chen, V.~Gorbenko, and J.~Maldacena, ``{Bra-ket wormholes in gravitationally
  prepared states},'' \href{http://arxiv.org/abs/2007.16091}{{\ttfamily
  arXiv:2007.16091 [hep-th]}}.

\bibitem{Hartman:2020khs}
T.~Hartman, Y.~Jiang, and E.~Shaghoulian, ``{Islands in cosmology},''
  \href{http://arxiv.org/abs/2008.01022}{{\ttfamily arXiv:2008.01022
  [hep-th]}}.

\bibitem{Kachru:2003aw}
S.~Kachru, R.~Kallosh, A.~D. Linde, and S.~P. Trivedi, ``{De Sitter vacua in
  string theory},'' \href{http://dx.doi.org/10.1103/PhysRevD.68.046005}{{\em
  Phys. Rev. D} {\bfseries 68} (2003) 046005},
  \href{http://arxiv.org/abs/hep-th/0301240}{{\ttfamily arXiv:hep-th/0301240}}.

\bibitem{preskill2018quantum}
J.~Preskill, ``Quantum computing in the nisq era and beyond,'' {\em Quantum}
  {\bfseries 2} (2018) 79.

\bibitem{arute2019quantum}
F.~Arute, K.~Arya, R.~Babbush, D.~Bacon, J.~C. Bardin, R.~Barends, R.~Biswas,
  S.~Boixo, F.~G. Brandao, D.~A. Buell, {\em et~al.}, ``Quantum supremacy using
  a programmable superconducting processor,'' {\em Nature} {\bfseries 574}
  no.~7779, (2019) 505--510.

\bibitem{Jordan:2011ne}
S.~P. Jordan, K.~S. Lee, and J.~Preskill, ``{Quantum Algorithms for Quantum
  Field Theories},'' \href{http://dx.doi.org/10.1126/science.1217069}{{\em
  Science} {\bfseries 336} (2012) 1130--1133},
  \href{http://arxiv.org/abs/1111.3633}{{\ttfamily arXiv:1111.3633
  [quant-ph]}}.

\bibitem{Jordan:2011ci}
S.~P. Jordan, K.~S. Lee, and J.~Preskill, ``{Quantum Computation of Scattering
  in Scalar Quantum Field Theories},'' {\em Quant. Inf. Comput.} {\bfseries 14}
  (2014) 1014--1080, \href{http://arxiv.org/abs/1112.4833}{{\ttfamily
  arXiv:1112.4833 [hep-th]}}.

\bibitem{jordan2014quantum}
S.~P. Jordan, K.~S. Lee, and J.~Preskill, ``Quantum algorithms for fermionic
  quantum field theories,'' {\em arXiv preprint arXiv:1404.7115} (2014) .

\bibitem{jordan2018bqp}
S.~P. Jordan, H.~Krovi, K.~S. Lee, and J.~Preskill, ``Bqp-completeness of
  scattering in scalar quantum field theory,'' {\em Quantum} {\bfseries 2}
  (2018) 44.

\bibitem{Preskill:2018fag}
J.~Preskill, ``{Simulating quantum field theory with a quantum computer},''
  \href{http://dx.doi.org/10.22323/1.334.0024}{{\em PoS} {\bfseries
  LATTICE2018} (2018) 024}, \href{http://arxiv.org/abs/1811.10085}{{\ttfamily
  arXiv:1811.10085 [hep-lat]}}.

\bibitem{Weinberg:2005vy}
S.~Weinberg, ``{Quantum contributions to cosmological correlations},''
  \href{http://dx.doi.org/10.1103/PhysRevD.72.043514}{{\em Phys. Rev. D}
  {\bfseries 72} (2005) 043514},
  \href{http://arxiv.org/abs/hep-th/0506236}{{\ttfamily arXiv:hep-th/0506236}}.

\bibitem{Senatore:2009cf}
L.~Senatore and M.~Zaldarriaga, ``{On Loops in Inflation},''
  \href{http://dx.doi.org/10.1007/JHEP12(2010)008}{{\em JHEP} {\bfseries 12}
  (2010) 008}, \href{http://arxiv.org/abs/0912.2734}{{\ttfamily arXiv:0912.2734
  [hep-th]}}.

\bibitem{Senatore:2012nq}
L.~Senatore and M.~Zaldarriaga, ``{On Loops in Inflation II: IR Effects in
  Single Clock Inflation},''
  \href{http://dx.doi.org/10.1007/JHEP01(2013)109}{{\em JHEP} {\bfseries 01}
  (2013) 109}, \href{http://arxiv.org/abs/1203.6354}{{\ttfamily arXiv:1203.6354
  [hep-th]}}.

\bibitem{Pimentel:2012tw}
G.~L. Pimentel, L.~Senatore, and M.~Zaldarriaga, ``{On Loops in Inflation III:
  Time Independence of zeta in Single Clock Inflation},''
  \href{http://dx.doi.org/10.1007/JHEP07(2012)166}{{\em JHEP} {\bfseries 07}
  (2012) 166}, \href{http://arxiv.org/abs/1203.6651}{{\ttfamily arXiv:1203.6651
  [hep-th]}}.

\bibitem{Gary:2009ae}
M.~Gary, S.~B. Giddings, and J.~Penedones, ``{Local bulk S-matrix elements and
  CFT singularities},''
  \href{http://dx.doi.org/10.1103/PhysRevD.80.085005}{{\em Phys. Rev. D}
  {\bfseries 80} (2009) 085005},
  \href{http://arxiv.org/abs/0903.4437}{{\ttfamily arXiv:0903.4437 [hep-th]}}.

\bibitem{Heemskerk:2009pn}
I.~Heemskerk, J.~Penedones, J.~Polchinski, and J.~Sully, ``{Holography from
  Conformal Field Theory},''
  \href{http://dx.doi.org/10.1088/1126-6708/2009/10/079}{{\em JHEP} {\bfseries
  10} (2009) 079}, \href{http://arxiv.org/abs/0907.0151}{{\ttfamily
  arXiv:0907.0151 [hep-th]}}.

\bibitem{Giddings:2009gj}
S.~B. Giddings and R.~A. Porto, ``{The Gravitational S-matrix},''
  \href{http://dx.doi.org/10.1103/PhysRevD.81.025002}{{\em Phys. Rev. D}
  {\bfseries 81} (2010) 025002},
  \href{http://arxiv.org/abs/0908.0004}{{\ttfamily arXiv:0908.0004 [hep-th]}}.

\bibitem{Fitzpatrick:2010zm}
A.~Fitzpatrick, E.~Katz, D.~Poland, and D.~Simmons-Duffin, ``{Effective
  Conformal Theory and the Flat-Space Limit of AdS},''
  \href{http://dx.doi.org/10.1007/JHEP07(2011)023}{{\em JHEP} {\bfseries 07}
  (2011) 023}, \href{http://arxiv.org/abs/1007.2412}{{\ttfamily arXiv:1007.2412
  [hep-th]}}.

\bibitem{ElShowk:2011ag}
S.~El-Showk and K.~Papadodimas, ``{Emergent Spacetime and Holographic CFTs},''
  \href{http://dx.doi.org/10.1007/JHEP10(2012)106}{{\em JHEP} {\bfseries 10}
  (2012) 106}, \href{http://arxiv.org/abs/1101.4163}{{\ttfamily arXiv:1101.4163
  [hep-th]}}.

\bibitem{Fitzpatrick:2011hu}
A.~Fitzpatrick and J.~Kaplan, ``{Analyticity and the Holographic S-Matrix},''
  \href{http://dx.doi.org/10.1007/JHEP10(2012)127}{{\em JHEP} {\bfseries 10}
  (2012) 127}, \href{http://arxiv.org/abs/1111.6972}{{\ttfamily arXiv:1111.6972
  [hep-th]}}.

\bibitem{Fitzpatrick:2011dm}
A.~Fitzpatrick and J.~Kaplan, ``{Unitarity and the Holographic S-Matrix},''
  \href{http://dx.doi.org/10.1007/JHEP10(2012)032}{{\em JHEP} {\bfseries 10}
  (2012) 032}, \href{http://arxiv.org/abs/1112.4845}{{\ttfamily arXiv:1112.4845
  [hep-th]}}.

\bibitem{Maldacena:2015iua}
J.~Maldacena, D.~Simmons-Duffin, and A.~Zhiboedov, ``{Looking for a bulk
  point},'' \href{http://dx.doi.org/10.1007/JHEP01(2017)013}{{\em JHEP}
  {\bfseries 01} (2017) 013}, \href{http://arxiv.org/abs/1509.03612}{{\ttfamily
  arXiv:1509.03612 [hep-th]}}.

\bibitem{Hamilton:2005ju}
A.~Hamilton, D.~N. Kabat, G.~Lifschytz, and D.~A. Lowe, ``{Local bulk operators
  in AdS/CFT: A Boundary view of horizons and locality},''
  \href{http://dx.doi.org/10.1103/PhysRevD.73.086003}{{\em Phys. Rev. D}
  {\bfseries 73} (2006) 086003},
  \href{http://arxiv.org/abs/hep-th/0506118}{{\ttfamily arXiv:hep-th/0506118}}.

\bibitem{Harlow:2018fse}
D.~Harlow, ``{TASI Lectures on the Emergence of Bulk Physics in AdS/CFT},''
  \href{http://dx.doi.org/10.22323/1.305.0002}{{\em PoS} {\bfseries TASI2017}
  (2018) 002}, \href{http://arxiv.org/abs/1802.01040}{{\ttfamily
  arXiv:1802.01040 [hep-th]}}.

\bibitem{Bunch:1978yq}
T.~Bunch and P.~Davies, ``{Quantum Field Theory in de Sitter Space:
  Renormalization by Point Splitting},''
  \href{http://dx.doi.org/10.1098/rspa.1978.0060}{{\em Proc. Roy. Soc. Lond. A}
  {\bfseries A360} (1978) 117--134}.

\bibitem{Jacobson:2005bg}
T.~Jacobson, S.~Liberati, and D.~Mattingly, ``{Lorentz violation at high
  energy: Concepts, phenomena and astrophysical constraints},''
  \href{http://dx.doi.org/10.1016/j.aop.2005.06.004}{{\em Annals Phys.}
  {\bfseries 321} (2006) 150--196},
  \href{http://arxiv.org/abs/astro-ph/0505267}{{\ttfamily
  arXiv:astro-ph/0505267}}.

\bibitem{Gibbons:1977mu}
G.~Gibbons and S.~Hawking, ``{Cosmological Event Horizons, Thermodynamics, and
  Particle Creation},'' \href{http://dx.doi.org/10.1103/PhysRevD.15.2738}{{\em
  Phys. Rev. D} {\bfseries 15} (1977) 2738--2751}.

\bibitem{Guth:1985ya}
A.~H. Guth and S.-Y. Pi, ``{The Quantum Mechanics of the Scalar Field in the
  New Inflationary Universe},''
  \href{http://dx.doi.org/10.1103/PhysRevD.32.1899}{{\em Phys. Rev. D}
  {\bfseries 32} (1985) 1899--1920}.

\bibitem{Albrecht:1992kf}
A.~Albrecht, P.~Ferreira, M.~Joyce, and T.~Prokopec, ``{Inflation and squeezed
  quantum states},'' \href{http://dx.doi.org/10.1103/PhysRevD.50.4807}{{\em
  Phys. Rev. D} {\bfseries 50} (1994) 4807--4820},
  \href{http://arxiv.org/abs/astro-ph/9303001}{{\ttfamily
  arXiv:astro-ph/9303001}}.

\bibitem{Lesgourgues:1996jc}
J.~Lesgourgues, D.~Polarski, and A.~A. Starobinsky, ``{Quantum to classical
  transition of cosmological perturbations for nonvacuum initial states},''
  \href{http://dx.doi.org/10.1016/S0550-3213(97)00224-1}{{\em Nucl. Phys. B}
  {\bfseries 497} (1997) 479--510},
  \href{http://arxiv.org/abs/gr-qc/9611019}{{\ttfamily arXiv:gr-qc/9611019}}.

\bibitem{Kiefer:1998qe}
C.~Kiefer, D.~Polarski, and A.~A. Starobinsky, ``{Quantum to classical
  transition for fluctuations in the early universe},''
  \href{http://dx.doi.org/10.1142/S0218271898000292}{{\em Int. J. Mod. Phys. D}
  {\bfseries 7} (1998) 455--462},
  \href{http://arxiv.org/abs/gr-qc/9802003}{{\ttfamily arXiv:gr-qc/9802003}}.

\bibitem{Liu:2016aaf}
J.~Liu, C.-M. Sou, and Y.~Wang, ``{Cosmic Decoherence: Massive Fields},''
  \href{http://dx.doi.org/10.1007/JHEP10(2016)072}{{\em JHEP} {\bfseries 10}
  (2016) 072}, \href{http://arxiv.org/abs/1608.07909}{{\ttfamily
  arXiv:1608.07909 [hep-th]}}.

\bibitem{mcardle2020quantum}
S.~McArdle, S.~Endo, A.~Aspuru-Guzik, S.~C. Benjamin, and X.~Yuan, ``Quantum
  computational chemistry,'' {\em Reviews of Modern Physics} {\bfseries 92}
  no.~1, (2020) 015003.

\bibitem{Liu:2020eoa}
J.~Liu and Y.~Xin, ``{Quantum simulation of quantum field theories as quantum
  chemistry},'' \href{http://arxiv.org/abs/2004.13234}{{\ttfamily
  arXiv:2004.13234 [hep-th]}}.

\bibitem{Easther:2001fi}
R.~Easther, B.~R. Greene, W.~H. Kinney, and G.~Shiu, ``{Inflation as a probe of
  short distance physics},''
  \href{http://dx.doi.org/10.1103/PhysRevD.64.103502}{{\em Phys. Rev. D}
  {\bfseries 64} (2001) 103502},
  \href{http://arxiv.org/abs/hep-th/0104102}{{\ttfamily arXiv:hep-th/0104102}}.

\bibitem{Jiang:2016nok}
H.~Jiang, Y.~Wang, and S.~Zhou, ``{On the initial condition of inflationary
  fluctuations},'' \href{http://dx.doi.org/10.1088/1475-7516/2016/04/041}{{\em
  JCAP} {\bfseries 04} (2016) 041},
  \href{http://arxiv.org/abs/1601.01179}{{\ttfamily arXiv:1601.01179
  [hep-th]}}.

\bibitem{Lewis:2019oyx}
A.~G. Lewis and G.~Vidal, ``{Classical Simulations of Quantum Field Theory in
  Curved Spacetime I: Fermionic Hawking-Hartle Vacua from a Staggered Lattice
  Scheme},'' \href{http://arxiv.org/abs/1911.12978}{{\ttfamily arXiv:1911.12978
  [gr-qc]}}.

\bibitem{Burgess:2009bs}
C.~Burgess, L.~Leblond, R.~Holman, and S.~Shandera, ``{Super-Hubble de Sitter
  Fluctuations and the Dynamical RG},''
  \href{http://dx.doi.org/10.1088/1475-7516/2010/03/033}{{\em JCAP} {\bfseries
  03} (2010) 033}, \href{http://arxiv.org/abs/0912.1608}{{\ttfamily
  arXiv:0912.1608 [hep-th]}}.

\bibitem{Brower:2019kyh}
R.~C. Brower, C.~V. Cogburn, A.~L. Fitzpatrick, D.~Howarth, and C.-I. Tan,
  ``{Lattice Setup for Quantum Field Theory in AdS$_2$},''
  \href{http://arxiv.org/abs/1912.07606}{{\ttfamily arXiv:1912.07606
  [hep-th]}}.

\bibitem{Harlow:2011az}
D.~Harlow, S.~H. Shenker, D.~Stanford, and L.~Susskind, ``{Tree-like structure
  of eternal inflation: A solvable model},''
  \href{http://dx.doi.org/10.1103/PhysRevD.85.063516}{{\em Phys. Rev. D}
  {\bfseries 85} (2012) 063516},
  \href{http://arxiv.org/abs/1110.0496}{{\ttfamily arXiv:1110.0496 [hep-th]}}.

\bibitem{Gubser:2016guj}
S.~S. Gubser, J.~Knaute, S.~Parikh, A.~Samberg, and P.~Witaszczyk, ``{$p$-adic
  AdS/CFT},'' \href{http://dx.doi.org/10.1007/s00220-016-2813-6}{{\em Commun.
  Math. Phys.} {\bfseries 352} no.~3, (2017) 1019--1059},
  \href{http://arxiv.org/abs/1605.01061}{{\ttfamily arXiv:1605.01061
  [hep-th]}}.

\bibitem{Heydeman:2016ldy}
M.~Heydeman, M.~Marcolli, I.~Saberi, and B.~Stoica, ``{Tensor networks,
  $p$-adic fields, and algebraic curves: arithmetic and the AdS$_3$/CFT$_2$
  correspondence},'' \href{http://dx.doi.org/10.4310/ATMP.2018.v22.n1.a4}{{\em
  Adv. Theor. Math. Phys.} {\bfseries 22} (2018) 93--176},
  \href{http://arxiv.org/abs/1605.07639}{{\ttfamily arXiv:1605.07639
  [hep-th]}}.

\bibitem{Bao:2017iye}
N.~Bao, C.~Cao, S.~M. Carroll, and L.~McAllister, ``{Quantum Circuit Cosmology:
  The Expansion of the Universe Since the First Qubit},''
  \href{http://arxiv.org/abs/1702.06959}{{\ttfamily arXiv:1702.06959
  [hep-th]}}.

\bibitem{Osborne:2017woa}
T.~J. Osborne and D.~E. Stiegemann, ``{Dynamics for holographic codes},''
  \href{http://dx.doi.org/10.1007/JHEP04(2020)154}{{\em JHEP} {\bfseries 04}
  (2020) 154}, \href{http://arxiv.org/abs/1706.08823}{{\ttfamily
  arXiv:1706.08823 [quant-ph]}}.

\bibitem{Martin:2000xs}
J.~Martin and R.~H. Brandenberger, ``{The TransPlanckian problem of
  inflationary cosmology},''
  \href{http://dx.doi.org/10.1103/PhysRevD.63.123501}{{\em Phys. Rev. D}
  {\bfseries 63} (2001) 123501},
  \href{http://arxiv.org/abs/hep-th/0005209}{{\ttfamily arXiv:hep-th/0005209}}.

\bibitem{GellMann:1951rw}
M.~Gell-Mann and F.~Low, ``{Bound states in quantum field theory},''
  \href{http://dx.doi.org/10.1103/PhysRev.84.350}{{\em Phys. Rev.} {\bfseries
  84} (1951) 350--354}.

\bibitem{JPwork}
A.~Lucia, D.~Poulin, J.~Preskill, and T.~Wang, ``{working in progress},''.

\bibitem{Easson:2018qgr}
D.~A. Easson and T.~Manton, ``{Stable Cosmic Time Crystals},''
  \href{http://dx.doi.org/10.1103/PhysRevD.99.043507}{{\em Phys. Rev. D}
  {\bfseries 99} no.~4, (2019) 043507},
  \href{http://arxiv.org/abs/1802.03693}{{\ttfamily arXiv:1802.03693
  [hep-th]}}.

\bibitem{kitaev2008wavefunction}
A.~Kitaev and W.~A. Webb, ``Wavefunction preparation and resampling using a
  quantum computer,'' {\em arXiv preprint arXiv:0801.0342} (2008) .

\bibitem{bunch1974triangular}
J.~R. Bunch and J.~E. Hopcroft, ``Triangular factorization and inversion by
  fast matrix multiplication,'' {\em Mathematics of Computation} {\bfseries 28}
  no.~125, (1974) 231--236.

\bibitem{coppersmith1987matrix}
D.~Coppersmith and S.~Winograd, ``Matrix multiplication via arithmetic
  progressions,'' in {\em Proceedings of the nineteenth annual ACM symposium on
  Theory of computing}, pp.~1--6.
\newblock 1987.

\bibitem{childs2019theory}
A.~M. Childs, Y.~Su, M.~C. Tran, N.~Wiebe, and S.~Zhu, ``A theory of trotter
  error,'' {\em arXiv preprint arXiv:1912.08854} (2019) .

\bibitem{Almheiri:2014lwa}
A.~Almheiri, X.~Dong, and D.~Harlow, ``{Bulk Locality and Quantum Error
  Correction in AdS/CFT},''
  \href{http://dx.doi.org/10.1007/JHEP04(2015)163}{{\em JHEP} {\bfseries 04}
  (2015) 163}, \href{http://arxiv.org/abs/1411.7041}{{\ttfamily arXiv:1411.7041
  [hep-th]}}.

\bibitem{Dong:2016eik}
X.~Dong, D.~Harlow, and A.~C. Wall, ``{Reconstruction of Bulk Operators within
  the Entanglement Wedge in Gauge-Gravity Duality},''
  \href{http://dx.doi.org/10.1103/PhysRevLett.117.021601}{{\em Phys. Rev.
  Lett.} {\bfseries 117} no.~2, (2016) 021601},
  \href{http://arxiv.org/abs/1601.05416}{{\ttfamily arXiv:1601.05416
  [hep-th]}}.

\bibitem{Pastawski:2015qua}
F.~Pastawski, B.~Yoshida, D.~Harlow, and J.~Preskill, ``{Holographic quantum
  error-correcting codes: Toy models for the bulk/boundary correspondence},''
  \href{http://dx.doi.org/10.1007/JHEP06(2015)149}{{\em JHEP} {\bfseries 06}
  (2015) 149}, \href{http://arxiv.org/abs/1503.06237}{{\ttfamily
  arXiv:1503.06237 [hep-th]}}.

\bibitem{May:2019odp}
A.~May, G.~Penington, and J.~Sorce, ``{Holographic scattering requires a
  connected entanglement wedge},''
  \href{http://dx.doi.org/10.1007/JHEP08(2020)132}{{\em JHEP} {\bfseries 20}
  (2020) 132}, \href{http://arxiv.org/abs/1912.05649}{{\ttfamily
  arXiv:1912.05649 [hep-th]}}.

\bibitem{somma2015quantum}
R.~D. Somma, ``Quantum simulations of one dimensional quantum systems,'' {\em
  arXiv preprint arXiv:1503.06319} (2015) .

\bibitem{macridin2018digital}
A.~Macridin, P.~Spentzouris, J.~Amundson, and R.~Harnik, ``Digital quantum
  computation of fermion-boson interacting systems,'' {\em Physical Review A}
  {\bfseries 98} no.~4, (2018) 042312.

\bibitem{low2017optimal}
G.~H. Low and I.~L. Chuang, ``Optimal hamiltonian simulation by quantum signal
  processing,'' {\em Physical review letters} {\bfseries 118} no.~1, (2017)
  010501.

\bibitem{wan2020fast}
K.~Wan and I.~Kim, ``Fast digital methods for adiabatic state preparation,''
  {\em arXiv preprint arXiv:2004.04164} (2020) .

\bibitem{Ooguri:2006in}
H.~Ooguri and C.~Vafa, ``{On the Geometry of the String Landscape and the
  Swampland},'' \href{http://dx.doi.org/10.1016/j.nuclphysb.2006.10.033}{{\em
  Nucl. Phys. B} {\bfseries 766} (2007) 21--33},
  \href{http://arxiv.org/abs/hep-th/0605264}{{\ttfamily arXiv:hep-th/0605264}}.

\bibitem{Bedroya:2019snp}
A.~Bedroya and C.~Vafa, ``{Trans-Planckian Censorship and the Swampland},''
  \href{http://dx.doi.org/10.1007/JHEP09(2020)123}{{\em JHEP} {\bfseries 09}
  (2020) 123}, \href{http://arxiv.org/abs/1909.11063}{{\ttfamily
  arXiv:1909.11063 [hep-th]}}.

\bibitem{Agrawal:2018own}
P.~Agrawal, G.~Obied, P.~J. Steinhardt, and C.~Vafa, ``{On the Cosmological
  Implications of the String Swampland},''
  \href{http://dx.doi.org/10.1016/j.physletb.2018.07.040}{{\em Phys. Lett. B}
  {\bfseries 784} (2018) 271--276},
  \href{http://arxiv.org/abs/1806.09718}{{\ttfamily arXiv:1806.09718
  [hep-th]}}.

\bibitem{hogarth1996predictability}
M.~Hogarth, {\em Predictability, computability, and spacetime}.
\newblock PhD thesis, University of Cambridge, 1996.

\bibitem{geroch1986computability}
R.~Geroch and J.~B. Hartle, ``Computability and physical theories,'' {\em
  Foundations of Physics} {\bfseries 16} no.~6, (1986) 533--550.

\bibitem{Bouland:2019pvu}
A.~Bouland, B.~Fefferman, and U.~Vazirani, ``{Computational pseudorandomness,
  the wormhole growth paradox, and constraints on the AdS/CFT duality},''
  \href{http://arxiv.org/abs/1910.14646}{{\ttfamily arXiv:1910.14646
  [quant-ph]}}.

\bibitem{Gheorghiu:2020sko}
A.~Gheorghiu and M.~J. Hoban, ``{Estimating the entropy of shallow circuit
  outputs is hard},'' \href{http://arxiv.org/abs/2002.12814}{{\ttfamily
  arXiv:2002.12814 [quant-ph]}}.

\bibitem{Susskind:2020kti}
L.~Susskind, ``{Horizons Protect Church-Turing},''
  \href{http://arxiv.org/abs/2003.01807}{{\ttfamily arXiv:2003.01807
  [hep-th]}}.

\bibitem{Kim:2020cds}
I.~Kim, E.~Tang, and J.~Preskill, ``{The ghost in the radiation: Robust
  encodings of the black hole interior},''
  \href{http://dx.doi.org/10.1007/JHEP06(2020)031}{{\em JHEP} {\bfseries 20}
  (2020) 031}, \href{http://arxiv.org/abs/2003.05451}{{\ttfamily
  arXiv:2003.05451 [hep-th]}}.

\bibitem{Yoshida:2020wpd}
B.~Yoshida, ``{Remarks on Black Hole Complexity Puzzle},''
  \href{http://arxiv.org/abs/2005.12491}{{\ttfamily arXiv:2005.12491
  [hep-th]}}.

\bibitem{campbell2019random}
E.~Campbell, ``Random compiler for fast hamiltonian simulation,'' {\em Physical
  review letters} {\bfseries 123} no.~7, (2019) 070503.

\bibitem{sieberer2019digital}
L.~M. Sieberer, T.~Olsacher, A.~Elben, M.~Heyl, P.~Hauke, F.~Haake, and
  P.~Zoller, ``Digital quantum simulation, trotter errors, and quantum chaos of
  the kicked top,'' {\em npj Quantum Information} {\bfseries 5} no.~1, (2019)
  1--11.

\bibitem{chen2020quantum}
C.-F. Chen, R.~Kueng, J.~A. Tropp, {\em et~al.}, ``Quantum simulation via
  randomized product formulas: Low gate complexity with accuracy guarantees,''
  {\em arXiv preprint arXiv:2008.11751} (2020) .

\bibitem{Ashwork}
V.~Hetherington and A.~Milsted, ``{private communications},''.

\bibitem{Danshita:2016xbo}
I.~Danshita, M.~Hanada, and M.~Tezuka, ``{Creating and probing the
  Sachdev-Ye-Kitaev model with ultracold gases: Towards experimental studies of
  quantum gravity},'' \href{http://dx.doi.org/10.1093/ptep/ptx108}{{\em PTEP}
  {\bfseries 2017} no.~8, (2017) 083I01},
  \href{http://arxiv.org/abs/1606.02454}{{\ttfamily arXiv:1606.02454
  [cond-mat.quant-gas]}}.

\bibitem{Landsman:2018jpm}
K.~Landsman, C.~Figgatt, T.~Schuster, N.~Linke, B.~Yoshida, N.~Yao, and
  C.~Monroe, ``{Verified Quantum Information Scrambling},''
  \href{http://dx.doi.org/10.1038/s41586-019-0952-6}{{\em Nature} {\bfseries
  567} no.~7746, (2019) 61--65},
  \href{http://arxiv.org/abs/1806.02807}{{\ttfamily arXiv:1806.02807
  [quant-ph]}}.

\bibitem{Brown:2019hmk}
A.~R. Brown, H.~Gharibyan, S.~Leichenauer, H.~W. Lin, S.~Nezami, G.~Salton,
  L.~Susskind, B.~Swingle, and M.~Walter, ``{Quantum Gravity in the Lab:
  Teleportation by Size and Traversable Wormholes},''
  \href{http://arxiv.org/abs/1911.06314}{{\ttfamily arXiv:1911.06314
  [quant-ph]}}.

\bibitem{Kobrin:2020xms}
B.~Kobrin, Z.~Yang, G.~D. Kahanamoku-Meyer, C.~T. Olund, J.~E. Moore,
  D.~Stanford, and N.~Y. Yao, ``{Many-Body Chaos in the Sachdev-Ye-Kitaev
  Model},'' \href{http://arxiv.org/abs/2002.05725}{{\ttfamily arXiv:2002.05725
  [hep-th]}}.

\bibitem{Garcia-Alvarez:2016wem}
L.~Garcia-Alvarez, I.~Egusquiza, L.~Lamata, A.~del Campo, J.~Sonner, and
  E.~Solano, ``{Digital Quantum Simulation of Minimal AdS/CFT},''
  \href{http://dx.doi.org/10.1103/PhysRevLett.119.040501}{{\em Phys. Rev.
  Lett.} {\bfseries 119} no.~4, (2017) 040501},
  \href{http://arxiv.org/abs/1607.08560}{{\ttfamily arXiv:1607.08560
  [quant-ph]}}.

\bibitem{Babbush:2018mlj}
R.~Babbush, D.~W. Berry, and H.~Neven, ``{Quantum Simulation of the
  Sachdev-Ye-Kitaev Model by Asymmetric Qubitization},''
  \href{http://dx.doi.org/10.1103/PhysRevA.99.040301}{{\em Phys. Rev. A}
  {\bfseries 99} no.~4, (2019) 040301},
  \href{http://arxiv.org/abs/1806.02793}{{\ttfamily arXiv:1806.02793
  [quant-ph]}}.

\bibitem{Xu:2020shn}
S.~Xu, L.~Susskind, Y.~Su, and B.~Swingle, ``{A Sparse Model of Quantum
  Holography},'' \href{http://arxiv.org/abs/2008.02303}{{\ttfamily
  arXiv:2008.02303 [cond-mat.str-el]}}.

\bibitem{Fedichev:2003bv}
P.~O. Fedichev and U.~R. Fischer, ``{'Cosmological' quasiparticle production in
  harmonically trapped superfluid gases},''
  \href{http://dx.doi.org/10.1103/PhysRevA.69.033602}{{\em Phys. Rev. A}
  {\bfseries 69} (2004) 033602},
  \href{http://arxiv.org/abs/cond-mat/0303063}{{\ttfamily
  arXiv:cond-mat/0303063}}.

\bibitem{Fischer:2004bf}
U.~R. Fischer and R.~Schutzhold, ``{Quantum simulation of cosmic inflation in
  two-component Bose-Einstein condensates},''
  \href{http://dx.doi.org/10.1103/PhysRevA.70.063615}{{\em Phys. Rev. A}
  {\bfseries 70} (2004) 063615},
  \href{http://arxiv.org/abs/cond-mat/0406470}{{\ttfamily
  arXiv:cond-mat/0406470}}.

\bibitem{Antonini:2019qkt}
S.~Antonini and B.~Swingle, ``{Cosmology at the end of the world},''
  \href{http://dx.doi.org/10.1038/s41567-020-0909-6}{{\em Nature Phys.}
  {\bfseries 16} no.~8, (2020) 881--886},
  \href{http://arxiv.org/abs/1907.06667}{{\ttfamily arXiv:1907.06667
  [hep-th]}}.

\bibitem{Yuan:2020xmq}
X.~Yuan, J.~Sun, J.~Liu, Q.~Zhao, and Y.~Zhou, ``{Quantum simulation with
  hybrid tensor networks},'' \href{http://arxiv.org/abs/2007.00958}{{\ttfamily
  arXiv:2007.00958 [quant-ph]}}.

\bibitem{Achucarro:2018vey}
A.~Ach\'ucarro and G.~A. Palma, ``{The string swampland constraints require
  multi-field inflation},''
  \href{http://dx.doi.org/10.1088/1475-7516/2019/02/041}{{\em JCAP} {\bfseries
  02} (2019) 041}, \href{http://arxiv.org/abs/1807.04390}{{\ttfamily
  arXiv:1807.04390 [hep-th]}}.

\bibitem{Tolley:2009fg}
A.~J. Tolley and M.~Wyman, ``{The Gelaton Scenario: Equilateral non-Gaussianity
  from multi-field dynamics},''
  \href{http://dx.doi.org/10.1103/PhysRevD.81.043502}{{\em Phys. Rev. D}
  {\bfseries 81} (2010) 043502},
  \href{http://arxiv.org/abs/0910.1853}{{\ttfamily arXiv:0910.1853 [hep-th]}}.

\bibitem{Achucarro:2010jv}
A.~Achucarro, J.-O. Gong, S.~Hardeman, G.~A. Palma, and S.~P. Patil, ``{Mass
  hierarchies and non-decoupling in multi-scalar field dynamics},''
  \href{http://dx.doi.org/10.1103/PhysRevD.84.043502}{{\em Phys. Rev. D}
  {\bfseries 84} (2011) 043502},
  \href{http://arxiv.org/abs/1005.3848}{{\ttfamily arXiv:1005.3848 [hep-th]}}.

\bibitem{Achucarro:2010da}
A.~Achucarro, J.-O. Gong, S.~Hardeman, G.~A. Palma, and S.~P. Patil,
  ``{Features of heavy physics in the CMB power spectrum},''
  \href{http://dx.doi.org/10.1088/1475-7516/2011/01/030}{{\em JCAP} {\bfseries
  01} (2011) 030}, \href{http://arxiv.org/abs/1010.3693}{{\ttfamily
  arXiv:1010.3693 [hep-ph]}}.

\bibitem{Wands:1998yp}
D.~Wands, ``{Duality invariance of cosmological perturbation spectra},''
  \href{http://dx.doi.org/10.1103/PhysRevD.60.023507}{{\em Phys. Rev. D}
  {\bfseries 60} (1999) 023507},
  \href{http://arxiv.org/abs/gr-qc/9809062}{{\ttfamily arXiv:gr-qc/9809062}}.

\bibitem{Finelli:2001sr}
F.~Finelli and R.~Brandenberger, ``{On the generation of a scale invariant
  spectrum of adiabatic fluctuations in cosmological models with a contracting
  phase},'' \href{http://dx.doi.org/10.1103/PhysRevD.65.103522}{{\em Phys. Rev.
  D} {\bfseries 65} (2002) 103522},
  \href{http://arxiv.org/abs/hep-th/0112249}{{\ttfamily arXiv:hep-th/0112249}}.

\bibitem{Khoury:2001wf}
J.~Khoury, B.~A. Ovrut, P.~J. Steinhardt, and N.~Turok, ``{The Ekpyrotic
  universe: Colliding branes and the origin of the hot big bang},''
  \href{http://dx.doi.org/10.1103/PhysRevD.64.123522}{{\em Phys. Rev. D}
  {\bfseries 64} (2001) 123522},
  \href{http://arxiv.org/abs/hep-th/0103239}{{\ttfamily arXiv:hep-th/0103239}}.

\bibitem{JPtoappear1}
J.~Liu, J.~Preskill, and B.~Sahinoglu, ``{Simulating kink scattering in a
  quantum computer},'' \href{http://arxiv.org/abs/to appear}{{\ttfamily to
  appear}}.
  \url{https://drive.google.com/file/d/1wOrPmO7lfutZlLEj-12JiMUAVg4nOdQr/view}.

\bibitem{JPtoappear2}
A.~Milsted, J.~Liu, J.~Preskill, and G.~Vidal, ``{Simulating kink scattering in
  tensor networks},'' \href{http://arxiv.org/abs/to appear}{{\ttfamily to
  appear}}. \url{https://www.youtube.com/watch?v=9Om--8LsqFw}.

\bibitem{Russo:2018akp}
J.~G. Russo and P.~K. Townsend, ``{Late-time Cosmic Acceleration from
  Compactification},'' \href{http://dx.doi.org/10.1088/1361-6382/ab0804}{{\em
  Class. Quant. Grav.} {\bfseries 36} no.~9, (2019) 095008},
  \href{http://arxiv.org/abs/1811.03660}{{\ttfamily arXiv:1811.03660
  [hep-th]}}.

\bibitem{Russo:2019fnk}
J.~G. Russo and P.~K. Townsend, ``{Time-dependent compactification to de Sitter
  space: a no-go theorem},''
  \href{http://dx.doi.org/10.1007/JHEP06(2019)097}{{\em JHEP} {\bfseries 06}
  (2019) 097}, \href{http://arxiv.org/abs/1904.11967}{{\ttfamily
  arXiv:1904.11967 [hep-th]}}.

\bibitem{Obied:2018sgi}
G.~Obied, H.~Ooguri, L.~Spodyneiko, and C.~Vafa, ``{De Sitter Space and the
  Swampland},'' \href{http://arxiv.org/abs/1806.08362}{{\ttfamily
  arXiv:1806.08362 [hep-th]}}.

\bibitem{Witten:1998qj}
E.~Witten, ``{Anti-de Sitter space and holography},''
  \href{http://dx.doi.org/10.4310/ATMP.1998.v2.n2.a2}{{\em Adv. Theor. Math.
  Phys.} {\bfseries 2} (1998) 253--291},
  \href{http://arxiv.org/abs/hep-th/9802150}{{\ttfamily arXiv:hep-th/9802150}}.

\bibitem{Ooguri:2018wrx}
H.~Ooguri, E.~Palti, G.~Shiu, and C.~Vafa, ``{Distance and de Sitter
  Conjectures on the Swampland},''
  \href{http://dx.doi.org/10.1016/j.physletb.2018.11.018}{{\em Phys. Lett. B}
  {\bfseries 788} (2019) 180--184},
  \href{http://arxiv.org/abs/1810.05506}{{\ttfamily arXiv:1810.05506
  [hep-th]}}.

\bibitem{Kachru:2018aqn}
S.~Kachru and S.~P. Trivedi, ``{A comment on effective field theories of flux
  vacua},'' \href{http://dx.doi.org/10.1002/prop.201800086}{{\em Fortsch.
  Phys.} {\bfseries 67} no.~1-2, (2019) 1800086},
  \href{http://arxiv.org/abs/1808.08971}{{\ttfamily arXiv:1808.08971
  [hep-th]}}.

\bibitem{Aaronson:2005qu}
S.~Aaronson, ``{NP-complete problems and physical reality},''
  \href{http://arxiv.org/abs/quant-ph/0502072}{{\ttfamily
  arXiv:quant-ph/0502072}}.

\bibitem{Bekenstein:1980jp}
J.~D. Bekenstein, ``{A Universal Upper Bound on the Entropy to Energy Ratio for
  Bounded Systems},'' \href{http://dx.doi.org/10.1103/PhysRevD.23.287}{{\em
  Phys. Rev. D} {\bfseries 23} (1981) 287}.

\bibitem{Welch:2006va}
P.~Welch, ``{The extent of computation in Malament-Hogarth spacetimes},''
  \href{http://arxiv.org/abs/gr-qc/0609035}{{\ttfamily arXiv:gr-qc/0609035}}.

\bibitem{Hayden:2007cs}
P.~Hayden and J.~Preskill, ``{Black holes as mirrors: Quantum information in
  random subsystems},''
  \href{http://dx.doi.org/10.1088/1126-6708/2007/09/120}{{\em JHEP} {\bfseries
  09} (2007) 120}, \href{http://arxiv.org/abs/0708.4025}{{\ttfamily
  arXiv:0708.4025 [hep-th]}}.

\bibitem{halpern2017toward}
N.~Y. Halpern, ``Toward physical realizations of thermodynamic resource
  theories,'' in {\em Information and Interaction}, pp.~135--166.
\newblock Springer, 2017.

\bibitem{Felder:2000hq}
G.~N. Felder and I.~Tkachev, ``{LATTICEEASY: A Program for lattice simulations
  of scalar fields in an expanding universe},''
  \href{http://dx.doi.org/10.1016/j.cpc.2008.02.009}{{\em Comput. Phys.
  Commun.} {\bfseries 178} (2008) 929--932},
  \href{http://arxiv.org/abs/hep-ph/0011159}{{\ttfamily arXiv:hep-ph/0011159}}.

\bibitem{bravyi2002fermionic}
S.~B. Bravyi and A.~Y. Kitaev, ``Fermionic quantum computation,'' {\em Annals
  of Physics} {\bfseries 298} no.~1, (2002) 210--226.

\bibitem{Chen:2017fvr}
Y.-A. Chen, A.~Kapustin, and D.~Radi\v{c}evi\'c, ``{Exact bosonization in two
  spatial dimensions and a new class of lattice gauge theories},''
  \href{http://dx.doi.org/10.1016/j.aop.2018.03.024}{{\em Annals Phys.}
  {\bfseries 393} (2018) 234--253},
  \href{http://arxiv.org/abs/1711.00515}{{\ttfamily arXiv:1711.00515
  [cond-mat.str-el]}}.

\bibitem{Chen:2018nog}
Y.-A. Chen and A.~Kapustin, ``{Bosonization in three spatial dimensions and a
  2-form gauge theory},''
  \href{http://dx.doi.org/10.1103/PhysRevB.100.245127}{{\em Phys. Rev. B}
  {\bfseries 100} no.~24, (2019) 245127},
  \href{http://arxiv.org/abs/1807.07081}{{\ttfamily arXiv:1807.07081
  [cond-mat.str-el]}}.

\bibitem{Chen:2019wlx}
Y.-A. Chen, ``{Exact bosonization in arbitrary dimensions},''
  \href{http://arxiv.org/abs/1911.00017}{{\ttfamily arXiv:1911.00017
  [cond-mat.str-el]}}.

\bibitem{Chatrchyan:2020syc}
A.~Chatrchyan, K.~T. Geier, M.~K. Oberthaler, J.~Berges, and P.~Hauke,
  ``{Analog reheating of the early universe in the laboratory},''
  \href{http://arxiv.org/abs/2008.02290}{{\ttfamily arXiv:2008.02290
  [cond-mat.quant-gas]}}.

\end{thebibliography}\endgroup
\end{document}